\documentclass{article}
\usepackage[left=1in,top=1.25in,right=1in,bottom=1.25in,head=1.25in]{geometry}
\usepackage{amsfonts,amsmath,amssymb,amsthm}
\usepackage{verbatim,float,url,enumerate}
\usepackage{graphicx,psfrag}
\usepackage{dsfont,bm,color}
 \usepackage{natbib}

\usepackage{setspace}
\usepackage{stmaryrd}

\usepackage{algorithm}

\newtheorem{proposition}{Proposition}

\def\bA{{\bf A}}
\def\bX {{\bf X}}

\newcommand{\Nu}{{\cal V}}
\usepackage[T1]{fontenc}

\def\TZ {{\rm TZ}}

\usepackage{authblk}

\title{More powerful post-selection inference,\\
with application to the Lasso}
\author[1]{Keli Liu\footnote{keliliu@stanford.edu}}
\author[1]{Jelena Markovic\footnote{jelenam@stanford.edu}}
\author[1,2]{Robert Tibshirani\footnote{tibs@stanford.edu}}
\affil[1]{Department of Statistics}
\affil[2]{Department of Biomedical Data Science}
\affil[ ]{Stanford University}
\date{}

\begin{document}

{
\setstretch{1}	

\maketitle

\begin{abstract} Investigators often use the data to generate interesting hypotheses and then perform inference for the generated hypotheses. P-values and confidence intervals must account for this explorative data analysis. A fruitful method for doing so is to condition any inferences on the components of the data used to generate the hypotheses, thus preventing information in those components from being used again. Some currently popular methods ``over-condition,'' leading to wide intervals. We show how to perform the minimal conditioning in a computationally tractable way. In high dimensions, even this minimal conditioning can lead to  intervals  that are too wide to be useful,  suggesting that up to now the cost of hypothesis generation has been underestimated. We show how to generate hypotheses in a strategic manner that sharply reduces the cost of data exploration and results in useful confidence intervals. Our discussion focuses on the problem of post-selection inference after fitting a lasso regression model, but  we also outline its extension to a much more general setting.
\end{abstract}
}

\section{Introduction}

In many applications we want to study how the mean of a response variable
is affected by various predictors. In idealized form, suppose we have the response vector $y=(y_1,\ldots,y_n)\in\mathbb{R}^n$ following
\[
y_{i}=\mu_{i}+\varepsilon_{i}\qquad\varepsilon_{i}\sim \mathcal{N}\left(0,\sigma^{2}\right)\qquad i=1,\ldots,n.
\]
The goal is to study the relation between the mean function $\mu$
and a set of $p$ explanatory variables $x_{1},\ldots,x_{p}$. The
standard approach to this problem approximates $\mu=(\mu_1,\ldots,\mu_n)\in\mathbb{R}^n$ using a linear
function of the explanatory variables: 
\[
\mu_{\text{lin}}\left(x\right)=x^{\top}\beta \qquad\beta=\left(\mathbf{X}^{\top}\mathbf{X}\right)^{-1}\mathbf{X}^{\top}\mu,
\]
where $x\in\mathbb{R}^p$ and $\mathbf{X}\in\mathbb{R}^{n\times p}$ is a matrix of features. An unbiased estimate of $\beta$ is given by the least squares estimator
of $y$ on the matrix $\mathbf{X}$; in addition, standard theory allows us to compute p-values and confidence intervals for $\beta$.

The standard approach was developed for the setting where the investigator, drawing on prior knowledge, has already focused in on a small set of explanatory variables. It is wholly inappropriate for tasks where the majority of variables are noise and exploratory analyses must be done to arrive at the ``variables of interest.'' In the latter case,
a popular approach for exploration is the lasso of \cite{tibshirani1996regression} which solves the $\ell_{1}$-penalized regression 
\begin{equation}
 \underset{\beta_0,\beta_1,\ldots,\beta_p}{{\rm minimize}}\;\frac{1}{2}\sum_{i=1}^n\Bigl(y_{i}-\beta_{0}-\sum_{j=1}^{p}x_{ij}\beta_{j}\Bigr)^2+\lambda\sum_{j=1}^{p}|\beta_{j}|. \label{eqn:lasso}
\end{equation}
For fixed choice of the tuning parameter $\lambda$, the lasso selects a set of interesting variables (the ``active set'') while setting $\hat{\beta}_{j}=0$ for uninteresting variables. How can we compute valid p-values and confidence intervals while accounting for exploration using the lasso? That is the problem that we address here. One might think of taking the active set, and computing the standard/``na\"ive'' p-values and confidence intervals based on least squares theory. But the standard theory assumes pre-specified hypotheses; without accounting for the fact that our hypotheses are generated from the data, it is quite common for a na\"ive 90\% confidence interval to cover less than 50\% of the time.

Performing inference after selection is notoriously difficult, e.g., \citet{leeb2005model, leeb2006} show the impossibility of estimating the conditional distribution of a test statistic used for inference, implying failure of the standard bootstrap for this purpose. On the other hand, \citet{lee2016} and \citet{TTLT2016} propose a conditional approach that leads to a truncated normal reference
distribution, that gives provably exact type I error and coverage
in finite samples. A drawback of this quite remarkable construction
is the fact that the confidence intervals are often very wide. The intervals ``over-condition,'' that is they overstate the cost of data exploration using the lasso (by quite a bit it turns out!). 

Consider data from \citet{stamey1989prostate} who study the relation between prostate specific antigen (PSA) and various clinical measures: log cancer volume (\texttt{lcavol}), log prostate weight (\texttt{lweight}), age, log of benign prostatic hyperplasia amount (\texttt{lbph}), seminal vesicle invasion (\texttt{svi}), log of capsular penetration (\texttt{lcp}), the Gleason score (\texttt{gleason}), and percent of Gleason scores 4 or 5 (\texttt{pgg45}). The data were collected from $n=97$ men
about to receive a radial prostatectomy. Running the lasso on the
prostate data (with $\lambda$ chosen by 10 fold cross-validation)
yields seven active variables (the log of capsular penetration was
excluded from the model). Figure \ref{fig:prostate} plots 90\% confidence intervals for the coefficients of the seven selected variables. The coefficients can be defined either with respect to a linear model of all the variables (``full'' coefficients) or a linear model of only the selected variables (``partial'' coefficients). We summarize  the findings from the figure as follows:

\begin{itemize}
\item The ``naive'' intervals are based on standard normal theory and ignore selection. They are too short and hence will undercover their targets.
\item  The various ``TZ'' methods account for selection:
operationally, this implies that the least squares estimates
for the selected variables have a \emph{truncated }normal distribution (hence TZ).  TZ$_{M}$ and TZ$_{Ms}$  are the approaches introduced by \citet{lee2016} and \citet{TTLT2016}. They have exact finite-sample coverage but often produce very long intervals.  Of the three variables (\texttt{lcavol}, \texttt{svi} and \texttt{lweight}) deemed highly significant by a least squares analysis, two now have 90\% confidence intervals containing 0, and the third has a very wide and right skewed interval.
\item The  truncated Z-statistic methods TZ$_{V}$ and TZ$_{\text{stab-}t}$ --- introduced in this paper --- have exact finite-sample coverage and give much shorter intervals compared to TZ$_{M}$ and TZ$_{Ms}$.
\end{itemize}

\begin{figure}[th]
	\begin{centering}
		\includegraphics[width=0.5\textwidth]{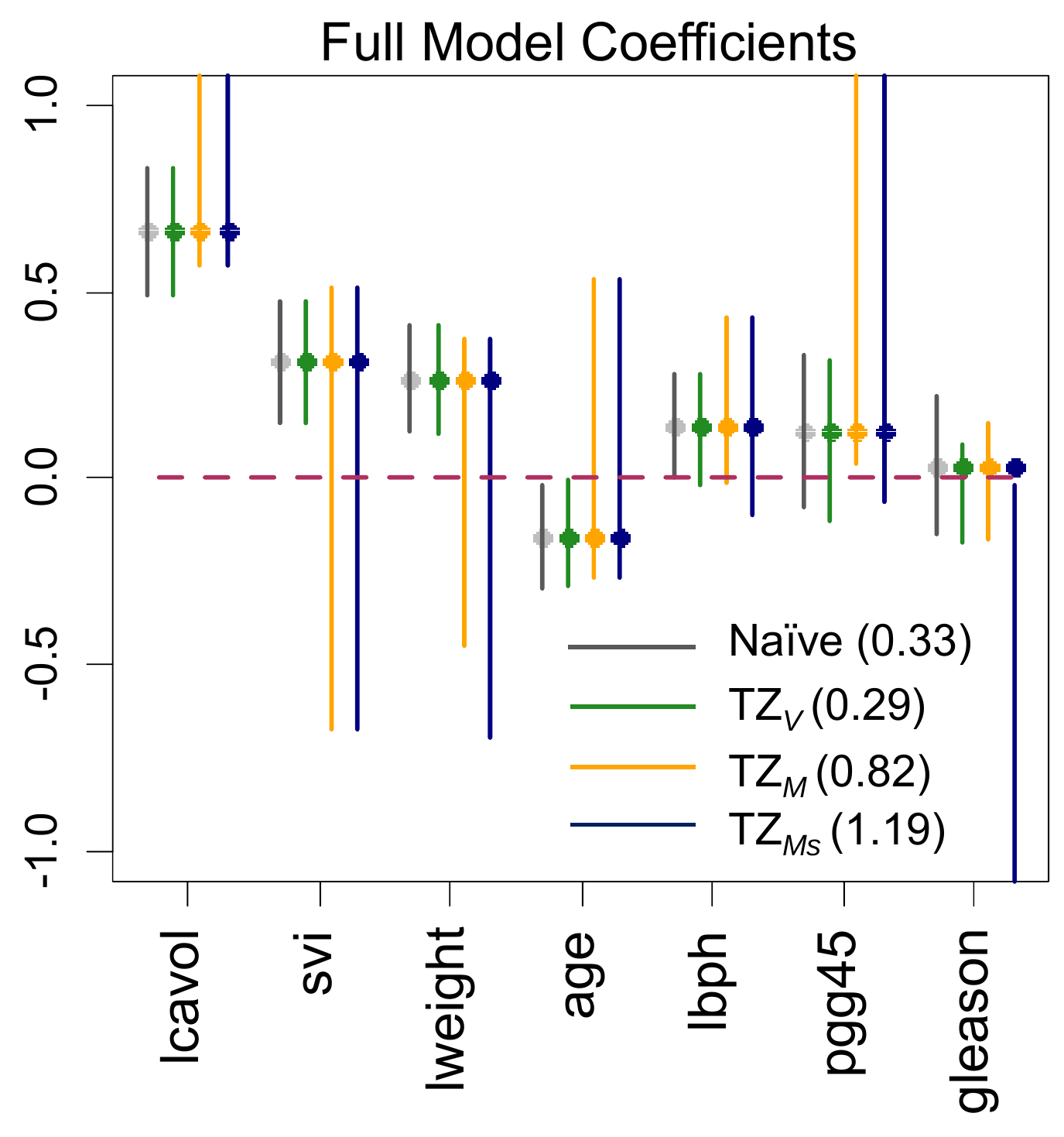}\includegraphics[width=0.5\textwidth]{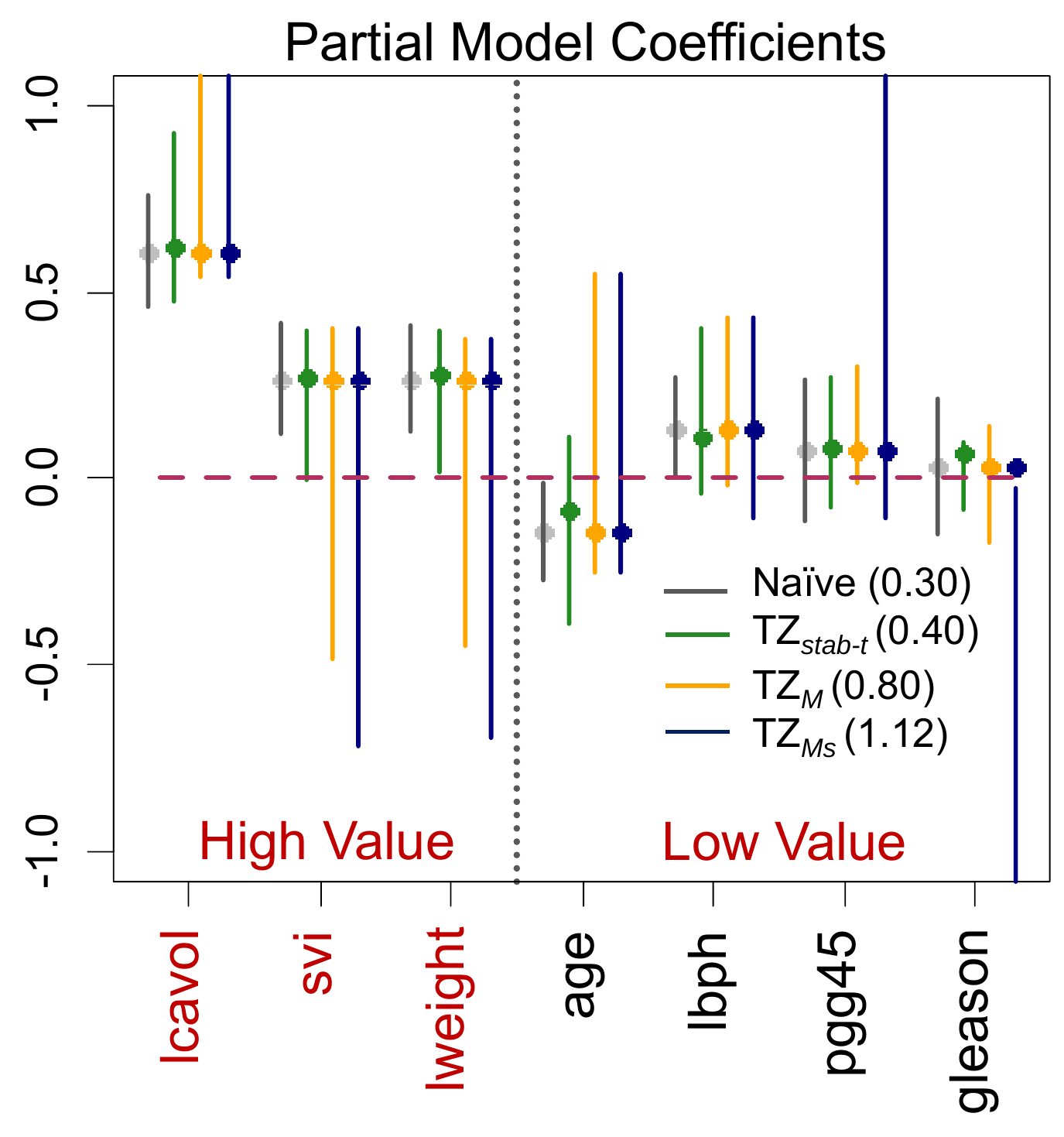}
		\par\end{centering}
	\caption{\em Regression analysis of log prostate specific antigen on $p=$8 clinical
		factors using data from $n=$97 men undergoing radial prostatectomy.
		A lasso was run with $\lambda=3.14$ (chosen by 10 fold cross-validation),
		selecting 7 variables (horizontal axis). 90\% confidence intervals
		were constructed for the regression coefficient of each selected variable:
		``coefficient'' is defined with respect to either the linear model
		of response on all 8 variables (``full'') or only the 7 selected
		variables (``partial''). Four interval construction methods were
		used, in all cases the noise standard deviation $\sigma$ was fixed
		at the value (0.70) estimated from a least squares fit on all variables.
		The naive/standard Z-intervals ignore selection while the TZ-intervals
		adjust for selection but to differing degrees. The TZ$_{\text{stab}-t}$
		method partitions the 7 active variables into high value and low value
		targets, depending on whether each variable's Z-statistc in the partial
		regression exceeds $\Phi^{-1}\left(1-\frac{0.1}{2p}\right)=2.49$. The depicted point estimates (circles) are the MLEs for the corresponding population quantities (note that the population targets for $\TZ_{\text{stab}-t}$ differ slightly from the usual ``partial'' regression coefficients; see Section \ref{subsec:reg_target}).}
	\label{fig:prostate}
\end{figure}

Where the various TZ methods differ is in the amount of conditioning.  This distinction is crucial; in a nutshell:
\bigskip

\noindent
\textsc{\textbf{More Conditioning $\shortrightarrow$ less information for inference  $\shortrightarrow$ wider confidence intervals}}

\bigskip
\noindent See \citet{fithian2014} for justification of this maxim. Heuristically, not all of the variation in
the Z-statistic can be used for inference because some of that variation has already been used for exploratory data analysis (in this case to screen out unimportant variables). As a result,  the reference distribution for the Z-statistic changes from a normal to a truncated normal. TZ$_{Ms}$ sets aside information contained in the lasso active set ($\hat{M}$) as well as information in the signs of the selected variables ($\hat{s}$);
TZ$_{M}$ sets aside only the former leading to less truncation (the
subscripts indicate what information has been excluded for inference
purposes). TZ$_{V}$ and TZ$_{\text{stab-}t}$ set aside even less
information. But what is the minimum level of truncation needed for
the amount of explorative data analysis we have done? And how should
we change our data analysis strategy if this minimal truncation
still turns out to be too severe? This paper answers these two questions as follows: 
\begin{enumerate}
	\item To arrive at the minimal truncation, we first make clear in Section \ref{subsec:two_costs} that explorative data analysis usually incurs two separate ``search costs'': the choice of variables on which
	to do inference (e.g., the active set from the lasso), and the choice of targets for inference (\emph{Which} model should we define our regression coefficients with respect to?). This has not been recognized previously, as the two sources are usually lumped together conceptually. But it has significant practical implications. 
	\item In Section \ref{sec:full}, we consider the case where the targets for inference for the selected variables are the {\em full population regression coefficients}. There is no adaptivity in the choice of target. We show how to construct intervals that perform the minimal truncation necessary to adjust for the lasso selection event; this is
	the TZ$_{V}$ method. 
	\item In Section \ref{sec:partial}, we consider the case where the targets for inference for the selected variables are the {\em population partial regression coefficients}. There is significant adaptivity in the choice of targets\textemdash the definition of ``coefficient'' changes depending on which model is chosen. TZ$_{M}$ is shown to
	give the minimal necessary truncation, but its intervals are impractically wide. In such instances, blindly using the lasso to form hypotheses is a bad idea! We provide a framework for \emph{judiciously} generating hypotheses using the lasso. TZ$_{\text{stab}-t}$ gives the truncation appropriate for this judicious approach and produces intervals of reasonable length.
\end{enumerate}
Section \ref{sec:sim_partial} contains simulation results and Section \ref{sec:realdata} illustrates the benefits of our
proposals on a larger real-world dataset ($n=1057$ and $p=210$).
While we use the example of the lasso to simplify exposition, the
tools we develop can be adapted to treat other polyhedral selection
rules, e.g., forward stepwise regression. We give some details of
the general idea in Section \ref{sec:general}.

\subsection{Review and Critique of Post-Selection Inference via the Truncated Gaussian} \label{sec:geninference} 

We briefly review the work of \citet{lee2016} and \citet{TTLT2016}. Let $\bX$ be our $n\times p$ predictor matrix and $y$ our $n$-dimensional outcome vector. To focus in on interesting variables, suppose we have run the lasso with a fixed $\lambda$ and let $\hat{M}$ be the active set the lasso returns. We can now use $\hat{M}$ to generate interesting questions. For example, we might be interested in $\beta^{\left(\hat{M}\right)}=\left(\mathbf{X}_{\hat{M}}^{\top}\mathbf{X}_{\hat{M}}\right)^{-1}\mathbf{X}_{\hat{M}}^{\top}\mu$, the coefficients in the projection of $\mu$ onto the selected variables. Hypothesis tests of the form $H_{0}:\beta_{j}^{\left(\hat{M}\right)}=0$ have the general linear form $H_{0}:\eta^{\top}\mu=0$ for some vector $\eta\in\mathbb{R}^n$. Standard theory
tells us what to do when $\eta$ is pre-specified: the Z-statistic
\begin{equation}
	Z=\frac{\eta^{\top}y-\eta^{\top}\mu}{\sigma\lVert\eta\rVert_{2}}\sim \mathcal{N}\left(0,1\right) \label{eq:standard_pivot}
\end{equation}
is a pivot from which we can derive p-values and confidence intervals. However, this theory does not hold because our $\eta$ depends on the data through $\hat{M}$, i.e., $\eta=\eta(\hat{M})$ and upon substituting this into the above expression, the pivotal relationship no longer holds.

When performing inference for such data generated hypotheses, we should
remove whatever information has already been used up during our initial
exploration in order to avoid ``using the data twice.'' This is
achieved by conditioning on $\hat{M}$ (note that conditional on $\hat{M}=M$,
$\eta(\hat{M})$ is back to being a fixed vector). Post-selection
p-values and intervals can then be built in the usual manner if we can identify quantities that are pivotal \emph{conditional on} $\hat{M}$. \citet{lee2016} showed how to build such a conditional pivot. We summarize their construction: 
\begin{enumerate}
	\item The event $\left\{ \hat{M}=M\right\} $ is difficult to characterize
	but \citet{lee2016} and \citet{TTLT2016} reveal a simple characterization
	of a more refined event 
	\[
	\left\{\hat{M}=M,\hat{s}=s\right\}=\left\{\bA y\leq b\right\},
	\]
	where $\hat{s}$ is the vector of signs of the lasso coefficients
	of the active variables. The RHS above is a polyhedral region in outcome space (the exact form for $\bA$ and $b$ can be derived from the KKT conditions at the solution). \citet{lee2016} and \citet{TTLT2016} construct a quantity that is pivotal conditional on $\hat{M}$ \emph{and} $\hat{s}$ (which implies that it is pivotal conditional on $\hat{M}$). 
	
	\item The standard pivot (\ref{eq:standard_pivot}) simply centers and scales $\eta^{\top}y$ by its unconditional mean and standard deviation.
	The conditional distribution of $\eta^{\top}y$ given \textbf{$\mathbf{A}y\leq b$} is difficult to work with because it depends on nuisance parameters. The usual way to get rid of nuisance parameters is by conditioning; in this case, we condition on $P_{\eta^{\perp}}y$, the piece of $y$ orthogonal to $\eta$. \citet{lee2016} show that 
	\begin{equation}
		\eta^{\top}y \:\big|\:\mathbf{A}y\leq b, P_{\eta^{\top}}y\sim \mathcal{TN}\left(\eta^{\top}\mu,\sigma^{2}\lVert\eta\rVert_{2}^{2},\left[\mathcal{V}^{-},\mathcal{V}^{+}\right]\right)\label{eqn:poly}
	\end{equation}
	where $\mathcal{TN}\left(\mu,\sigma^{2},\left[c,d\right]\right)$ denotes the
	$\mathcal{N}\left(\mu,\sigma^{2}\right)$ distribution truncated to $\left[c,d\right]$. Above, $\mathcal{V}^{-}$ and $\mathcal{V}^{+}$ are computable functions of $\mathbf{A}$, $b$ and $P_{\eta^{\top}}y$ (see Appendix \ref{app:truncation:points}). This implies that the usual $Z$-statistic (\ref{eq:standard_pivot}) has a truncated normal distribution as its conditional!
	\item The CDF of a truncated Gaussian distribution\index{truncated normal distribution} with support confined to $[c,d]$ is 
	\begin{eqnarray}
		F_{\mu,\sigma^{2}}^{c,d}(x)=\frac{\Phi((x-\mu)/\sigma)-\Phi((c-\mu)/\sigma)}{\Phi((d-\mu)/\sigma)-\Phi((c-\mu)/\sigma)},\label{eq:tg}
	\end{eqnarray}
	with $\Phi$ the CDF\index{cumulative distribution function} of the standard Gaussian. The CDF of a random variable, evaluated at the value of that random variable, has a uniform distribution on $(0,1)$. Hence 
	\begin{eqnarray}
		F_{\eta^{\top}\mu,\sigma^{2}\lVert\eta\rVert_{2}^{2}}^{\mathcal{V}^{-},\mathcal{V}^{+}}\left(\eta^{\top}y\right)\;\Big|\;\left\{\hat{M}=M,\hat{s}=s\right\}\sim\mathcal{U}(0,1).\label{eq:tn0}
	\end{eqnarray}
\end{enumerate}
We obtain confidence intervals by inverting the pivotal quantity (\ref{eq:tn0}): we denote these as the TZ$_{Ms}$ intervals, for truncated Z-statistic intervals conditional on the selected model variables $M$ and their signs $s$. For a $1-\alpha$ interval, we find the largest and the smallest
$\eta^{\top}\mu$ such that the value of pivotal quantity remains
in the interval $\left[\frac{\alpha}{2},1-\frac{\alpha}{2}\right]$.
\bigskip

The technique of \citet{lee2016} and \citet{TTLT2016} allows for
exact inferences after using the lasso (or any other procedure whose
selection event is polyhedral) to generate hypotheses. However, it
is not without shortcomings: 
\begin{itemize}
	\item The TZ$_{Ms}$ intervals condition on $\left(\hat{M},\hat{s}\right)$
	rather than just $\hat{M}$. This would be appropriate if we also
	used information in the signs when forming our hypotheses. For example,
	we might test against a one sided alternative, $H_{1}:\beta_{j}^{\left(\hat{M}\right)}>0$,
	if the sign of the $j$th variable is positive. However, in most applications,
	the signs are not used in hypothesis generation, hence we have thrown
	away unused information by conditioning on $\hat{s}$. As a result,
	the intervals tend to be quite long (Figure \ref{fig:prostate}).
	\item \citet{lee2016} makes mention of the previous point and shows how
	to condition only on $\hat{M}$ by an enumeration of all possible
	sign vectors. Their enumeration method is intractable when $\left|\hat{M}\right|$
	is large as there are $2^{\left|\hat{M}\right|}$ possible sign vectors.
	In Section \ref{sec:partial}, we provide a workable implementation. 
	\item In some applications, $\hat{M}$ is used to generate hypotheses in a very restrictive manner. The restrictions allow us to condition on less than $\hat{M}$, thus leading to more precise intervals; see Section \ref{sec:full}. 
\end{itemize}
Appendix \ref{app:simple:example:conditioning:events} presents a simple example that shows the effects of different conditioning strategies.

\medskip

The  above criticisms have focused on the ``over-conditioning'' of
the TZ$_{Ms}$ intervals. There is however an entirely different (and
practically pressing) concern: What if the intervals are uselessly
wide even when we adopt the ``minimal'' conditioning? It turns out
that when we are interested in the partial regression coefficients,
conditioning on $\hat{M}$ is the ``minimal'' conditioning; yet
as Figure \ref{fig:prostate} shows, the resulting TZ$_{M}$ intervals
are still unacceptably wide. Such imprecision is in fact the normal
state of affairs when $p>n$ (see Section \ref{sec:sim_partial}).
What such an occurrence should tell us is that we have used up too much
information in our initial data exploration\textemdash and did not
leave enough information for inference. 

Thus we need to strike a more
balanced approach to how we use the data to generate hypotheses. \citet{TT2015}
offer one such approach: perform our exploratory analysis not on the
actual data but a noise-perturbed version of it (similarly, one could
split the data). That is, we would obtain our active set not by running
lasso on the original data but by running lasso on the original data
with noise added on. The resulting intervals are substantially shorter
than the TZ$_{Ms}$ intervals but do require greater computation
in the form of MCMC sampling. More concerning, however, is that the
variables we select will vary across realizations of noise and will
differ from the variables we select had we used the original data.
This paper takes a different route by giving a deterministic method
for controlling the cost of data exploration: we will continue to
generate hypotheses using $\hat{M}$, the lasso active set from the
original data, but we will restrict how $\hat{M}$ can be used (Section
\ref{sec:partial}). The key to our development is distinguishing
between two often confounded ways in which the data can be used to
shape the direction of our investigation.

\subsection{Forming a Data Driven Query: Two Costs}\label{subsec:two_costs} 

In many applications, investigators come
to the data without pre-specified questions. Instead, they want to use the data to generate interesting hypotheses. Often, the hypothesis generation serves to answer two questions. Out of a possibly larger set of explanatory variables, which ones should we focus on? What aspect of the interesting variables should we study? The two processes that answer these questions are: 
\begin{itemize}
	\item \emph{Variable Selection}: The data is used to decide which variables are worthy of attention, e.g., running the lasso and focusing on the active set. 
	\item \emph{Target Formation}: Having settled on a subset $M\subset\left\{ 1,\ldots,p\right\} $
	of variables for careful study, we need to decide how to summarize the relation between the selected variables and $\mu$. That is, what should be the target of our estimation? We could try and estimate
	the regression coefficients from projecting $\mu$ onto $\mathbf{X}_{M}$:
	\[
	\beta^{\left(M\right)}=\left(\mathbf{X}_{M}^{\top}\mathbf{X}_{M}\right)^{-1}\mathbf{X}_{M}^{\top}\mu,
	\]
	summarizing the effect of variable $j$ through $\beta_{j}^{\left(M\right)}$. In this approach, the way of capturing the effect of variable $j$ depends on what other variables are in $M$. We call this the \textbf{partial target}, since they are the (population) partial regression coefficients.
	
	In contrast, prior to looking at the data, we could decide to \emph{always}
	summarize the effect of $j$ via its coefficient $\beta_{j}^{F}$
	in a regression of $\mu$ on all the variables 
	\[
	\beta^{F}=\left(\mathbf{X}^{\top}\mathbf{X}\right)^{-1}\mathbf{X}^{\top}\mu.
	\]
	In this case, we have a fixed (non data-driven) target, $\beta^{F}$.
	We call this the \textbf{full target}, since they are the (population)
	regression coefficients in the full model. While we may use variable
	selection to decide which components of $\beta^{F}$ to estimate,
	we do not use the data to inform how we should summarize the relation
	between selected variables and $\mu$. 
\end{itemize}
Both variable selection and target formation carry a cost that needs
to be acknowledged when doing inference. When we condition on the
lasso active set $\hat{M}$ for post-selection inference, our inferences
are valid \emph{regardless} of how $\hat{M}$ is used in variable
selection and target formation. However, if we know the precise way
in which $\hat{M}$ is used for these two processes, we can often
do (much) better than simply conditioning on $\hat{M}$. For example,
in Section \ref{subsec:reg_target}, we consider a procedure where
only the ``best variables'' in $\hat{M}$ are used in target formation
rather than all of $\hat{M}$\textemdash hence we only need to condition
on the identity of the ``best variables.'' For these reasons, distinguishing
between the costs of variable selection and target formation is key
to improving over the standard truncated-Z inferences.

\section{Inference for the full regression targets} \label{sec:full}

In this section, we are concerned with the (easier) case where
the data has been used to choose a set of interesting variables but
the data has \textit{not} been used to decide how to summarize the
relation between the response and the selected variables. In particular,
we assume that the investigator has decided a priori that the relation
between a variable $x_{j}$ and the response variable should be captured
by its coefficient in a regression of $\mu$ on all the variables:

\[
\beta^{F}=\left(\mathbf{X}^{\top}\mathbf{X}\right)^{-1}\mathbf{X}^{\top}\mu.
\]
The data is only used to choose the components of $\beta^{F}$ we
are interested in. This is the {\em full target} discussed above.

We focus on the case where the lasso is used to choose the set $\hat{M}$
of variables. The relation between $\mu$ and each variable $j\in\hat{M}$
is summarized by $\beta_{j}^{F}$ (again, we emphasize that this choice
of summary was agreed upon before exploring the data). Now the goal
is to build a confidence interval $C_{j}$ for $\beta_{j}^{F}$. We
note that an interval is built for $\beta_{j}^{F}$ whenever $j\in\hat{M}\subset\left\{ 1,\ldots,p\right\} $.
Hence, it seems natural to require that our interval have good coverage
properties conditional on the interval existing in the first place:
\begin{eqnarray} \label{eq:cover_full}
	\mathbb{P}\left(\beta_{j}^{F}\notin C_{j}\:\big|\:j\in\hat{M}\right)\leq\alpha 
\end{eqnarray}
for a given level $\alpha$. In particular, we do \emph{not }condition on the entire active set:
when hypothesis generation involves only variable selection, it suffices
to condition on variable $j$ being in the active set (when performing
inference for $j$). In contrast, \citet{lee2016} conditions on both
the selected model $\hat{M}$ as well as the signs $\hat{s}$. By
conditioning on less, more variation remains in the data for determining
the precise value of $\beta_{j}^{F}$. Since the goal of conditioning
is to restrict the sample space to those data where confidence intervals
for $\beta_{j}^{F}$ will actually be constructed, the event $\left\{ j\in\hat{M}\right\} $
is precisely the minimal conditioning event.

To construct intervals conditional on $\left\{ j\in\hat{M}\right\} $,
we provide a particularly simple characterization of this set. For
motivation, consider the case where $\mathbf{X}$ has orthonormal columns in
which case $\beta_{j}^{F}=\eta_{j}^{\top}\mu$ where $\eta_{j}=x_{j}$.
Then the event $\left\{ j\in\hat{M}\right\} $ reduces to a union
of two disjoint intervals
\[
j\in\hat{M}\iff\left|\eta_{j}^{\top}y\right|>\lambda.
\]
That is, as long as the inner product between $x_{j}$ and $y$ exceeds
a cutoff, lasso will select variable $j$. Since $\eta_{j}^{\top}y\sim \mathcal{N}\left(\beta_{j}^{F},\sigma^{2}\right)$,
the conditional distribution of $\eta_{j}^{\top}y\:|\: j\in\hat{M}$ is
simply this normal distribution truncated to $\left[-\infty,-\lambda\right]\cup\left[\lambda,\infty\right]$.
A pivot can be obtained from this truncated normal distribution
and then inverted to obtain confidence intervals. 

In the general case, we no longer have $\eta_{j}=x_{j}$. Instead,
\[
\beta_{j}^{F}=\eta_{j}^{\top}\mu\qquad\eta_{j}=\mathbf{X}\left(\mathbf{X}^{\top}\mathbf{X}\right)^{-1}e_{j},
\]
where $e_j\in\mathbb{R}^p$ is a vector of all zeros with one at its $j$th coordinate.
It turns out that this general case can be reduced to the special
case of orthonormal predictors by conditioning. In particular, let
us decompose variation in $y$ into two components: variation in the
direction of $\eta_{j}$ and variation orthogonal to $\eta_{j}$.
\[
y=z_{j}\cdot c_{j}+\nu_{j}\qquad z_{j}=\eta_{j}^{\top}y\qquad c_{j}=\frac{\eta_{j}}{\lVert\eta_{j}\rVert_{2}^{2}}.
\]
We will condition on the component orthogonal to $\eta_{j}$, i.e.,
$\nu_{j}$, and show that the law of $\eta_{j}^{\top}y\:|\:j\in\hat{M},\nu_{j}$
is the $\mathcal{N}\left(\beta_{j}^{F},\sigma^{2}\lVert\eta_{j}\rVert_{2}^{2}\right)$
distribution truncated to $\left[-\infty,a_{j}\right]\cup\left[b_{j},\infty\right]$
for some constants $a_{j}$ and $b_{j}$. That is to say, the general
case simply replaces $\lambda$ with $a_{j}$ and $b_{j}$ but much
of our intuition in the orthogonal case still transfers over. Such a simple result is hinted at by \citet{fithian2015} who derived truncation regions of this form in the context of sequential testing. We shall refer to intervals based on this truncated distribution for the Z-statistic as $\TZ_V$ intervals (the $V$ here stands for variable inclusion); intervals which condition on $\hat M$ or both $\hat M$ and $\hat s$ are similarly termed the $\TZ_{M}$ and $\TZ_{Ms}$ intervals respectively (see Section \ref{sec:geninference}).

\bigskip
\begin{proposition}
	\label{prop:full}
	The conditional distribution $\eta_{j}^{\top}y\:|\:j\in\hat{M},\nu_{j}$
	is $\mathcal{N}\left(\beta_{j}^{F},\sigma^{2}\lVert\eta_{j}\rVert_{2}^{2}\right)$
	truncated to $\left[-\infty,a_{j}\right]\cup\left[b_{j},\infty\right]$
	where
	\[
	a_{j}=\lVert\eta_{j}\rVert_{2}^{2}\cdot\left(x_{j}^{\top}r_{j}-\lambda\right)\qquad b_{j}=\lVert\eta_{j}\rVert_{2}^{2}\cdot\left(x_{j}^{\top}r_{j}+\lambda\right)
	\]
	$r_{j}=\mathbf{X}_{-j}\hat{\beta}_{-j}-\nu_{j}$ and $\nu_{j}=\left(\mathbf{I}_n-\frac{1}{\lVert\eta_{j}\rVert_{2}^{2}}\eta_{j}\eta_{j}^{\top}\right)y$. Here, $x_j$ is the $j$th column of $\mathbf{X}$, $\mathbf{X}_{-j}$ is the submatrix of $\mathbf{X}$ after removing the $j$th column of $\mathbf{X}$ and $\hat{\beta}_{-j}$ is the lasso fit of $\nu_j$ to $\mathbf{X}_{-j}$ with penalty $\lambda$.
\end{proposition}

\bigskip

\begin{proof}
	The argument is simply an examination of the KKT conditions for the
	lasso. Using our decomposition for $y$ and noting 
	\[
	\mathbf{X}^{\top}c_{j}=\frac{e_{j}}{\lVert\eta_{j}\rVert_{2}^{2}}
	\]
	we have 
	\[
	\mathbf{X}^{\top}\left(\mathbf{X}\hat{\beta}-y\right) + \lambda \hat{s}=\mathbf{X}^{\top}\left(\mathbf{X}\hat{\beta}-\nu_{j}\right)-e_{j}\cdot\frac{z_{j}}{\lVert\eta_{j}\rVert_{2}^{2}}+\lambda\hat{s}.
	\]
	 The lasso solution satisfies 
	\begin{equation*}
	\begin{aligned}
	\mathbf{X}^{\top}\left(\mathbf{X}\hat{\beta}-\nu_{j}\right)-e_{j}\cdot\frac{z_{j}}{\lVert\eta_{j}\rVert_{2}^{2}}+\lambda\hat{s}=0 \\
	\hat{s}_{j}=\text{sign}\left(\hat{\beta}_{j}\right)\qquad\hat{\beta}_{j}\neq 0 \\
	\hat{s}_{j}\in\left[-1,1\right]\qquad\hat{\beta}_{j}=0.
	\end{aligned}
	\end{equation*}
	Since $\nu_{j}$ is fixed, our goal is to find the range of $z_{j}=\eta_{j}^{\top}y$
	values for which $\hat{\beta}_{j}\neq0$. When $\hat{\beta}_{j}=0$, we have 
	\[
	\left|x_{j}^{\top}\left(\mathbf{X}_{-j}\hat{\beta}_{-j}-\nu_{j}\right)-\frac{z_{j}}{\lVert\eta_{j}\rVert_{2}^{2}}\right|<\lambda.
	\]
	This shows that $\hat{\beta}_{j}=0$ if and only if $z_{j}\in\left[a_{j},b_{j}\right]$, where 
	\begin{equation*}
	\begin{aligned}
	a_{j}=\lVert\eta_{j}\rVert_{2}^{2}\cdot\left(x_{j}^{\top}r_{j}-\lambda\right) &\qquad b_{j}=\lVert\eta_{j}\rVert_{2}^{2}\cdot\left(x_{j}^{\top}r_{j}+\lambda\right) \\
	r_{j}=&\mathbf{X}_{-j}\hat{\beta}_{-j}-\nu_{j}.
	\end{aligned}
	\end{equation*}
	Note that in the orthonormal case, $x_{j}^{\top}r_{j}=0$ and $\lVert\eta_{j}\rVert_{2}^{2}=1$
	so that $a_{j}=-\lambda$ and $b_{j}=\lambda$.
\end{proof}

When $p > n$, the matrix $\mathbf{X}^{\top} \mathbf{X}$ is no longer invertible. However, the full target can be approximated by $\beta^F = \mathbf{A} \mathbf{X}^{\top}\mu$, where $\mathbf{A}$ is some approximation to the inverse of $\mathbf{X}^{\top} \mathbf{X}$; see, e.g., \citet{JM2015}. However, under this construction, we no longer have the condition that $\eta_j$ is orthogonal to $x_k$ for all $k\neq j$, i.e., as we move along the direction $\eta_j$ we affect the lasso solution not only for variable $j$ but also for the other variables. As a result, the proof of Proposition \ref{prop:full} no longer holds and we do not have easily computed formulas characterizing $\{j \in \hat M\}$. Nonetheless, it is still feasible to compute the truncation set via pathwise algorithms (Section \ref{sec:general}).

\subsection{Simulations for the full target case}\label{subsec:full_sim}

We conduct some simple simulations to demonstrate the benefit of conditioning only on $j\in\hat{M}$ (the minimal conditioning) rather than the entire active set $\hat{M}$.
\begin{itemize}
	\item Since the full regression coefficients only exist when $n>p$, we consider two cases: $\left(n,p\right)=\left(100,50\right)$ and $\left(n,p\right)=\left(100,90\right)$, the latter of which is close to being non-identifiable.
	\item The variables for each individual are independent standard multivariate normal: $x_{i}\sim \mathcal{N}\left(0,\mathbf{I}_{p}\right)$. (Note $x_i$, $i=1,\ldots,n$, denote the rows of $\mathbf{X}$ and $x_j$, $j=1,\ldots,p$ denote the columns of $\mathbf{X}$.)
	\item The outcome $y$ was generated as $y_{i}=x_{i}^{\top}\beta^{F}+\varepsilon_{i}$, where $\varepsilon_{i}\sim \mathcal{N}\left(0,1\right)$, for all $i=1,\ldots,n$.
	\item The first $k=5$ components of $\beta_{j}^{F}$ were set to either $0$, $\delta_{\text{low}}$, or $\delta_{\text{high }}$ (see following bullet). The remaining components were set to 0.
	\item We consider three signal settings: no signal, low signal, and high signal. To arrive at a meaningful definition of low and high signal,
	we considered the distribution of $\max_{j}n^{-1}\left|x_{j}^{\top}y\right|$ under the global null (where $y\sim \mathcal{N}\left(0,\mathbf{I}_n\right)$). Asymptotically,
	this noise distribution concentrates around $\sqrt{\frac{2\log p}{n}}$;
	for our purposes we simulated from this distribution 1000 times and
	set $\delta_{\text{low}}$ as the median of this distribution and
	$\delta_{\text{high}}$ as the 99th quantile of this distribution
	plus an additional offset of $0.25$. When the signal $\beta_{j}^{F}=\delta_{\text{low}}$,
	it is difficult to distinguish variable $j$ from noise; when $\beta_{j}^{F}=\delta_{\text{high}}$,
	the signal has exceeded the noise threshold and is easy to detect.
	\item For our penalty parameter $\lambda$, we considered two choices. For
	the first choice we try to approximate the $\lambda$ that would have
	been picked by cross validation. To do this, we ran lasso at each
	setting 100 times, computed the $\lambda$ that was chosen by 10-fold
	cross validation each time, and set $\lambda$ to be the median of
	the resulting distribution. For our second choice, we chose the universal threshold value: $\lambda=\sqrt{\frac{2\log p}{n}}$. In all instances the ``cv'' $\lambda$ is quite a bit smaller than the ``universal threshold'' $\lambda$; cross validation prioritizes predictive performance and is not model selection consistent (it lets in too many noise variables) while the universal threshold is model selection consistent.
	\item We simulated 1000 independent datasets. For each dataset, we constructed a 90\% confidence interval of $\beta_{j}^{F}$ for each variable $j$ in the lasso active set. 
\end{itemize}
Figures \ref{fig:full_null}, \ref{fig:full_n100p50}, and \ref{fig:full_n100p90} display boxplots of interval lengths produced by each method, as well as giving coverage information. When inverting the truncated normal pivot, numerical inaccuracies occasionally result in ``infinite'' intervals (this occurs if the observed Z-statistic is too close to an endpoint of the truncation region). The median length and empirical coverage can still be computed using the usual definitions; for display purposes, however, our boxplots truncate the length of intervals with infinite length to the maximum finite length observed in a simulation. We summarize the key takeaways:
\begin{itemize}
	\item Many practitioners incorrectly assume that Bonferroni corrected intervals can protect against the effects of selection. When there is no signal $\beta^{F}=0$, Figure \ref{fig:full_null} shows that the Bonferroni intervals can severely under-cover. 
	\item Conditioning on $j\in\hat{M}$ rather than $\hat{M}$ produces intervals that are significantly shorter. The benefits are greater when $\lambda$ is smaller (chosen by cross validation). This is because $\hat{M}$ carries less information about $y$ when $\lambda$ is larger (say, near the universal threshold). In other words, increasing $\lambda$ is one way to control our exploration cost, a fact that we shall exploit in Section \ref{subsec:reg_target}. For sufficiently large $\lambda,$
	we will obtain $\hat{M}=\emptyset$ with high probability in which case there is no exploration cost because there is no exploration.
	\item $\TZ_V$ produces infinite length intervals much less frequently than $\TZ_M$ or $\TZ_{Ms}$ (almost never). Since truncation is less severe for $\TZ_V$, the observed Z-statistic is less likely to fall near the left-most or right-most boundary of the truncation region.
	\item The simulation results are consistent with the good performance of the $\TZ_V$ intervals observed for the prostate data (Figure \ref{fig:prostate}) and HIV data (Figure \ref{fig:hiv}).
\end{itemize}

\begin{figure}[bhtp]
	\begin{centering}
		\includegraphics[width=0.33\paperwidth]{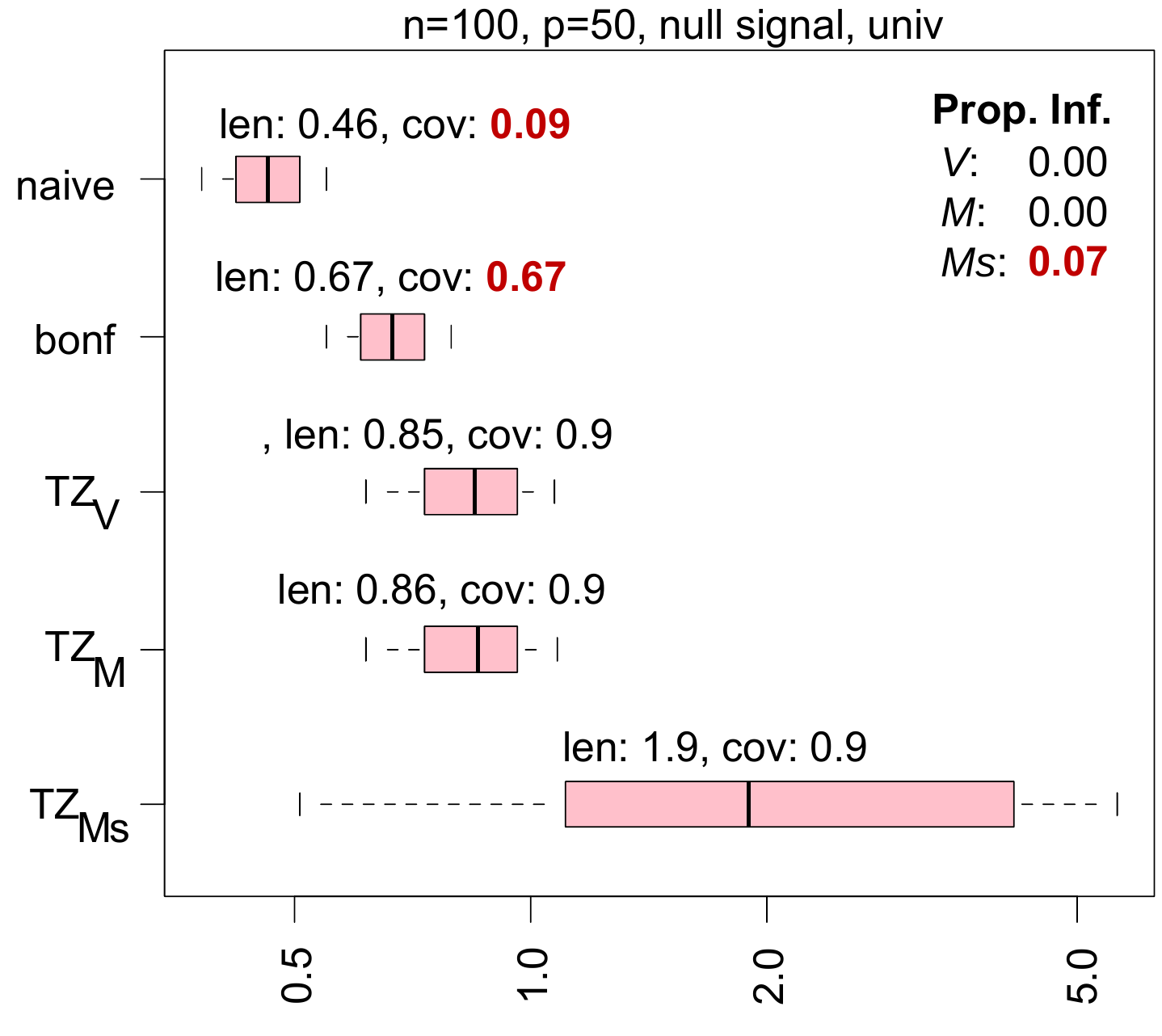}\includegraphics[width=0.33\paperwidth]{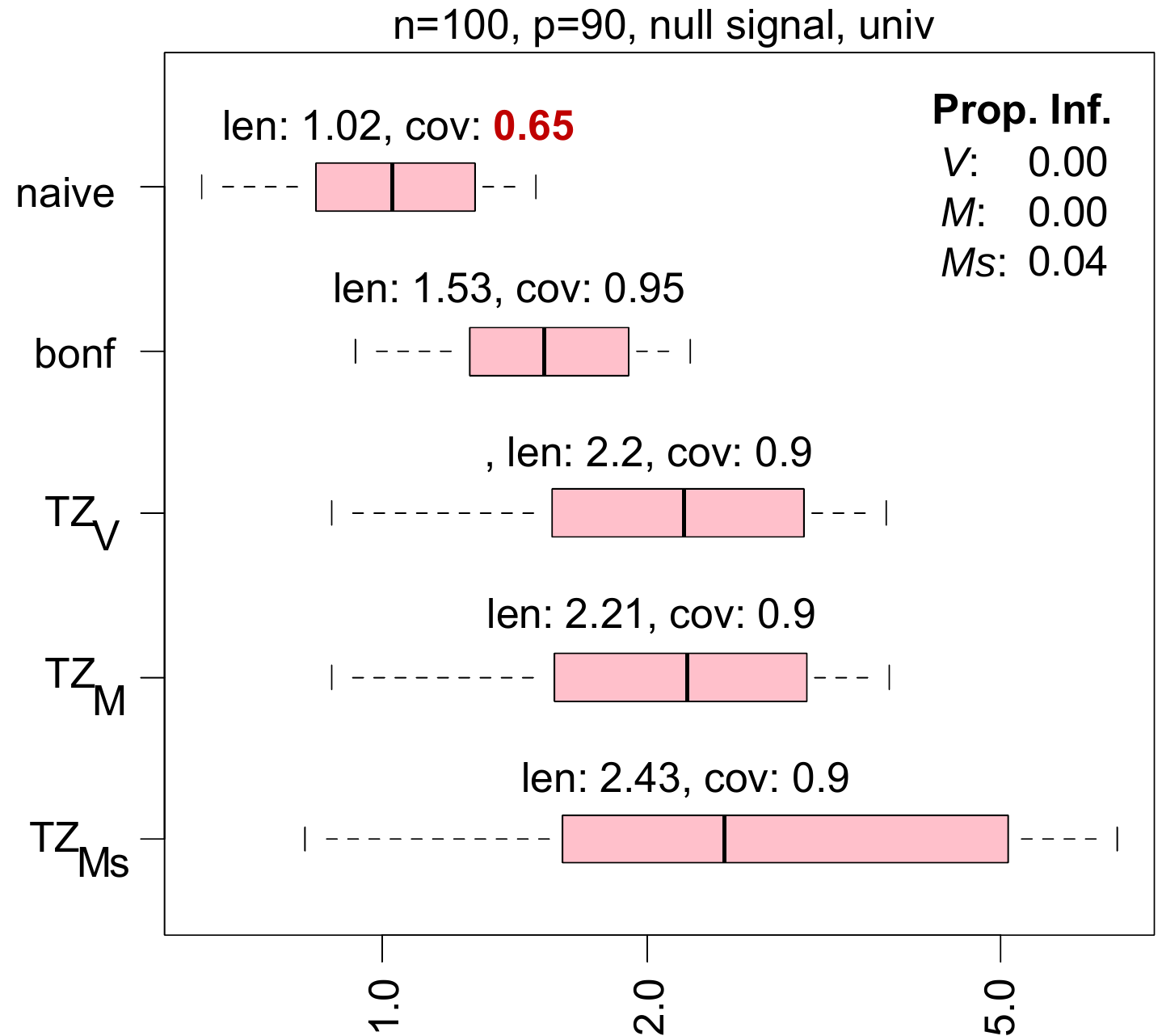}
	\par\end{centering}
	
	\caption{\em Boxplot of lengths of 90\% confidence intervals for ``full''
		regression coefficients. Five interval methods are compared: naive
		(ignoring selection), Bonferroni adjusted, $\TZ_{V}$ , $\TZ_{M}$,
		and $\TZ_{Ms}$. Reported are the median interval length, the empirical
		coverage, and the proportion of ``infinite'' intervals (the infinite
		length results from numerical inaccuracies when inverting a truncated
		normal CDF); the boxplots set infinite lengths to the maximum finite observed length. In the left panel $n=100,p=50$ and in the right
		panel $n=100,p=90$. There is no signal, $\beta=0$. The lasso penalty parameter $\lambda$ is set to $\sqrt{\frac{2\log(p)}{n}}$. }
	\label{fig:full_null}
\end{figure}

\begin{figure}[bhtp]
	\begin{centering}
		\includegraphics[width=0.33\paperwidth]{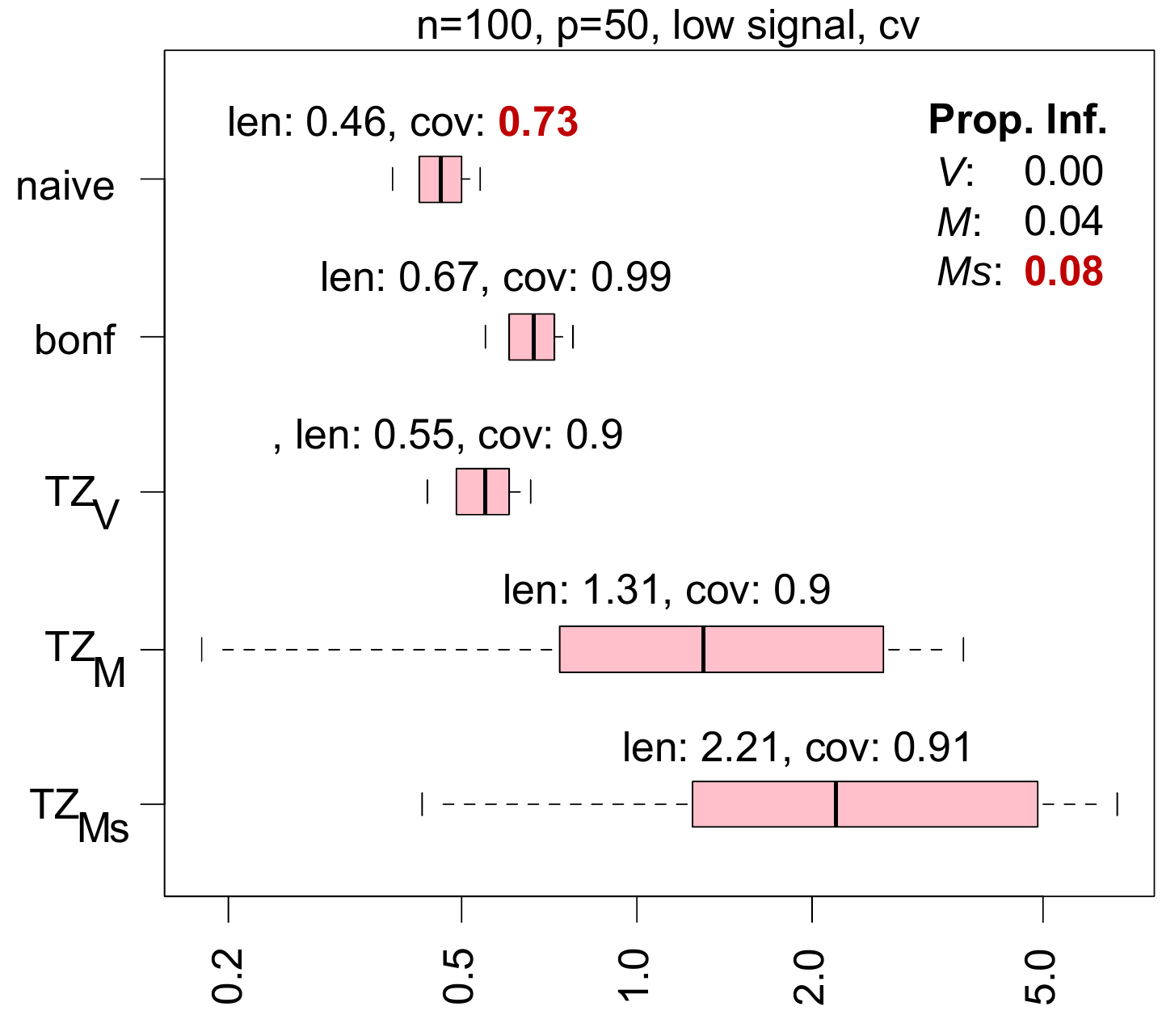}\includegraphics[width=0.33\paperwidth]{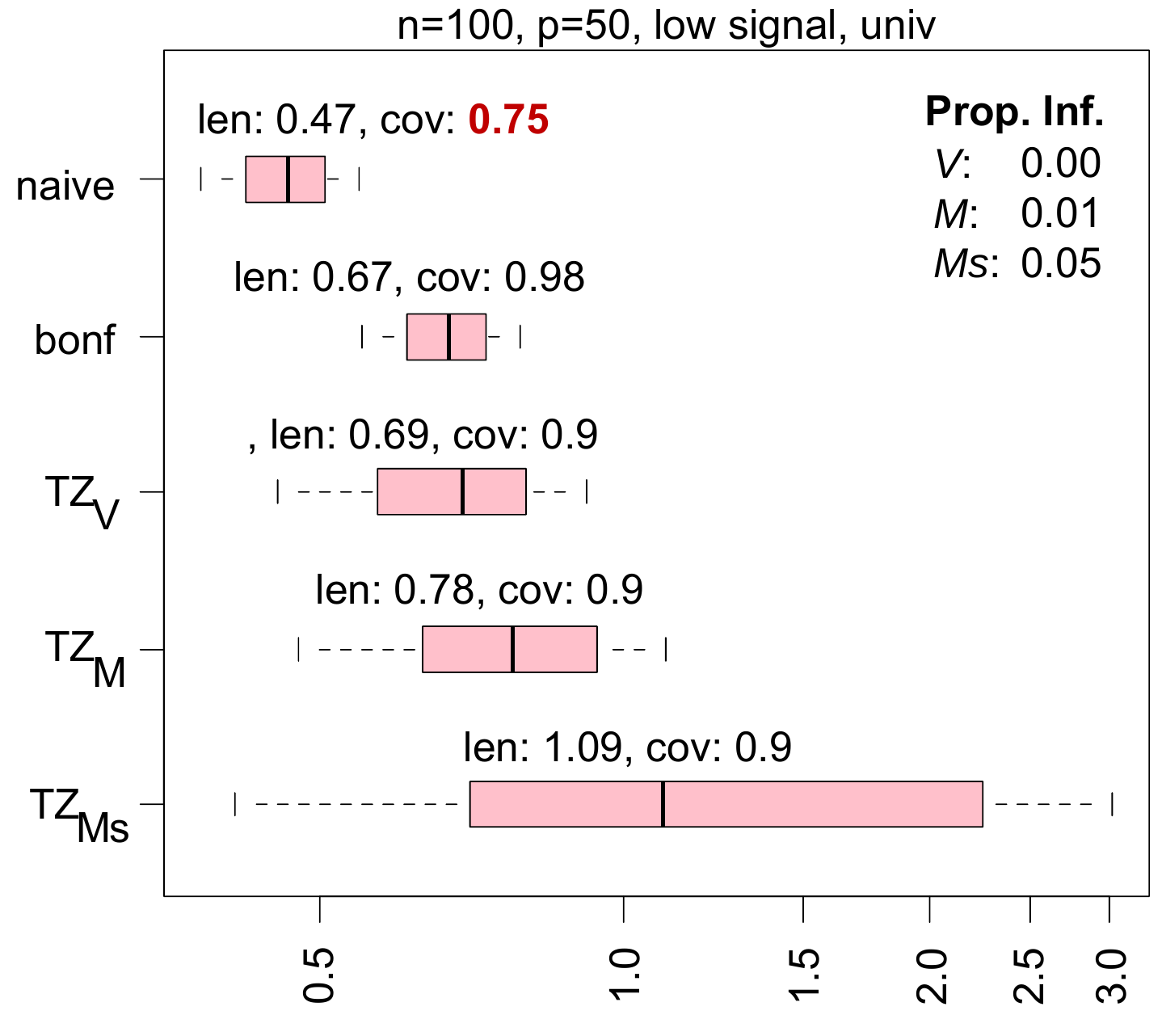}
		\par\end{centering}
	\begin{centering}
		\includegraphics[width=0.33\paperwidth]{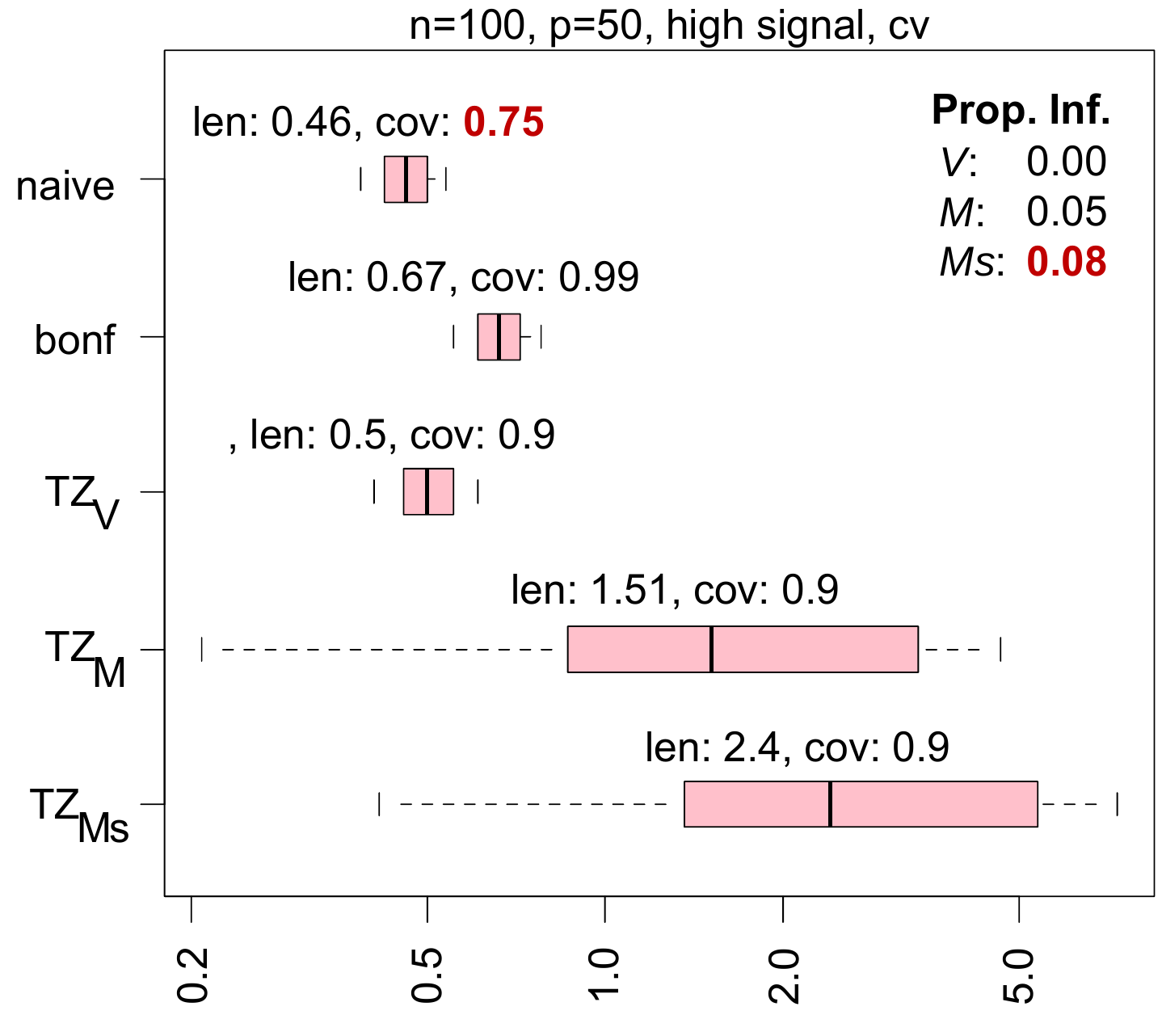}\includegraphics[width=0.33\paperwidth]{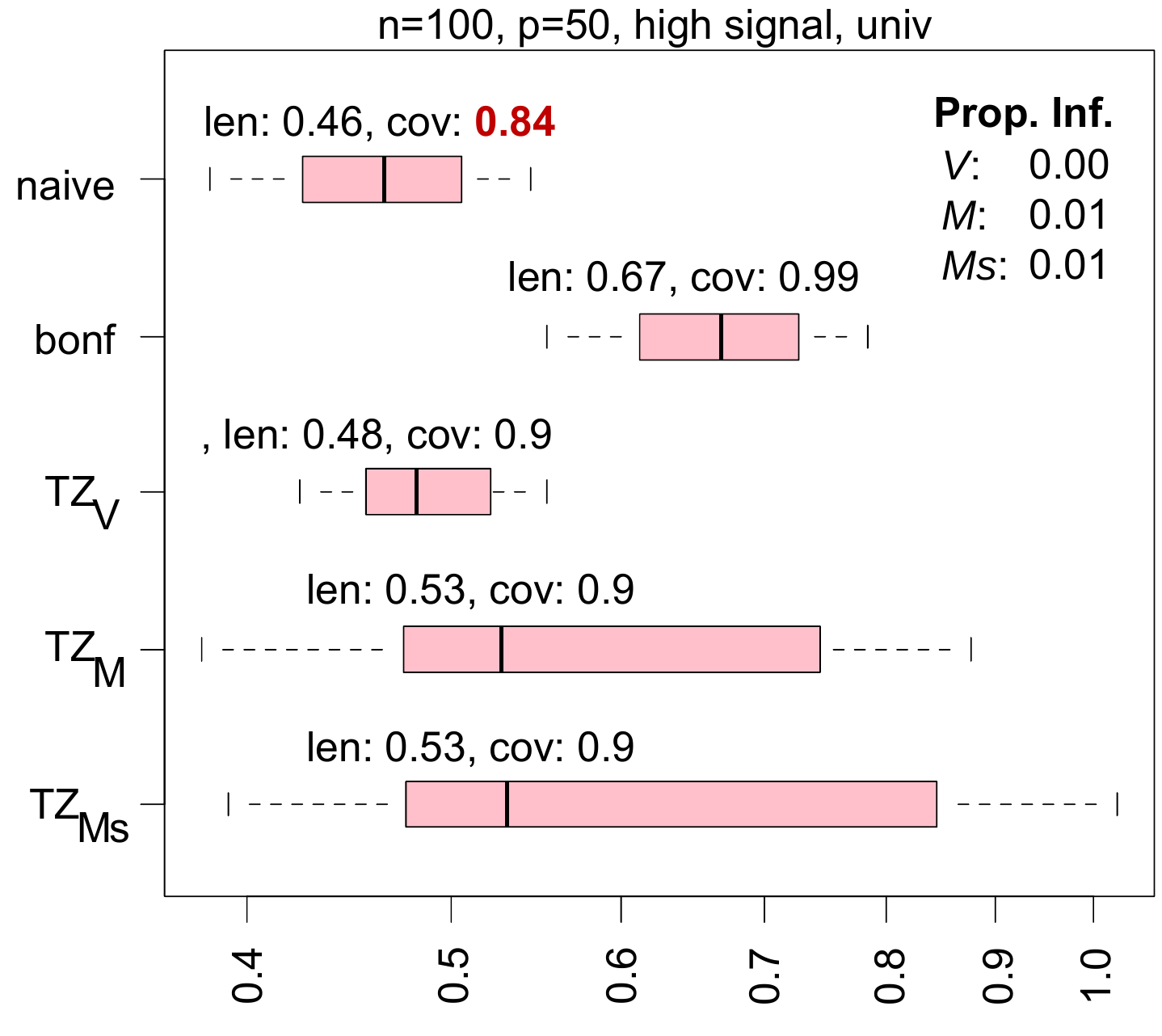}
		\par\end{centering}
	\caption{\em Boxplot of lengths of 90\% confidence intervals for ``full''
		regression coefficients. Five interval methods are compared: naive
		(ignoring selection), Bonferroni adjusted, $\TZ_{V}$ , $\TZ_{M}$,
		and $\TZ_{Ms}$. Reported are the median interval length, the empirical
		coverage, and the proportion of ``infinite'' intervals (the infinite
		length results from numerical inaccuracies when inverting a truncated
		normal CDF); the boxplots set infinite lengths to the maximum finite
		observed length. Here $n=100,p=50$. The first five components of
		$\beta$ are set to $\delta_{\text{low}}=0.24$ (top panels) or $\delta_{\text{high}}=0.62$
		(bottom panels) and the remaining components are 0. The lasso penalty
		is either set at the universal threshold value $\sqrt{\frac{2\log p}{n}}\approx 0.28$
		(right panels) or at a value approximating the behavior of 10-fold
		cross validation (0.14 and 0.11 respectively for the low and high
		signal cases).}
	\label{fig:full_n100p50}
\end{figure}

\begin{figure}[bhtp]
	\begin{centering}
		\includegraphics[width=0.33\paperwidth]{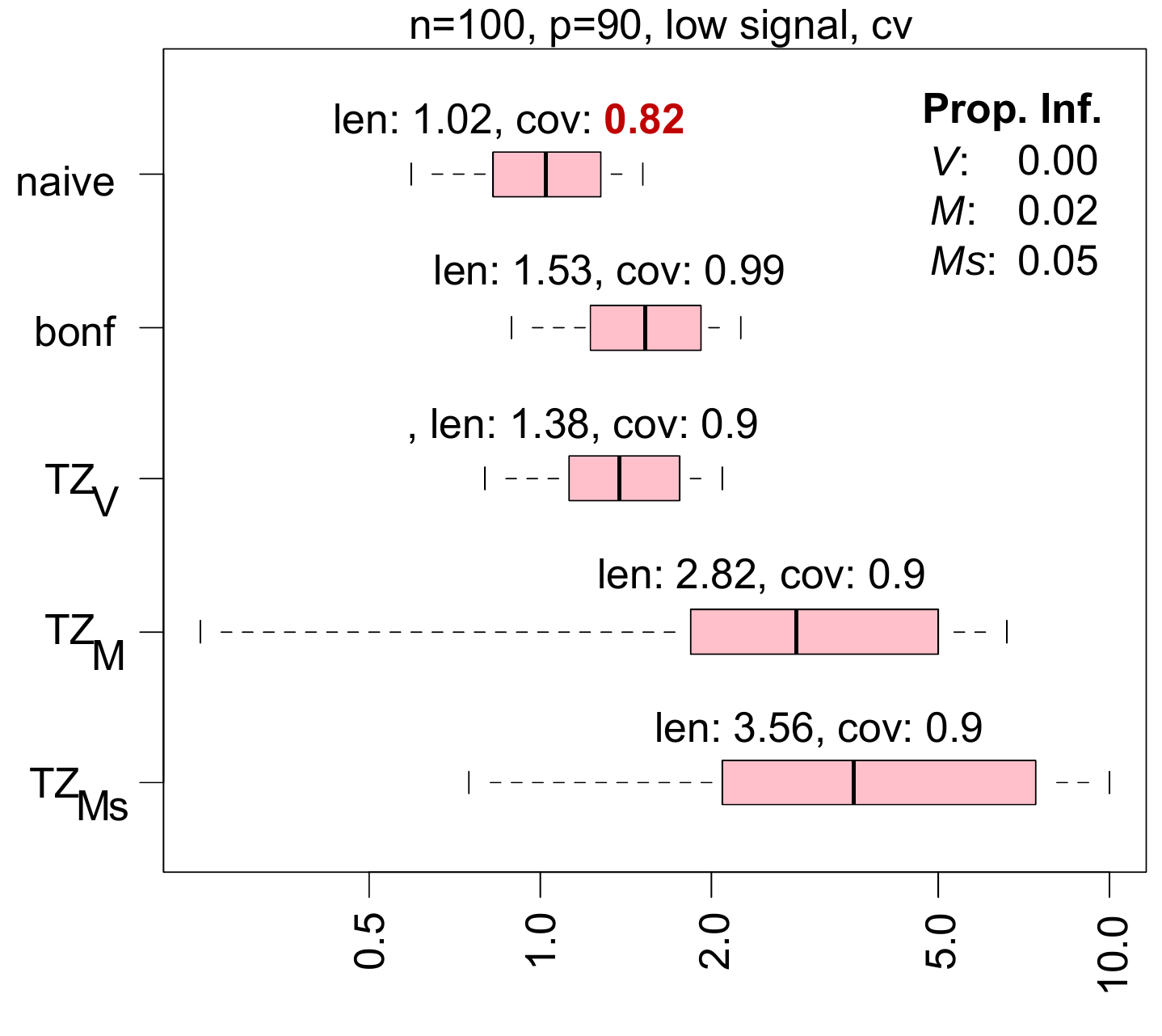}\includegraphics[width=0.33\paperwidth]{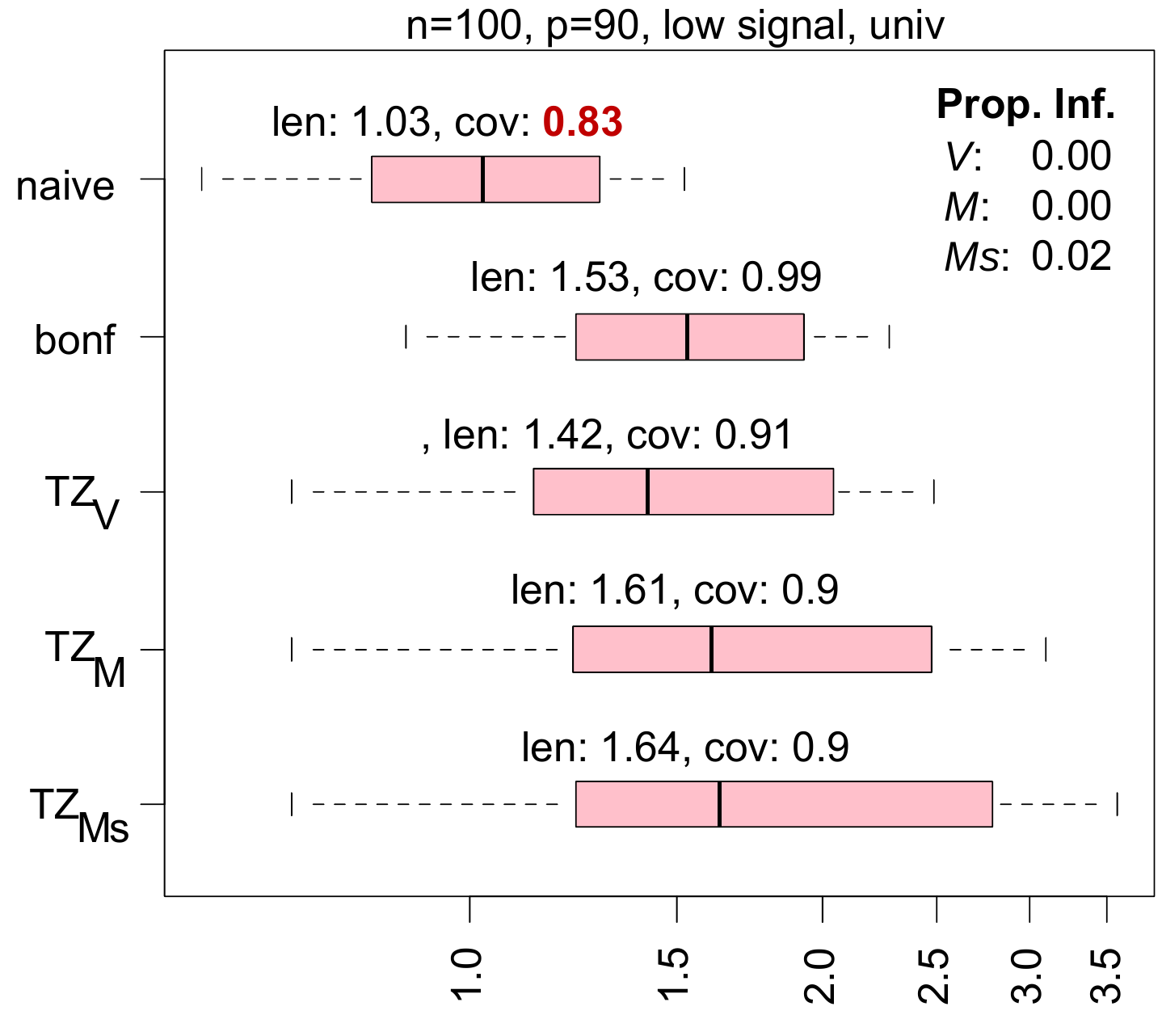}
		\par\end{centering}
	\begin{centering}
		\includegraphics[width=0.33\paperwidth]{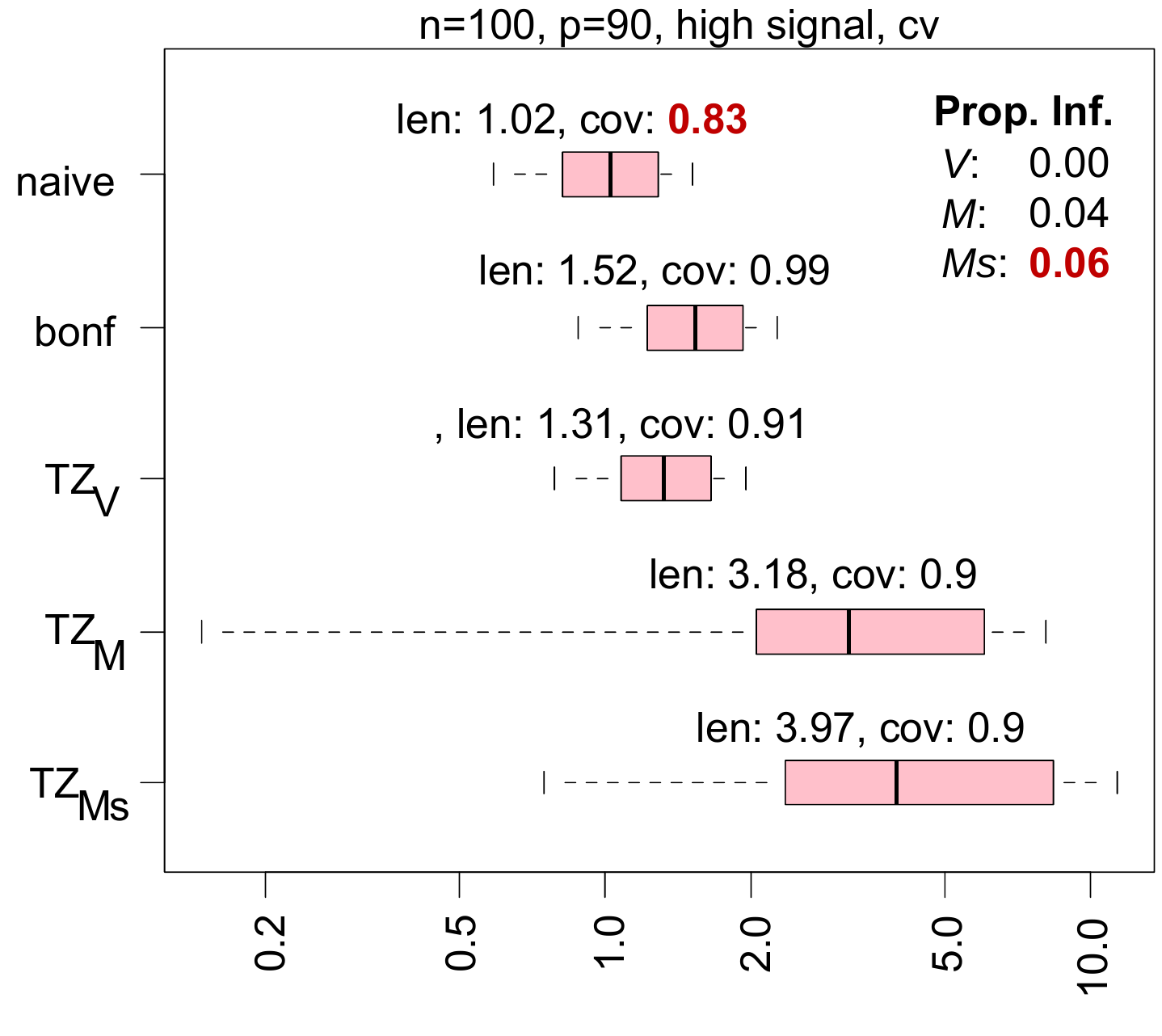}\includegraphics[width=0.33\paperwidth]{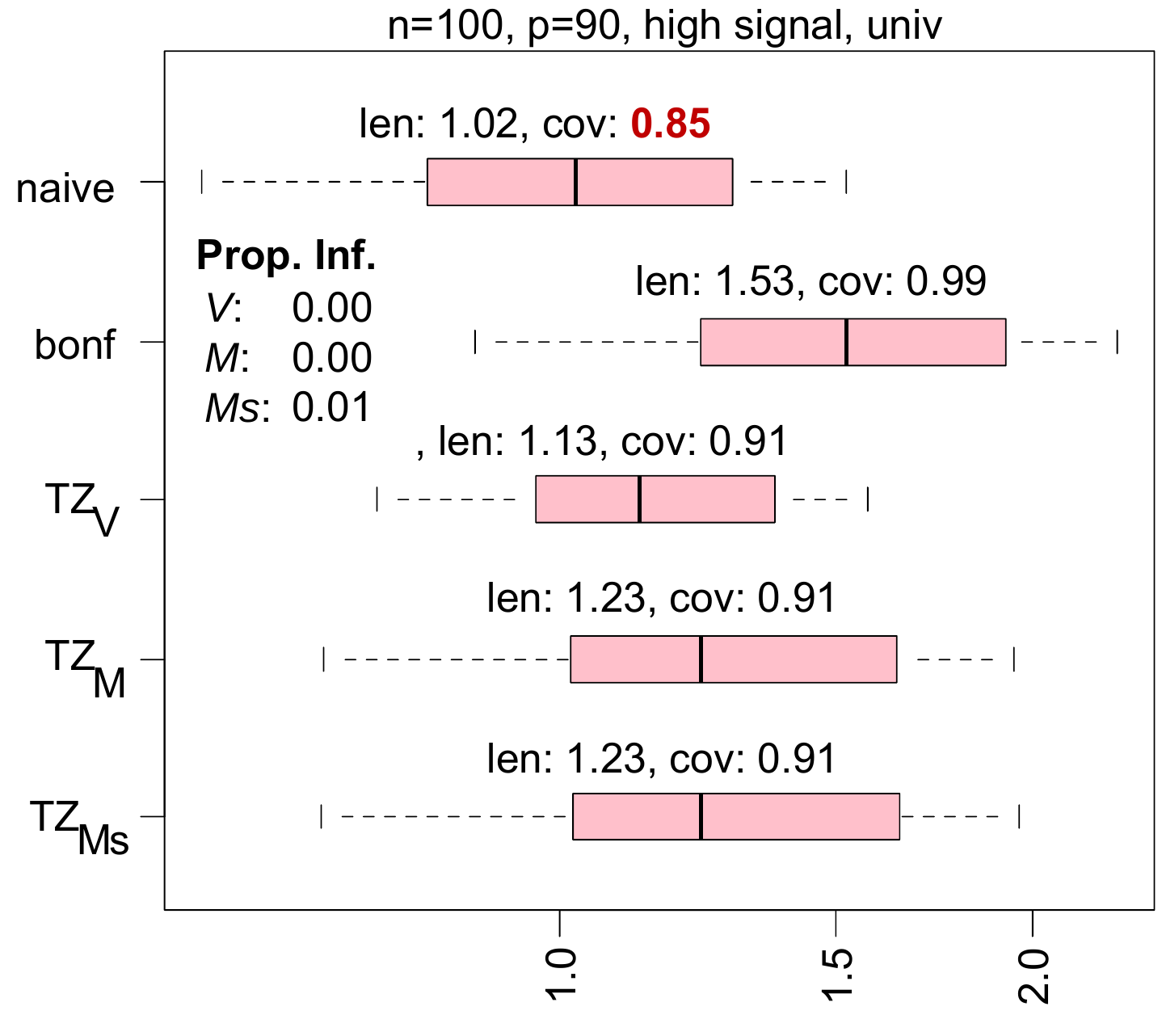}
		\par\end{centering}
	\caption{\em Boxplot of lengths of 90\% confidence intervals for ``full''
		regression coefficients. Five interval methods are compared: naive
		(ignoring selection), Bonferroni adjusted, $\TZ_{V}$ , $\TZ_{M}$,
		and $\TZ_{Ms}$. Reported are the median interval length, the empirical
		coverage, and the proportion of ``infinite'' intervals (the infinite
		length results from numerical inaccuracies when inverting a truncated
		normal CDF); the boxplots set infinite lengths to the maximum finite
		observed length. Here $n=100,p=90$. The first five components of
		$\beta$ are set to $\delta_{\text{low}}=0.26$ (top panels) or $\delta_{\text{high}}=0.63$
		(bottom panels) and the remaining components are 0. The lasso penalty
		is either set at the universal threshold value $\sqrt{\frac{2\log p}{n}}\approx0.30$
		(right panels) or at a value approximating the behavior of 10-fold
		cross validation (0.14 and 0.12 respectively for the low and high
		signal cases).}
		\label{fig:full_n100p90}
\end{figure}

\section{Inference for partial regression targets}

\label{sec:partial} Unlike Section \ref{sec:full}, here we are concerned
with situations where the data is used both to select variables and
to determine how the relationship between selected variables and the
response is summarized. For example when the number of variables exceeds
the sample size, the full regression coefficient $\beta^{F}$ is no
longer well defined so we cannot always summarize the effect of variable
$j$ using $\beta_{j}^{F}$. We could still decide a priori to study
only the marginal relationship between $x_{j}$ and $y$ if $j$ is
selected (e.g., the marginal correlation). In this case the data remains
uninvolved in target formation and we can continue to condition only
on $j\in\hat{M}$. However, often, our goal is to understand how the variables
work in concert, to define a target that can meaningfully capture
key aspects of the joint distribution between the outcome and our
predictors.

In high dimensional settings, this necessarily means that we need
to use the data to inform our choice of target. Usually, this occurs when the lasso is used to select a set of variables $\hat{M}=M$ and then the effect of variable $j\in M$ is summarized by $\beta_{j}^{\left(M\right)}$, where 
\[
\beta^{\left(M\right)}=\left(\mathbf{X}_{M}^{\top}\mathbf{X}_{M}\right)^{-1}\mathbf{X}_{M}^{\top}\mu
\]
is the regression of $\mu$ on the selected variables. These are the population partial regression coefficients, and we called this the {\em partial target} earlier.

Now, when variable $j$ is selected, how we choose to summarize its effect depends on what other variables have been selected. It is insufficient to condition on $\left\{ j\in\hat{M}\right\} $; this event does not uniquely determine the quantity we are trying to infer. We need to know the identity of the other selected variables. In other words, two conditions must be satisfied for us to build a confidence interval for $\beta_{j}^{\left(M\right)}$: (i) $j\in\hat{M}$ {\em and} (ii) $\hat{M}=M$ (otherwise we might be making an interval for $\beta_{j}^{\left(M'\right)}$ with $M'\neq M$). Therefore, the conditional coverage guarantee we seek is 
\[
	\mathbb{P}\left(\beta_{j}^{\left(M\right)}\notin C_{j}\:\big|\:j\in\hat{M}, \hat{M}=M\right)\leq\alpha.
\]
Compared to (\ref{eq:cover_full}), the additional conditioning on
$\hat{M}=M$ is the cost of target formation. \citet{lee2016} showed how to construct such intervals with the following proposition.

\begin{proposition} \label{prop:partial} For $\eta_{j}=\mathbf{X}_{M}\left(\mathbf{X}_{M}^{\top}\mathbf{X}_{M}\right)^{-1}e_{j}$,
	the conditional distribution $\eta_{j}^{\top}y\:|\:j\in\hat{M},\hat{M}=M,\nu_{j}$ is $\mathcal{N}\left(\beta_{j}^{\left(M\right)},\sigma^{2}\lVert\eta_{j}\rVert_{2}^{2}\right)$
	truncated to a disjoint union of intervals $\cup_{s}\left[l_{s},u_{s}\right]$ where the union is over $s\in\left\{ -1,1\right\} ^{\left|M\right|}$.
\end{proposition}

The fact that the truncation region is a union of intervals is not
surprising (see Section \ref{sec:general}); what is amazing is that
the interval endpoints $l_{s}$ and $u_{s}$ have explicit formulas
(see Appendix \ref{app:truncation:points}). Proposition \ref{prop:partial} should be compared
to Proposition \ref{prop:full}. The truncation $\left[-\infty,a_{j}\right]\cup\left[b_{j},\infty\right]$
includes values of $\eta_{j}^{\top}y$ for which the lasso selects
a model $N\neq M$ so long as $j\in N$; the truncation to $\cup_{s}\left[l_{s},u_{s}\right]$
further excludes such values of $\eta_{j}^{\top}y$ (it is a subset
of $\left[-\infty,a_{j}\right]\cup\left[b_{j},\infty\right]$). In
fact, even in moderate dimension problems (say $p\approx25$), the
intervals $\left[l_{s},u_{s}\right]$ which fall within a practical
range, say $\left[-10^{2},10^{2}\right]$, are very few (intervals
falling outside this practical range contribute negligibly to p-value
calculations and can be numerically ignored). This is because $\hat{M}$
is quite variable even in moderate dimensions. As we perturb the value
of $\eta_{j}^{\top}y$, it becomes quite likely for one of the $p$
variables to enter or leave the active set.

All this is to say that when we condition on $\hat{M}=M$, the range
of our test statistic $\hat{\beta}_{j}^{\left(M\right)}=\eta_{j}^{\top}y$
is quite severely truncated. This truncation leads to a loss in power,
hence longer intervals. In practice, we find that intervals conditioning
on $\hat{M}=M$ are often unacceptably long. This says that the cost
of target formation is particularly high (relative to variable selection
using the lasso) and we must exercise great caution in how we let
the data inform our queries. As evidence of this claim, one can refer
back to our prostate example (Figure \ref{fig:prostate}) and the
simulation results for the full target (Figure \ref{fig:full_n100p50}).
There, we see that whether we condition on $j\in\hat{M}$ or $\hat{M}=M$
makes a huge difference (the discrepancy is in fact the cost of target
formation!); on the other hand, whether one conditions on the signs
$\hat{s}$ makes much less of a difference.

\subsection{Stable Target Formation} \label{subsec:reg_target}

The above discussion suggests that we cannot afford to let the data overly influence our choice of target. In particular,
choosing to summarize the effect of variable $j$ using $\beta_{j}^{\hat{M}}$
where $\hat{M}$ is chosen by the lasso can lead to uselessly long
intervals. There are two simple remedies for limiting the role of
the data in target formation: 
\begin{itemize}
	\item \emph{Data Splitting}: We can split the data and choose our target based on only a part of the data. 
	\item \emph{Randomization}: We can add some noise to the data and choose our target based on the corrupted data. 
\end{itemize}
Both approaches can lead to drastically shorter intervals. But they
also share some downsides: 
\begin{itemize}
	\item With a different split of the data (or different noise instance), not only will our intervals be different but the very targets of our inference may change. 
	\item The model selected from a sample split or on corrupted data can be much worse than the model selected on the full data. We may get shorter intervals but we are also more likely to be asking the wrong questions!
\end{itemize}
For these reasons, we seek a deterministic approach for limiting the
cost of target formation. As discussed, the problem with summarizing
our effects using $\beta_{j}^{\hat{M}}$ is the instability of $\hat{M}$.
If the target of our research changes too often, that is probably not
a good thing!

Intuitively, we want to use a regularized/more stable version of $\hat{M}$
to choose our target. We can think of $\hat{M}$ as containing two
types of variables: those with very strong signals and those with
marginal signals (whose significance is ambiguous). We will call the
first type of variable \emph{high value targets} (these are variables
we would not want to miss out on) and the second type of variable
\emph{low value targets.} When we perturb the data, it is the low
value targets that are susceptible to entering or leaving the lasso
active set. The reason why target formation is so costly is because
we have allowed low value targets to participate equally as high value
targets in determining our target. The solution seems straightforward
then: allow only the high value targets to influence our choice of
target.

We will choose a subset $\hat{H}\subset\hat{M}$ of high value targets (see details below). How we choose to summarize the effect of a variable $j\in\hat{M}$ depends on whether $j$ is a high value target: 
\begin{itemize}
	\item \emph{High Value}: We summarize the effect of $j$ using $\beta_{j}^{\hat{H}}$
	where 
	\[
	\beta^{\hat{H}}=\left(\mathbf{X}_{\hat{H}}^{\top}\mathbf{X}_{\hat{H}}\right)^{-1}\mathbf{X}_{\hat{H}}^{\top}\mu.
	\]
	So our choice of target is fully adaptive for high value targets. 
	\item \emph{Low Value}: If variable $j$ is selected by the lasso but is not deemed a high value target, we summarize its effect via $\beta_{j}^{\hat{H}\cup\{j\}}$ where
	\[
	\beta^{\hat{H}\cup\{j\}}=\left(\mathbf{X}_{\hat{H}\cup\{j\}}^{\top}\mathbf{X}_{\hat{H}\cup\{j\}}\right)^{-1}\mathbf{X}_{\hat{H}\cup\{j\}}^{\top}\mu
	\]
	and $\mathbf{X}_{\hat{H}\cup\{j\}}$ is the matrix containing the high value targets as well as variable $j$. The coefficient $\beta^{\hat{H}\cup\{j\}}_j$ is the effect of variable $j$ after partialing out the effect of the high value targets, i.e., it allows us to ask the question whether variable $j$ contributes any explanatory power beyond the variables in $\hat{H}$. 
\end{itemize}
One may be worried whether we have greatly sacrificed interpretability by defining our inferential targets in this way. For a variable $j$, we seek its regression coefficient after partialing out the effect of variables in $\hat{H}$; in contrast, the usual ``partial'' target seeks the regression coefficient of $j$ after partialing out the effect of variables in $\hat{M}$. Note that these regression coefficients should be quite similar in magnitude if (i) our low value targets
in $\hat{M}\setminus\hat{H}$ have small coefficients or (ii) if the active variables in $\hat{M}$ are not too highly correlated. But of course, if we define ``high value'' targets appropriately, then (i) should occur naturally; and the fact that the lasso usually selects only one out of a group of highly correlated variables makes (ii) likely as well. Thus, we should expect expect the inferential target we define
above to be quite close to the familiar partial regression coefficient based on $\hat{M}$; re-examining again the prostate data (Figure \ref{fig:prostate}), this turns out to be the case (the MLEs for the targets defined above are very close to the MLEs of the partial regression coefficients).

The chief effect of the above construction is that our choice of target,
denoted $\theta^{\hat{H}}$, depends not on $\hat{M}$ but on $\hat{H}$
which is much more stable. We construct a confidence interval for
$\theta_{j}^{H}$ whenever (i) $j\in\hat{M}$ and (ii) $\hat{H}=H$.
Hence the conditional confidence guarantee that we seek is 
\[
	\mathbb{P}\left(\theta_{j}^{H}\in C_{j}\:\big|\:j\in\hat{M},\hat{H}=H\right)\leq\alpha.
\]
Our notation serves to emphasize that: 
\begin{itemize}
	\item The less stable $\hat{M}$ is used for variable selection, i.e., choosing the variables whose effects we want to summarize. 
	\item The more stable $\hat{H}$ is used for target formation, i.e., determining how the effects of the chosen variables are to be summarized. This makes sense since target formation is the more costly component of hypothesis generation. 
\end{itemize}
We will see that by stabilizing the target formation process, the
resulting intervals are much shorter especially the intervals for
high value targets.

\subsection{How to Define a High Value Target?}

From among the selected variables $\hat{M}$, how should we choose
the subset of high value targets $\hat{H}$? We introduce two methods:
the stable-$\ell_{1}$ approach and the stable-$t$ approach. In fact,
there are many possible proposals for creating a more stable version
of $\hat{M}$. We found the stable-$\ell_{1}$ and stable-$t$ approaches
to be a good compromise between interpretability and computational
feasibility and one that we hope will appeal to a wide group of practitioners. However, it should be clear that our techniques can be easily modified to suit other definitions of stability.

\bigskip{}

\textbf{Stable}-$\ell_{1}$: We saw in Section \ref{subsec:full_sim}
that the $\TZ_{Ms}$ intervals perform better when the lasso penalty
parameter is high (say near the universal threshold); this is because
the active set tends to be more stable for larger penalties. Thus,
a natural approach to constructing $\hat{H}$ is simply to let it
be the lasso active set but with a higher penalty parameter, $\lambda_{\text{high}}$,
than was used to select $\hat{M}$. In this case, there is no guarantee
that $\hat{H}\subset\hat{M}$ (though it often is) but the above development does not require this. The proposal is also attractive because as the following proposition shows, it is computationally no more difficult than conditioning on $\hat{M}=M$. 

\begin{proposition} 
The conditional distribution of $\eta^{\top}y$ given (i) $j\in\hat{M}$, (ii) $\hat{H}=H$, and (iii) $P_{\eta^{\perp}}y$ is the $\mathcal{N}\left(\eta^{\top}\mu,\sigma^{2}\lVert\eta\rVert_{2}^{2}\right)$
	distribution truncated to a set of the form 
	\[
	\left(\left[-\infty,a\right]\cup\left[b,\infty\right]\right)\cap\left(\cup_{s}\left[l_{s},u_{s}\right]\right).
	\]
	The endpoints $a,b$ are determined by the constraint $j\in\hat{M}$
	and the union is over all sign vectors $s\in\left\{-1, 1\right\} ^{\left|\hat{H}\right|}$ such that $\left(H,s\right)$ is compatible with the $\lambda_{\text{high}}$ lasso solution for some element of $\left\{ y':y'=\frac{\eta}{\lVert\eta\rVert_{2}^{2}}\cdot z+P_{\eta^{\perp}}y,z\in\mathbb{R}\right\} $.
\end{proposition} 
The proposition is a simple consequence of propositions
\ref{prop:full} and \ref{prop:partial}. Conditioning on $\hat{H}=H$
truncates $\eta^{\top}y$ to a union of intervals. The larger the
penalty parameter corresponding to $\hat{H}$, the larger this union
of intervals tends to be, i.e., the less severe the truncation. We
usually choose $\lambda_{\text{high}}$ close to the universal threshold
value of $\sqrt{\frac{2\log p}{n}}$.

\bigskip{}

\textbf{Stable-$t$}: Our second approach is motivated by the fact
that while the Bonferroni correction does not guarantee valid coverage
in all situations, it usually works pretty well. This suggests taking
$\hat{H}$ to be those variables in $\hat{M}$ with $t$-statistics
surpassing a Bonferroni corrected threshold. We first fit a OLS model
using all the variables in $\hat{M}$, i.e., 
\[
\hat{\beta}^{\hat{M}}=\left(\mathbf{X}_{\hat{M}}^{\top}\mathbf{X}_{\hat{M}}\right)^{-1}\mathbf{X}_{\hat{M}}^{\top}y
\]
and allow $j$ to be a high value target if the $t$-statistic for
$\hat{\beta}_{j}^{\hat{M}}$ is large, i.e., if 
\[
\left|\frac{\hat{\beta}_{j}^{\hat{M}}}{\sigma\left(\mathbf{X}_{\hat{M}}^{\top}\mathbf{X}_{\hat{M}}\right)_{jj}^{-1}}\right|>c
\]
for some cutoff $c$. If we choose $c$ by Bonferroni, it has the
form $\left|\Phi^{-1}\left(\frac{\alpha}{2p}\right)\right|\approx\sqrt{2\log p}$ for large $p$; for such a choice, the stable-$t$ method behaves similar to the stable-$\ell_{1}$ method (with $\lambda_{\text{high}}=\sqrt{\frac{2\log p}{n}}$) in many scenarios. The Bonferroni corrected cutoff ensures that we remove most variables whose importance is ambiguous. Once again, this proposal is attractive because it is computationally no more costly
than conditioning on $\hat{M}=M$. 

\begin{proposition} The conditional
	distribution of $\eta^{\top}y$ given (i) $j\in\hat{M}$, (ii) $\hat{H}=H$, and (iii) $P_{\eta^{\perp}}y$ is the $\mathcal{N}\left(\eta^{\top}\mu,\sigma^{2}\lVert\eta\rVert_{2}^{2}\right)$
	distribution truncated to a set of the form 
	\[
	\left(\left[-\infty,a\right]\cup\left[b,\infty\right]\right)\cap\left(\cup_{k}\left[l_{k},u_{k}\right]\right)
	\]
	where the endpoints $l_{k}$ and $u_{k}$ have explicit formulae.
\end{proposition} 

\begin{proof} Below, all statements are implicitly conditioned on $P_{\eta^{\perp}}y$ which restricts the data to a line. Note that $\hat H = \hat H (\hat M)$ is a function of $\hat M$ and that
\[
\left\{\hat H(\hat M) = H, j \in \hat M\right\} = \bigcup_{M: j \in M} 
\left\{\hat H(M) = H, \hat M = M\right\}.
\]
We already know $\{ \hat M = M\}$ is a union of intervals whose endpoints have known formulae. To finish characterizing the LHS of the above expression, it then suffices to show that $\{\hat H(M)=H\}$ is a union of intervals and compute the interval endpoints.

Fix $M$. Let us write 
\[
y'=cz+\nu\qquad c=\frac{\eta}{\lVert\eta\rVert_{2}^{2}}\qquad\nu=P_{\eta^{\perp}}y
\]
and regress $y'$ on $\mathbf{X}_{M}$ . The $t$-statistic for
the $k$th regression coefficient has the form 
\[
t_{k}=\frac{y'^{\top}\tilde{x}_{k}}{\sigma}\qquad\tilde{x}_{k}=\frac{\mathbf{X}_{M}\left(\mathbf{X}_{M}^{\top}\mathbf{X}_{M}\right)^{-1}e_{k}}{\lVert\mathbf{X}_{M}\left(\mathbf{X}_{M}^{\top}\mathbf{X}_{M}\right)^{-1}e_{k}\rVert_{2}}.
\]
Note that $\left|t_{k}\right|<\xi\sigma$ if $\left|z\cdot c^{\top}\tilde{x}_{k}+\nu^{\top}\tilde{x}_{k}\right|<\xi\sigma$, or equivalently 
\[
z\in\left[c_{k},d_{k}\right]\qquad c_{k}=\frac{-\text{sign}\left(c^{\top}\tilde{x}_{k}\right)\cdot\nu^{\top}\tilde{x}_{k}-\xi\sigma}{\left|c^{\top}\tilde{x}_{k}\right|}\qquad d_{k}=\frac{-\text{sign}\left(c^{\top}\tilde{x}_{k}\right)\cdot\nu^{\top}\tilde{x}_{k}+\xi\sigma}{\left|c^{\top}\tilde{x}_{k}\right|}.
\]
Similarly, $\left|t_{k}\right|>\xi\sigma$ if 
\[
z\in\left[-\infty,c_{k}\right]\cup\left[d_{k},\infty\right].
\]
For $\hat{H}(M)=H$, we require $\left|t_{k}\right|>\xi\sigma$ for all
$k\in H$ and $\left|t_{k}\right|<\xi\sigma$ for all $k\in M\setminus H$.
Thus we require 
\[
z\in\left(\underset{k\in H}{\bigcap}\left(\left[-\infty,c_{k}\right]\cup\left[d_{k},\infty\right]\right)\right)\cap\left(\underset{k\in M\setminus H}{\bigcap}\left[c_{k},d_{k}\right]\right)
\]
which simplifies to a union of intervals.  \end{proof}

\subsection{Prostate data revisited}

The right panel of Figure \ref{fig:prostate} shows the result of applying the stable-$t$ approach to the prostate data. Three high-value targets have been identified,  and for all targets, the  $\TZ_{\text{stab}-t}$  intervals are shorter than the $\TZ_{Ms}$ intervals of \citet{lee2016}.

\section{Simulations for the  partial regression target}\label{sec:sim_partial}

We test our proposals across a variety of simulation settings; we
focus on the case where the data informs both variable selection and
target formation (for examples involving only variable selection,
see Section \ref{subsec:full_sim}). 
\begin{itemize}
	\item We compare six methods: naive intervals, Bonferroni corrected, $\TZ_{Ms}$, $\TZ_{M},\TZ_{\text{stab}-t}$ and $\TZ_{\text{stab}-\ell_1}$. Recall that $\TZ_{Ms}$ condition on $\left(\hat{M},\hat{s}\right)$ while $\TZ_{M}$ conditions only on $\hat{M}$. 
	\item We consider four possibilities for $\left(n,p\right)$: $\left(100,50\right)$, $\left(100,125\right)$, $\left(100,250\right)$ and $\left(100,1250\right)$. 
	\item We generate the outcome as $y_{i}=x_{i}^{\top}\beta+\varepsilon_{i}$, $i=1,\ldots,n$. For the moment we let $x_{i}\sim \mathcal{N}\left(0,\mathbf{I}_{p}\right)$ and $\varepsilon_{i}\sim \mathcal{N}\left(0,1\right)$; however, results for the case of correlated predictors and non-normal errors are reported in Appendices \ref{app:correlated} and \ref{app:violated:assumptions}.
	\item The first $k=5$ components of $\beta$ are non-zero, the rest are set to 0. The non-null components are all set to either $\delta_{\text{low}}$ or $\delta$$_{\text{high}}$. For each $\left(n,p\right)$ configuration, we chose $\delta_{\text{low}}$ and $\delta_{\text{high}}$ based on the null distribution of $\max_{j}n^{-1}\left|x_{j}^{\top}y\right|$ following the rules detailed in Section \ref{subsec:full_sim}. 
	\item Following Section \ref{subsec:full_sim}, we consider two values of $\lambda$: one based on CV and one based on the universal threshold value of $\sqrt{\frac{2\log p}{n}}$. 
	\item For the stable-$\ell_{1}$ method, we set $\lambda_{\text{high}}=\sqrt{\frac{2\log p}{n}}$ when $\lambda$ is based on CV and $\lambda_{\text{high}}=1.25\lambda$ when $\lambda$ is already chosen as the universal threshold. 
	\item For the stable-$t$ method, we set $c=\left|\Phi^{-1}\left(\frac{\alpha}{2p}\right)\right|$ where $\alpha=0.1$ (chosen because we are building 90\% confidence
	intervals) 
\end{itemize}
The results for the cases $n=100,p=250$ and $n=100,p=1250$ are presented in Figures \ref{fig:partial_n100p250} and \ref{fig:partial_n100p1250} respectively (results for $n=100,p=50$ and $n=100,p=125$ are given in Appendix \ref{app:partial:target}); Figure \ref{fig:partial_n100p250_null} presents the null case ($\beta=0$). We summarize the key takeaways: 
\begin{enumerate}
	\item Bonferroni adjusted intervals have poor coverage (only 47\% for 90\% confidence interval) when the global null holds.
	\item The $\TZ_{M}$ method improves on the $\TZ_{Ms}$ intervals but as the simulations demonstrate they are often many times wider than the na\"ive or Bonferroni-corrected intervals. In practice, the inferences they deliver would not be useful. This suggests that we need to go beyond minimal conditioning (which in the case of partial targets is $\hat{M}$) and further restrict how data is used for target formation. 
	\item The stable-$\ell_{1}$ and especially the stable-$t$ method produce intervals that are often only slightly longer than the Bonferroni corrected intervals (and in high signal settings can be quite a bit shorter than Bonferroni). In addition, they tend to be more stable numerically, i.e., less chance of an ``infinite'' length interval.
	\item As we saw in Section \ref{subsec:full_sim}, a higher value of $\lambda$ (i.e., using the universal threshold rather than cross validation) helps to make $\hat{M}$ more stable and hence reduce the length of the $\TZ_{Ms}$ and $\TZ_{M}$ intervals. However, we see that in truly high dimensional problems such as when $n=100$ and $p=1250$, the stable-$t$ method is half the length of the $\TZ_{M}$ method
	even when $\lambda$ is selected to be the universal threshold. In addition, many practitioners prefer choosing $\lambda$ by cross-validation in which case use of the stable methods is a necessity if we want reasonable interval lengths. 
\end{enumerate}
In addition to the above results for independent predictors, we also
considered the case of correlated predictors. We explored block-equicorrelation and Toeplitz covariance structures for the predictors with the maximum pairwise correlation set at 0.5 (see Appendix \ref{app:correlated} for details and results).
What we observed was that the stable methods outperform the $\TZ_{M}$ and $\TZ_{Ms}$ intervals by an even larger margin in the presence of strongly correlated predictors.  Intuitively as we increase the correlation between predictors, the set of lasso active variables $\hat{M}$ becomes less stable (as we move in the direction $\eta_{j}$ capturing
the effect of $x_{j}$, we are likely to flip the active/inactive
status of variables correlated with $x_{j}$), hence increasing the
need for stabilization.

It has become dogma that we need to regularize parameter estimates
in high dimensions. Similarly, we hope that these simulations argue
clearly for the need to regularize/stabilize how we choose targets
for inference. It is usually not advisable to form a model using all
variables in the lasso active set even if it is technically possible
to account for such selection when doing inference.

\begin{figure}[hbtp]
	\begin{centering}
		\includegraphics[width=0.33\paperwidth]{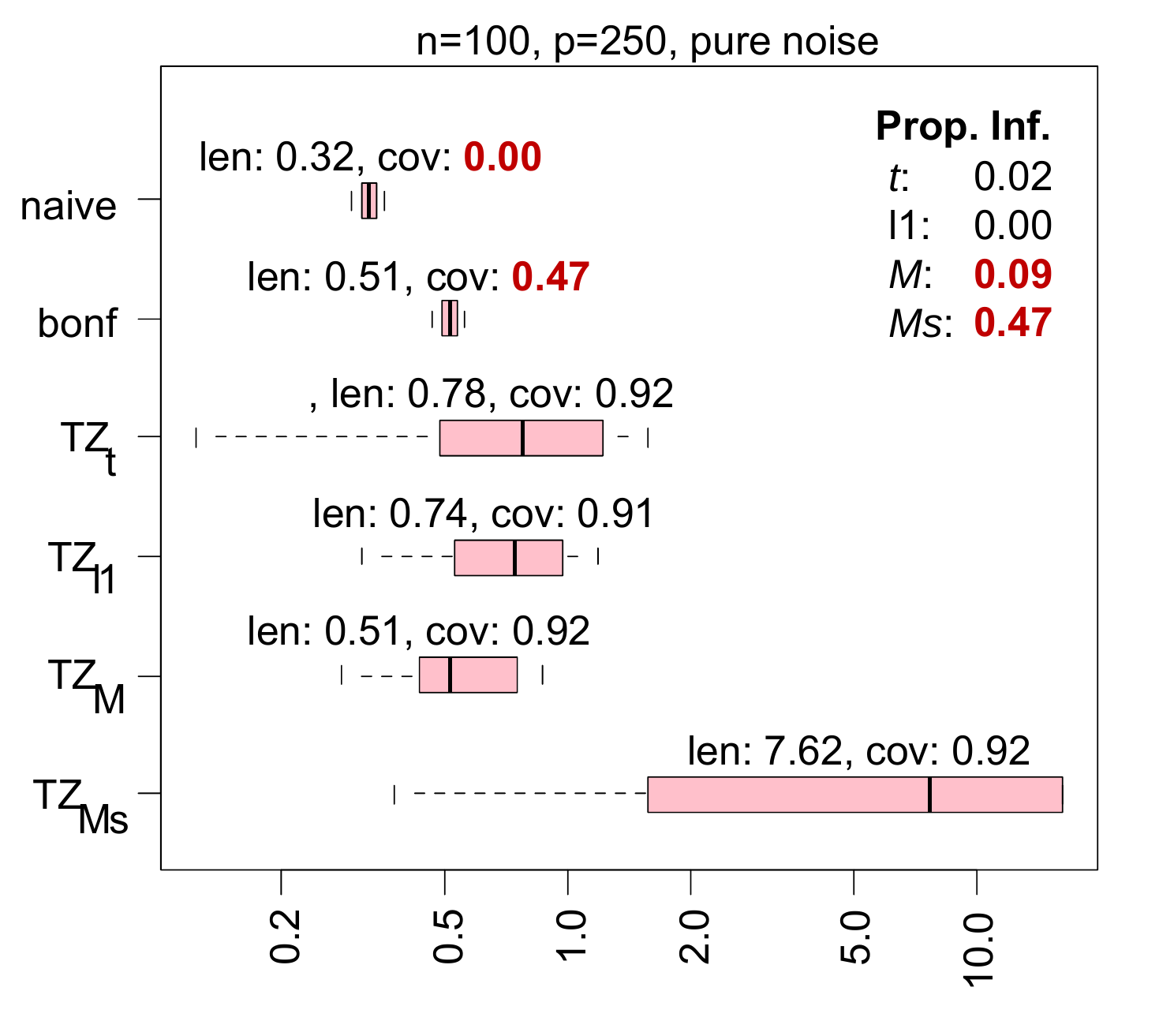}
		\par\end{centering}
	\caption{\em Boxplot of lengths of 90\% confidence intervals for ``partial''
		regression coefficients. Six interval methods are compared: naive
		(ignoring selection), Bonferroni adjusted, $\TZ_{\text{stab}-t}$, $\TZ_{\text{stab}-\ell_1}$, $\TZ_{M}$,
		and $\TZ_{Ms}$. Reported are the median interval length, the empirical
		coverage, and the proportion of ``infinite'' intervals (the infinite
		length results from numerical inaccuracies when inverting a truncated
		normal CDF); the boxplots set infinite lengths to the maximum finite
		observed length. Here, $n=100, p=250$ and there is no signal, $\beta=0$. The lasso penalty parameter $\lambda$ is set to $\sqrt{\frac{2\log(p)}{n}}$.}
		\label{fig:partial_n100p250_null}
	
\end{figure}

\begin{figure}[hbtp]
	\begin{centering}
		\includegraphics[width=0.33\paperwidth]{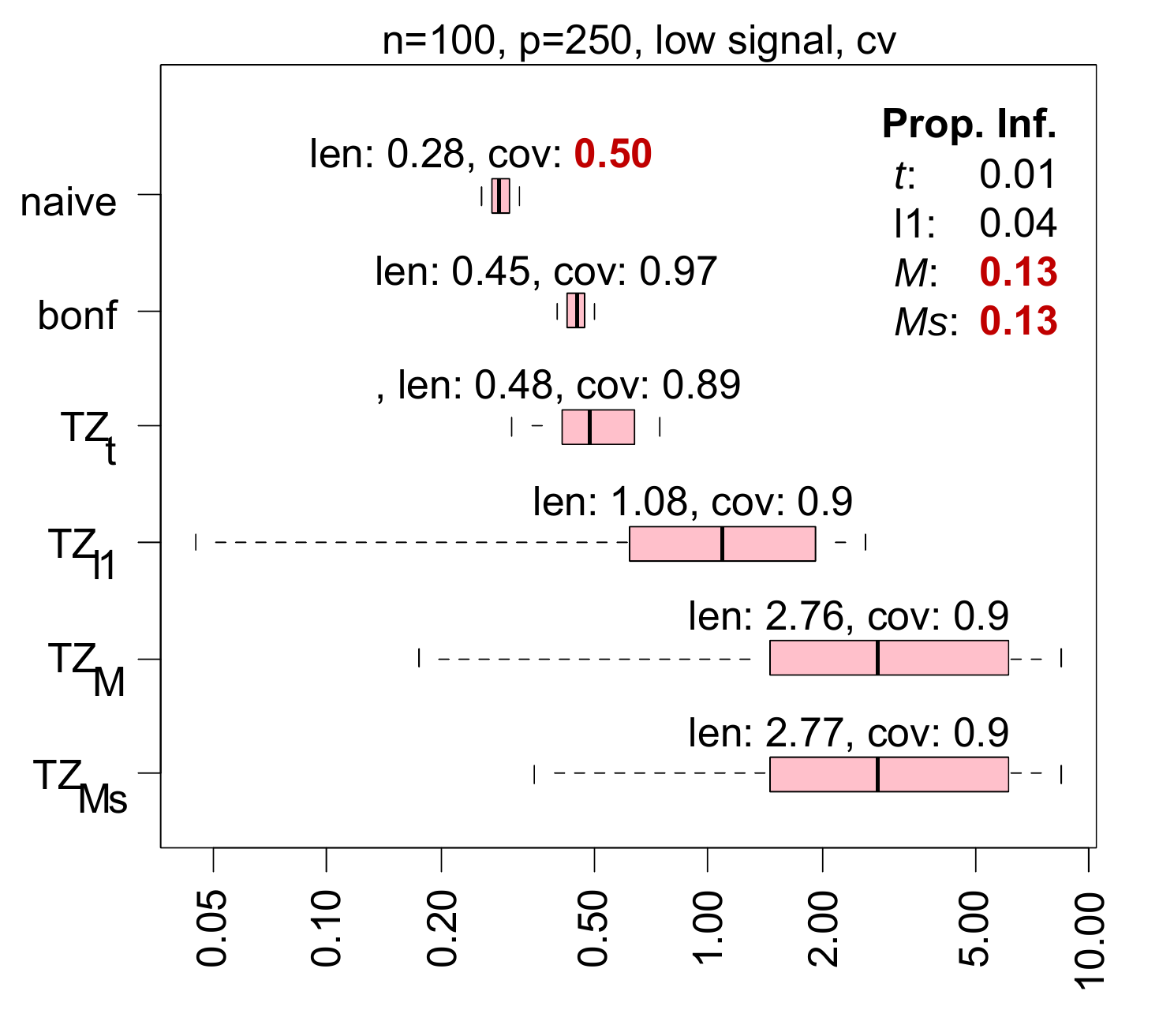}\includegraphics[width=0.33\paperwidth]{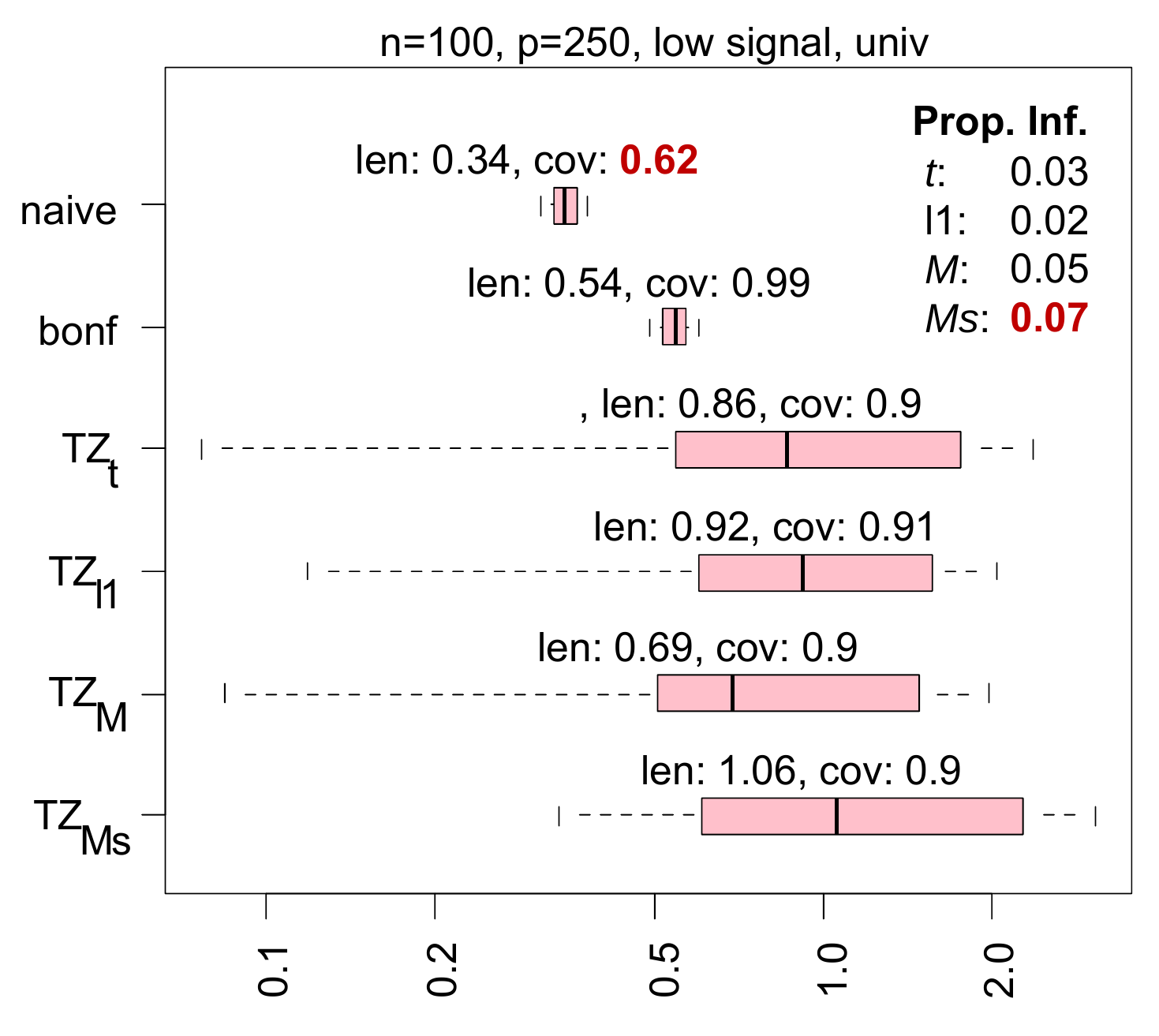} 
		\par\end{centering}
	\begin{centering}
		\includegraphics[width=0.33\paperwidth]{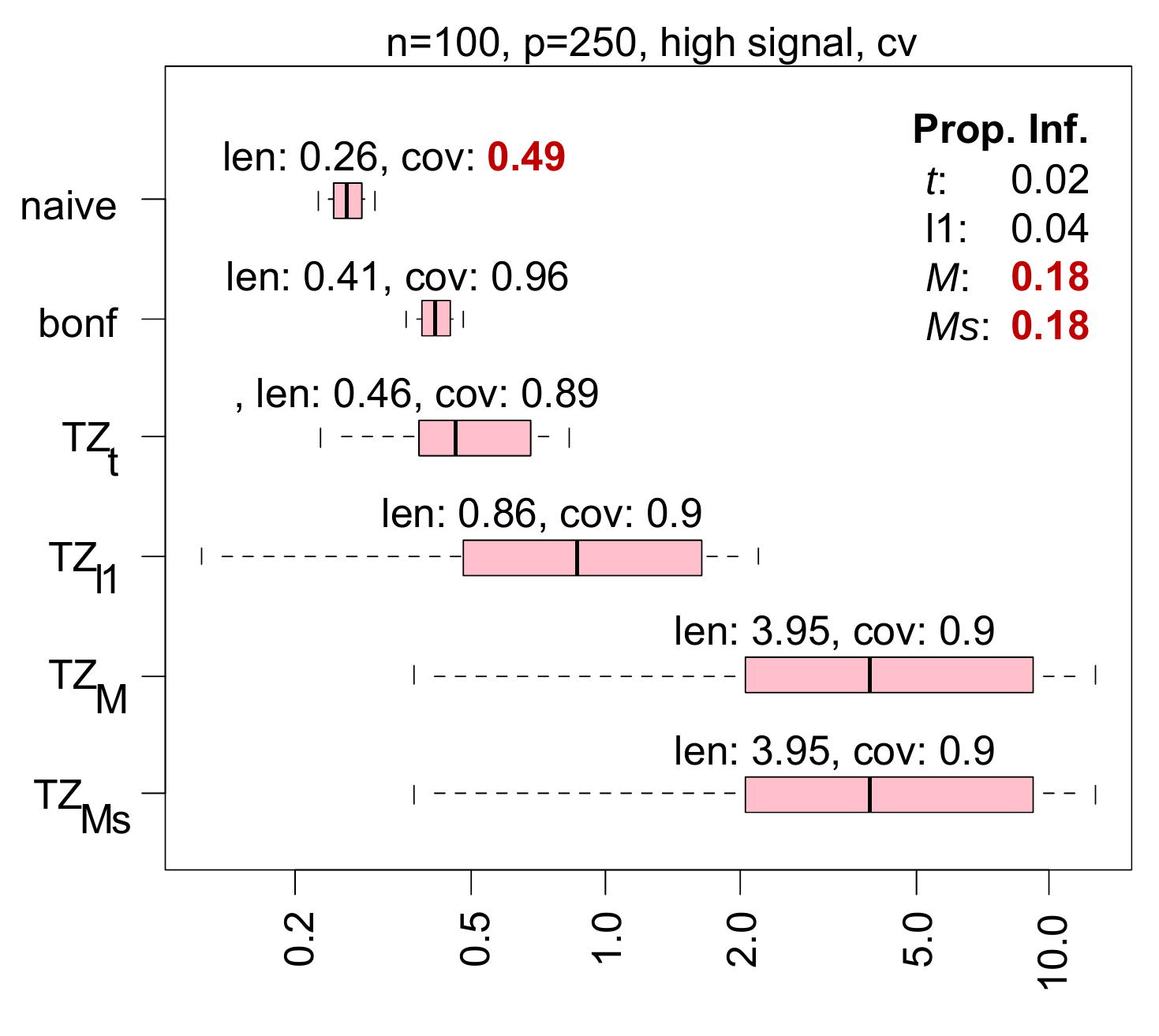}\includegraphics[width=0.33\paperwidth]{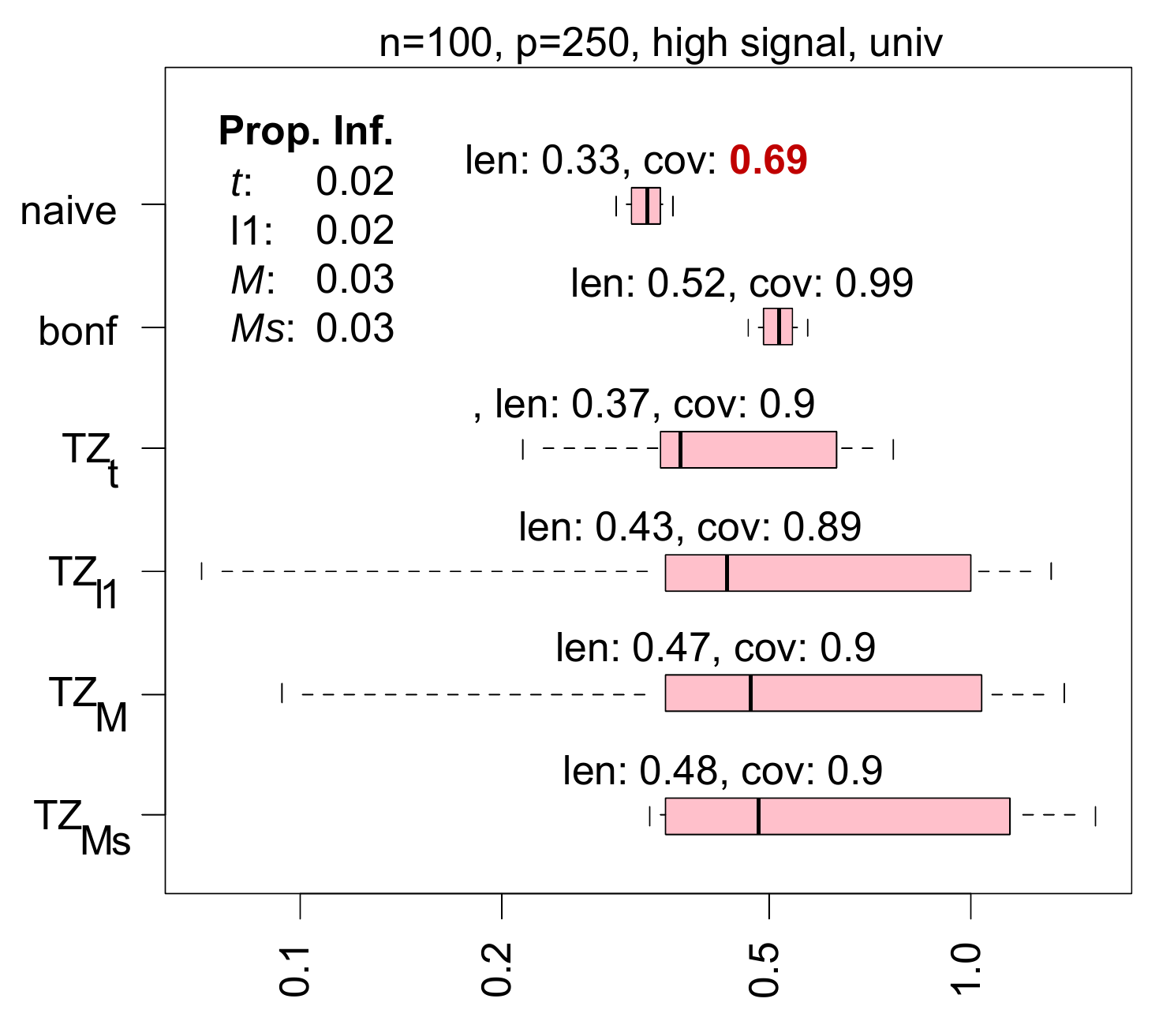} 
		\par\end{centering}
	
	\caption{\em Boxplot of lengths of 90\% confidence intervals for ``partial'' regression coefficients. Six interval methods are compared: naive (ignoring selection), Bonferroni adjusted, $\TZ_{\text{stab}-t}$, $\TZ_{\text{stab}-\ell_1}$, $\TZ_{M}$,
		and $\TZ_{Ms}$. Reported are the median interval length, the empirical coverage, and the proportion of ``infinite'' intervals (the infinite length results from numerical inaccuracies when inverting a truncated normal CDF); the boxplots set infinite lengths to the maximum finite
		observed length. Here $n=100,p=250$. The first five components of
		$\beta$ are set to $\delta_{\text{low}}=0.29$ (top panels) or $\delta_{\text{high}}=0.68$ (bottom panels) and the remaining components are 0. The lasso penalty is either set at the universal threshold value $\sqrt{\frac{2\log p}{n}}\approx0.33$ (right panels) or at a value approximating the behavior of 10-fold cross validation (0.18 and 0.14 respectively for the low and high signal cases).}
		\label{fig:partial_n100p250}
\end{figure}

\begin{figure}[hbtp]
	\begin{centering}
		\includegraphics[width=0.33\paperwidth]{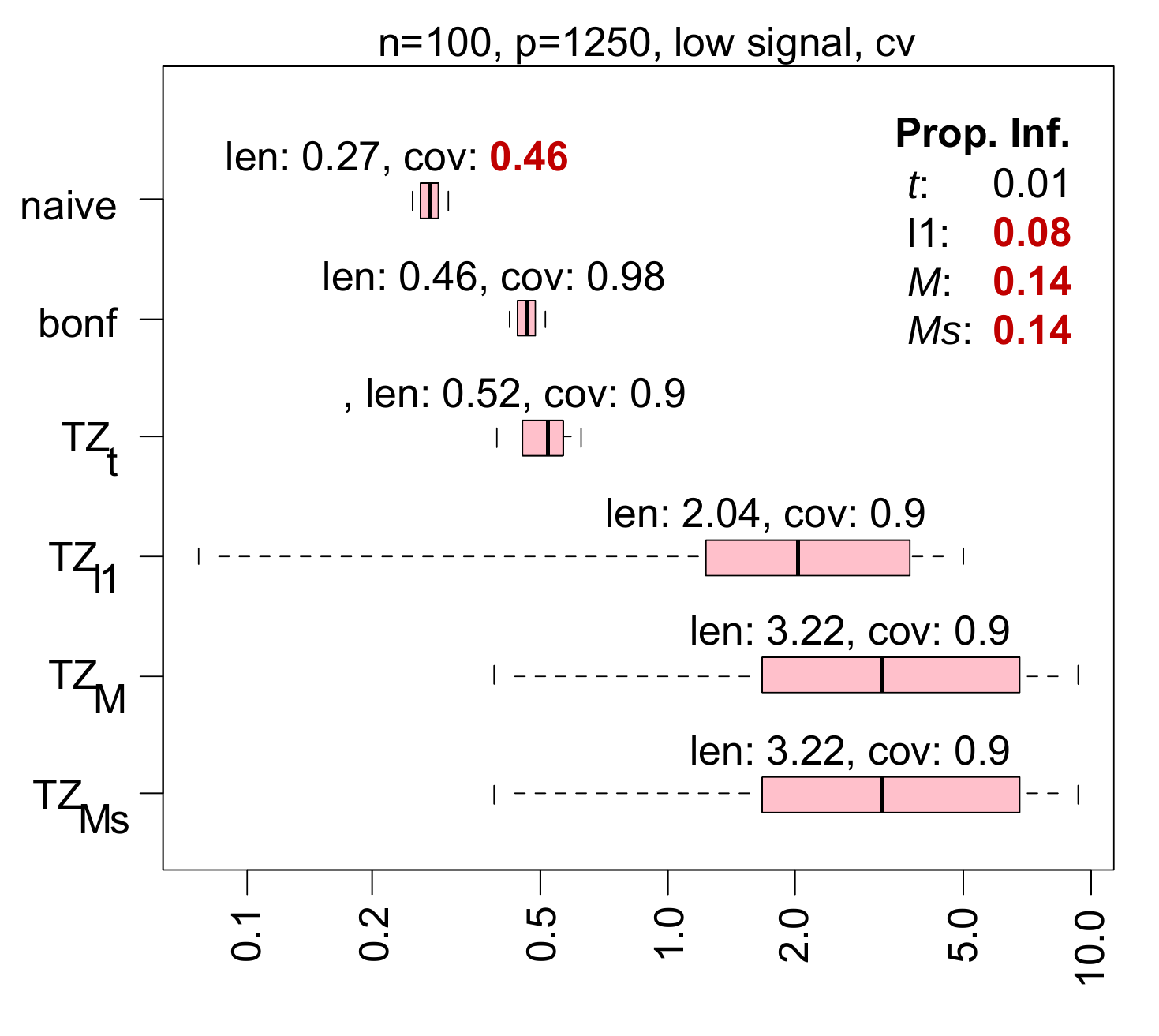}\includegraphics[width=0.33\paperwidth]{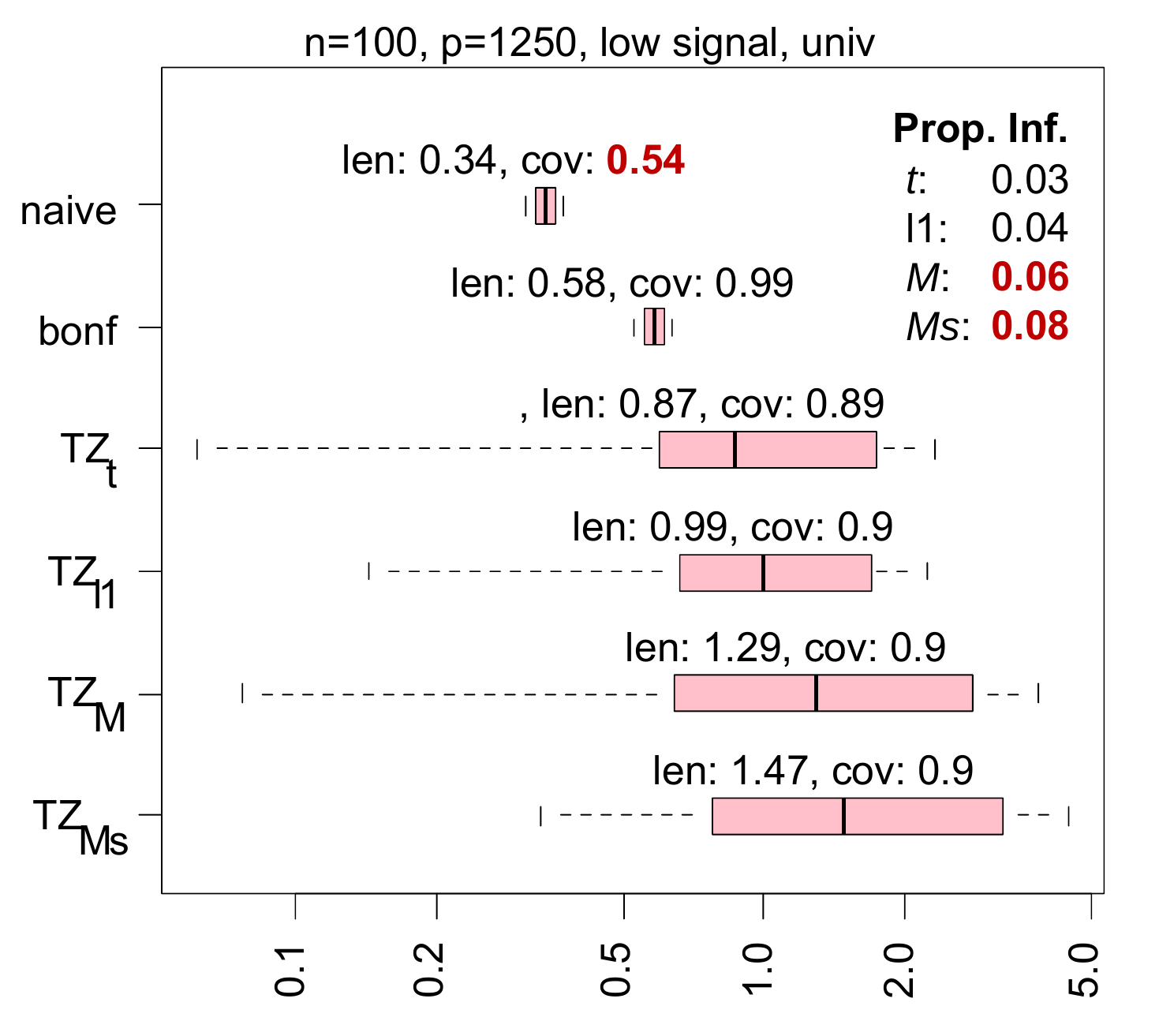} 
		\par\end{centering}
	\begin{centering}
		\includegraphics[width=0.33\paperwidth]{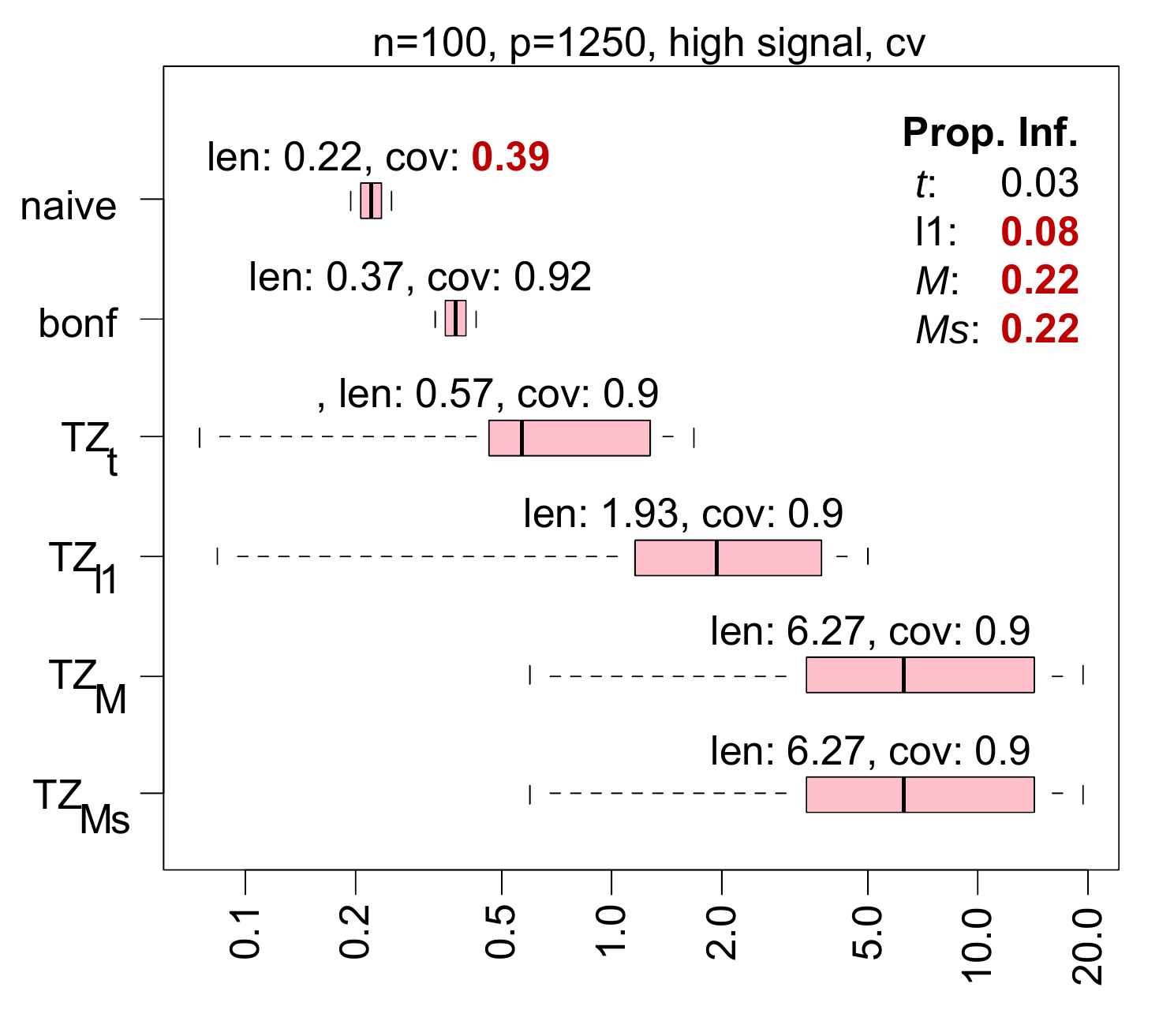}\includegraphics[width=0.33\paperwidth]{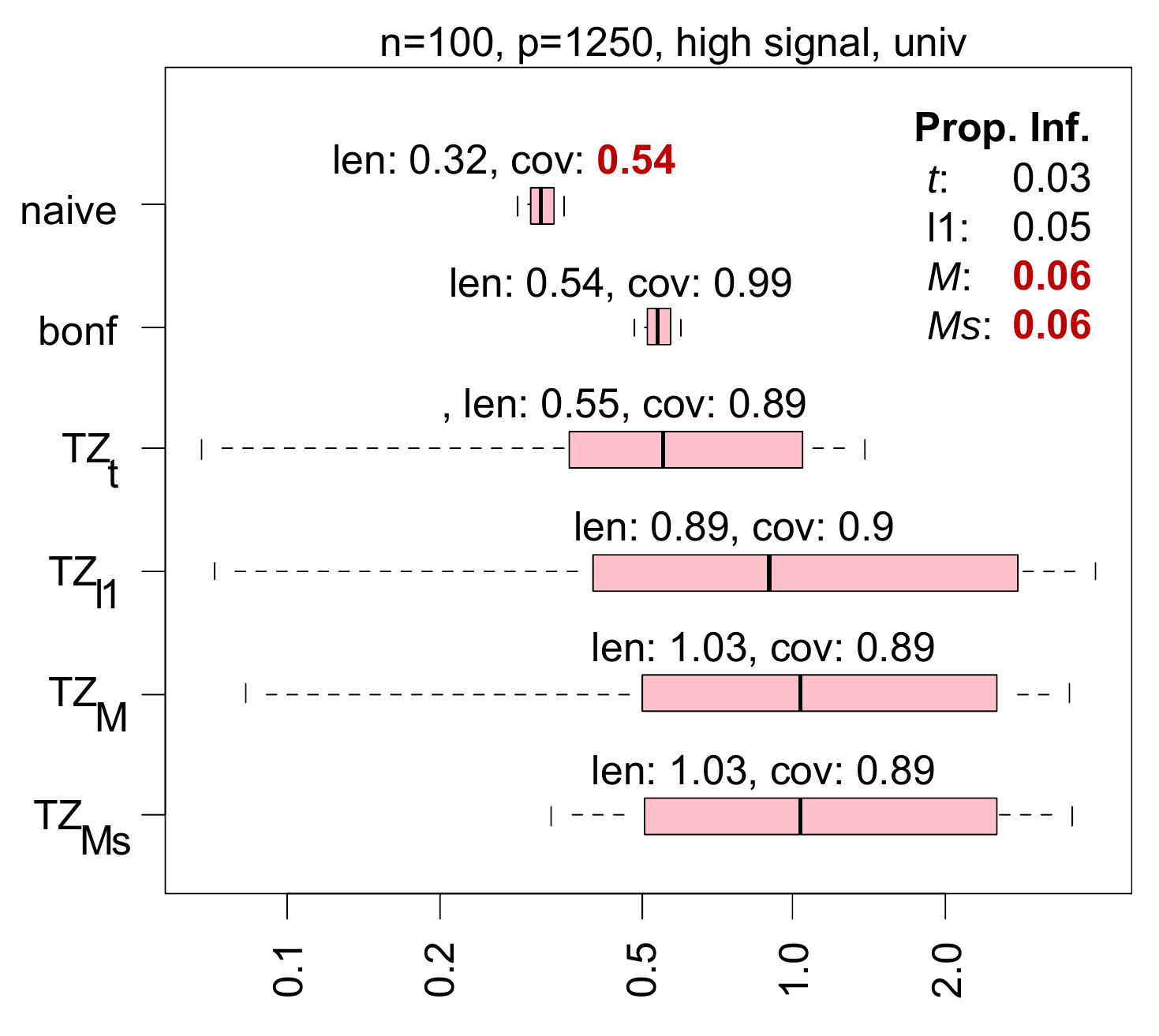} 
		\par\end{centering}
	\caption{\em Boxplot of lengths of 90\% confidence intervals for ``partial'' regression coefficients. Six interval methods are compared: naive
		(ignoring selection), Bonferroni adjusted, $\TZ_{\text{stab}-t}$, $\TZ_{\text{stab}-\ell_1}$, $\TZ_{M}$, and $\TZ_{Ms}$. Reported are the median interval length, the empirical coverage, and the proportion of ``infinite'' intervals (the infinite length results from numerical inaccuracies when inverting a truncated normal CDF); the boxplots set infinite lengths to the maximum finite
		observed length. Here $n=100,p=1250$. The first five components of $\beta$ are set to $\delta_{\text{low}}=0.34$ (top panels) or $\delta_{\text{high}}=0.69$ (bottom panels) and the remaining components are 0. The lasso penalty is either set at the universal threshold value $\sqrt{\frac{2\log p}{n}}\approx 0.38$ (right panels) or at a value approximating the behavior of 10-fold cross validation (0.27 and 0.19 respectively for the low and high signal cases). }
	\label{fig:partial_n100p1250}
	
\end{figure}

\subsection{Performance When Assumptions Are Violated}

In Appendix \ref{app:violated:assumptions}, we investigate whether the good performance of the proposed methods are robust to various violations of our assumptions. We consider non-normality of the errors, the effect of having to estimate
$\sigma^2$, and the effect of a data driven choice of $\lambda$.
We find that the proposed TZ intervals are robust to various model
violations, and the lengths of the stable intervals seem less likely
to be negatively impacted by model violations than the lengths of
the $\TZ_{M}$ and $\TZ_{Ms}$ intervals.

\section{Predicting HIV Drug Resistance} \label{sec:realdata} 

The target of HIV drugs often becomes resistant through mutation. For six nucleotide reverse transcriptase inhibitors
(NRTIs) used to treat HIV, \citet{rhee2003} investigate which mutations are most predictive of drug resistance. We focus on the drug Lamivudine (3TC). The primary response is log susceptibility to drug. The explanatory variables are binary indicators for whether a mutation has occurred at a particular site. The data have $n=1057$ samples and $p=210$ sites.

We fit a lasso model with $\lambda=$0.0124 (chosen via 10-fold cross
validation), yielding 36 active mutations. One of the mutations, p184, had a particularly large effect: its OLS estimate from a regression on all the variables is 0.477 with the accompanying naive 90\% confidence interval being (0.471, 0.485). TZ$_{V}$, TZ$_{M}$ and TZ$_{Ms}$ all produced the same selection adjusted interval (0.437, $\infty$);
the infinite upper bound here is due to numerical inaccuracy from
inverting a truncated Gaussian pivot. When considering the partial
coefficient for p184, the results were qualitatively identical (the
truncated-Z intervals still had an infinite upper bound). The effect
estimates for the remaining 35 active variables are an order of magnitude smaller and their 90\% confidence intervals are plotted in Figure \ref{fig:hiv}. Here are the key points:
\begin{itemize}
	\item The results are consistent with the analysis of the prostate data
	(see Figure \ref{fig:prostate}). The TZ$_{V}$ and TZ$_{\text{stab}-t}$
	intervals tend to be roughly half as long as the TZ$_{M}$ intervals
	and only a bit longer than the naive intervals. TZ$_{M}$ and TZ$_{Ms}$
	frequently produce very long intervals.
	\item Using the stable-$t$ criterion (with threshold set to $\left|\Phi^{-1}\left(\frac{0.1}{2p}\right)\right|=3.5$),
	we selected 6 high value targets: p65, p67, p69, p90, p184 (not pictured),
	and p215. With the exception of p184 discussed above, the TZ$_{V}$
	and TZ$_{\text{stab}-t}$ methods yielded very stable intervals for
	the effects of these high value targets; despite the strong signal
	from these 6 mutations, TZ$_{M}$ and TZ$_{Ms}$ yielded very imprecise
	intervals for p67 (interval contains 0), p215, and p69 (in the case
	of full).
	\item The target that the TZ$_{\text{stab}-t}$ intervals attempt to cover
	is different from the partial regression coefficients (defined with
	respect to active set); see Section \ref{subsec:reg_target}. But
	as the bottom panel of Figure \ref{fig:hiv} shows, they are very
	similar in practice. This is because, as we have discussed above, the
	two targets differ significantly only if the low value targets have
	large coefficients (and are correlated with the high value targets).
	\item TZ$_{\text{stab}-t}$ will sometimes yield longer intervals than TZ$_{M}$
	and TZ$_{Ms}$ (e.g., p215) but this occurs very rarely.
\end{itemize}

\begin{figure}[pbth]
	\begin{centering}
		\includegraphics[width=0.7\paperwidth]{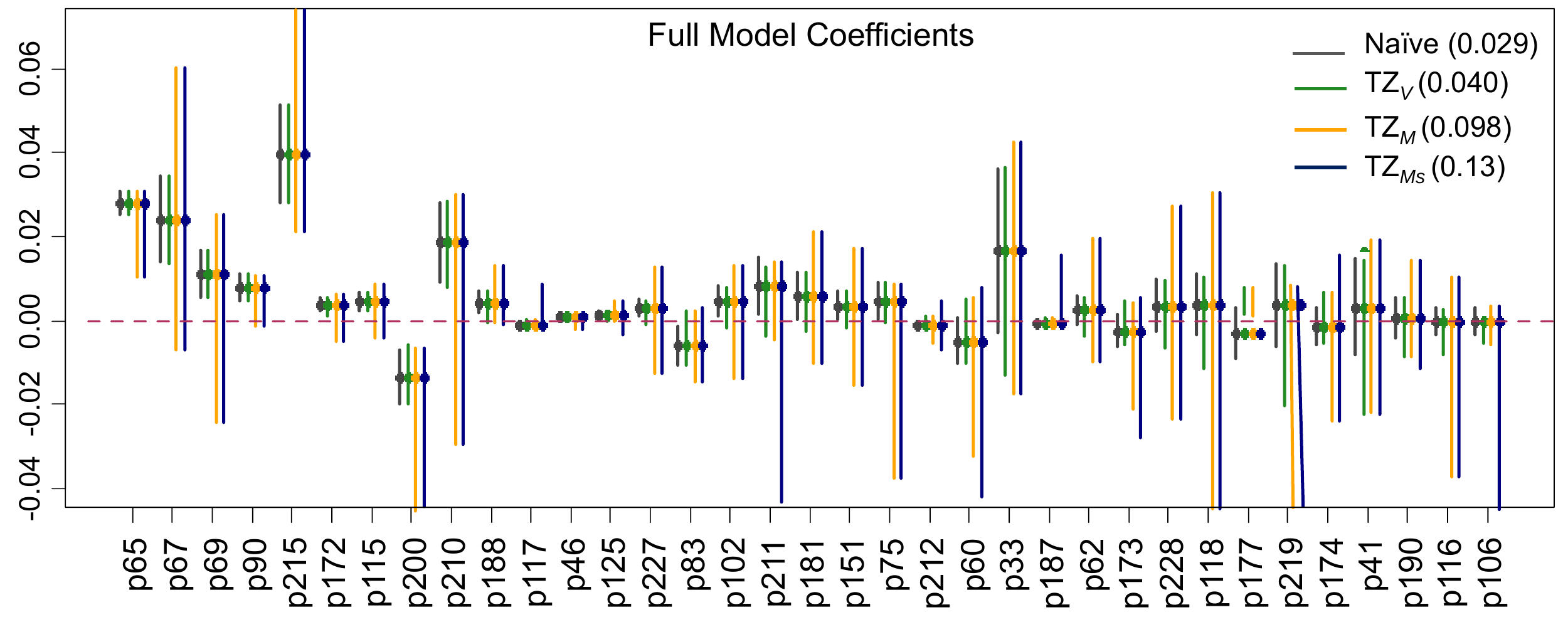}
		\par\end{centering}
	\begin{centering}
		\includegraphics[width=0.7\paperwidth]{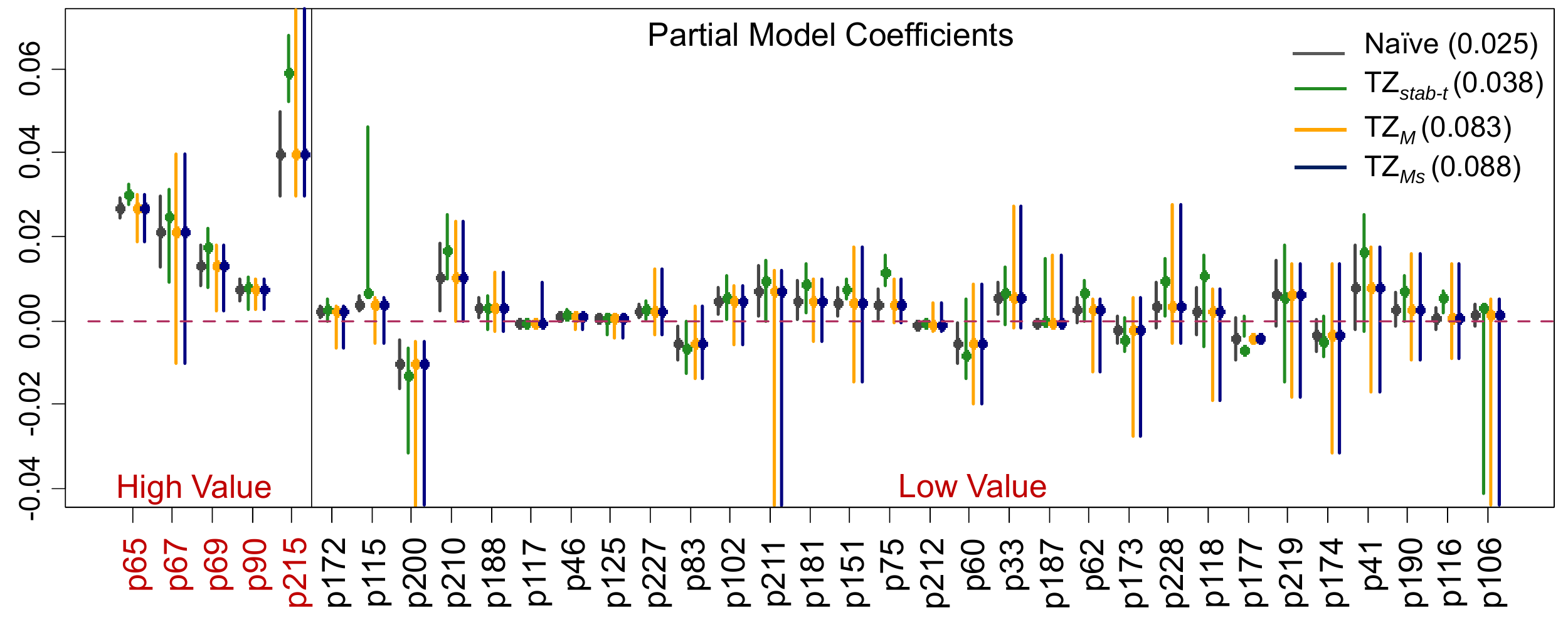}
		\par\end{centering}
	\caption{\em Regression analysis of log susceptibility to Lamivudine on indicators
		for 210 inhibitor mutations ($n=1057).$ Plotted are 90\% confidence
		intervals for 35 of the 36 mutations chosen by lasso ($\lambda=0.0124$
		from 10-fold cross validation); interval for p184 was excluded because
		its effect is an order of magnitude larger. The noise standard deviation
		$\sigma$ was assumed known and fixed at its OLS estimated value (0.23).
		6 high value targets were identified (including p184 which is not
		pictured); their Z-statistic from regression on active variables exceeded $\Phi^{-1}\left(1 - \frac{0.1}{2p}\right)=3.5$.
		The median interval length (across the 36 mutations) for each method
		is given in parentheses. The
		depicted point estimates (circles) are the MLEs of the corresponding population quantities.}
	\label{fig:hiv}
\end{figure}

\section{A General Recipe} \label{sec:general}

For concreteness, this paper has focused on variable selection and
target formation using the lasso. However, many of the ideas are applicable
to any model selection rule with polyhedral selection regions or even
non-polyhedral selection regions. What changes in the general case
is that the computation cost of the truncation region for the Z-statistic
can be much greater. Let us begin with the ingredients that do not
change:
\begin{itemize}
	\item We are interested in inference on linear functions $\mu^{\top}\eta_{j}^{\hat{H}}$ of the mean $\mu$ for variables $j\in\hat{M}$. Here $\hat{M}$ is the set of variables we are interested in, while $\hat{H}$ is (possibly) another set of variables used to define our inferential target (through $\eta^{\hat{H}}$). Above, we saw three cases (i) $\hat{H}=\left\{ 1,\ldots,p\right\}$ and $\hat{M}$ is the lasso active set (full targets) (ii) $\hat{H}=\hat{M}$ is the lasso active set (partial targets) and (iii) $\hat{M}$ is the lasso active set and $\hat{H}$ is the subset of $\hat{M}$ with large $t$-statistics (stable-$t$).
	\item For $\eta=\eta_{j}^{\hat{H}}$, we condition our inferences
	on $P_{\eta^{\perp}}y=\left(\mathbf{I}_n-\frac{\eta\eta^{\top}}{\lVert\eta\rVert_{2}^{2}}\right)y$; see \citet{fithian2014} for ways to remove this conditioning. That is, we first restrict our data to the line 
	\[
	\mathcal{L}=\left\{ y':y'=P_{\eta^{\perp}}y+z\cdot\frac{\eta}{\lVert\eta\rVert_{2}^2}, z\in\mathbb{R}\right\} .
	\]
	This line represents perturbations of the data in the direction $\eta$ only. The line is parametrized by $z$ and note that for $y'=y$, we have $z=y^{\top}\eta\sim \mathcal{N}\left(\eta^{\top}\mu,\sigma^{2}\|\eta\|_2^2\right)$ is the Z-statistic (up to standardization of variance). 
	\item For $\eta=\eta_{j}^{\hat{H}}$, we further condition
	on $j\in\hat{M}$ and $\hat{H}=H$ to adjust for the costs variable
	selection and target formation respectively. The effect of this conditioning
	is to restrict the line $\mathcal{L}$ to a range of $z$-values (i.e.,
	we obtain a truncated Z-statistic); we can necessarily express the
	intersection of \emph{$\mathcal{L}$} with the sets $j\in\hat{M}$
	and $\hat{H}=H$ as a union of intervals (of $z$-values).
\end{itemize}
The recipe above is quite general: we made no restrictions on what
kind of procedure can be used to construct $\hat{M}$ or $\hat{H}$.
Of course, for complicated $\hat{M}$ and $\hat{H}$, finding the
range of values that the Z-statistic has been truncated to can be
difficult (if not impossible). We have seen two particularly easy
to compute situations:
\begin{itemize}
	\item TZ$_{Ms}$ intervals: Letting $\hat{M}=\hat{H}$ be the lasso active set and also conditioning on the sign of the active variables, we obtain a polyhedral region, which when intersected with the line $\mathcal{L}$
	results in an interval. The power of using ``polyhedral'' selection rules is that it greatly simplifies computation for the truncation interval: in particular, once we have expressed the selection region as $\left\{ \mathbf{A}y\leq b\right\}$, formulas in Appendix \ref{app:truncation:points} give the interval endpoints. Another common polyhedral selection rule is forward stepwise regression. If $\hat{M}$ is constructed using $K$ steps of forward stepwise regression, then the set $\left\{\hat{M}=M,\hat{s}=s\right\} =\left\{\mathbf{A}y\geq 0\right\}$ where $\hat{s}$ are the signs of the $K$-selected variables; see \citet{taylor2014forwardstepwise} for the formula for $\mathbf{A}$. 
	\item Full target: When $\hat{M}$ is the lasso active set and $H=\left\{1,\ldots,p\right\}$, then the set of allowable $z$-values is of the form $\left(-\infty,a\right)\cup\left(b,\infty\right)$ where $a$ and $b$ have simple formulas (see Proposition \ref{prop:full}). 
\end{itemize}
Outside these special cases, simple formulas usually do not exist
for characterizing the intersection of $\mathcal{L}$ with $\left\{ j\in\hat{M}\right\} $
and $\left\{ \hat{H}=H\right\} $. Nonetheless, one can usually run
a pathwise algorithm to determine the truncation region. Let $\mathcal{Z}=\left\{ z_{1},\ldots,z_{G}\right\} $
be a grid of $z$-values (which map to points on $\mathcal{L}$).
For each $z_{k}$, we compute $\hat{M}\left(z_{k}\right)$ and $\hat{H}\left(z_{k}\right)$.
This yields a subset $\mathcal{Z}_{\text{trunc}}\subset\mathcal{Z}$
for which $j\in\hat{M}\left(z_{k}\right)$ and $\hat{H}\left(z_{k}\right)=H$.
We then approximate $\mathcal{Z}_{\text{trunc}}$ using a union of
intervals and this becomes our truncation set. We note that:
\begin{itemize}
	\item For algorithms such as the lasso, we can take advantage of warm starts.
	That is, having obtained the lasso solution for $z_{k}$, we can use it as an initial starting point for $z_{k+1}$. In fact, for the lasso, warm starts are not really necessary, because one can examine the constraint $\left\{ \mathbf{A}y\leq b\right\} $ characterizing the selection region for the lasso to determine which variables will enter/leave
	the active set at the endpoints of any particular interval of $z$-values characterized by a $\left(\hat{M},\hat{s}\right)$ pair.
	\item Very positive and very negative $z$-values do not affect inference. Hence, it is usually reasonable to consider only grid values less than
	say $20\sigma$.
\end{itemize}
The above recipe in fact allows us to adjust for selection using any blackbox model selection procedure\textemdash all that is needed is that we be able to run this blackbox procedure for all $z$-values in some grid. Again, the limiting factor here is computation: how quickly we can re-evaluate our model selection procedure on all grid values.

\section{Further topics}

Scientific research in the modern era has become much less ``pre-specified'' and much more dependent on exploratory data analysis to generate interesting hypotheses and models. Accounting for this exploration is crucial
for producing replicable results. Unfortunately, the simplest selection adjusted intervals, as in \citet{lee2016} and \citet{TTLT2016}, are often too long to be useful\textemdash they ``over-condition,'' that is they over-state the amount of information used up for hypothesis generation. A more refined accounting of costs reveals two primary ways through which the data guides many investigations:
\begin{itemize}
	\item Selecting interesting variables for further investigation.
	\item Choosing a ``parametrization'' of the scientific problem so that, for example, the effect of a variable can now be captured by some coefficient in a chosen model (target formation).
\end{itemize}
Decoupling these two activities is crucial because the inferential
cost incurred by target formation is much larger than that incurred by variable selection. For example, in the prostate and HIV datasets (as well as across our simulations), we saw that the TZ$_{V}$ intervals (adjusts for cost of variable selection) are much shorter than the TZ$_{M}$ intervals (adjusts for both variable selection and target formation). 

There is of course an extreme response to this asymmetric cost scheme: fix your targets beforehand and use the data only for variable selection. This, however, is too limiting. The appropriate response is to exert greater control over the cost of target formation. This can be done, for example, by splitting the data in two and choosing our targets
based on only half the data; \citet{TT2015} develop more sophisticated/efficient randomization approaches to controlling the cost. These approaches,
however, simultaneously use less information to choose our variables and less information to form our targets. But given that the cost of variable selection is not very high, it surely makes sense to use all the data to choose our important variables (hence arriving at a better set of ``important'' variables) and decrease only the cost of target formation. Rather than using a single lever to control the
cost of data exploration, an innovation of this paper is to use to
introduce two levers. Concretely, we found that selecting variables using the lasso but forming inferential targets based on a more stable set of variables, yielded good practical results \textemdash ``good'' here means intervals that are slightly longer than the naive intervals
but significantly shorter than the TZ$_{M}$ intervals (frequently,
less than half the length). 

Our results have relied on several assumptions, e.g., $\lambda$ is
fixed. In the regime where $p$ is fixed and $n$ is growing, these
assumptions can be dispensed with:
\begin{itemize}
	\item \citet{tibshirani2015uniform, tian2017asymptotics} show that the truncated normal pivot
	underlying the TZ$_{Ms}$ method holds asymptotically for a wide class
	of (non-normal) error distributions. The arguments therein continue
	to hold for the other truncated-Z approaches in this paper.
	\item \citet{TT2015} and \citet{markovic2017} show that if an estimator
	of $\sigma^{2}$ is consistent unconditionally, then it remains consistent
	conditional on selection (by the lasso). In the $p$ fixed regime,
	this implies that the pairs bootstrap can be used to consistently
	estimate various covariances.
	\item $\lambda$ can be chosen via randomized cross validation. The details are given in \citet{markovic2017}; the key is that we continue to
	enjoy (asymptotically) a joint multivariate normality between (i)
	the test statistic, (ii) the statistics characterizing the lasso active set and (iii) the randomized cross validation errors (over a grid of $\lambda$ values). The multivariate normality then allows us to derive the conditional distribution of the test statistic.
\end{itemize}

When $p>n$, our theoretical understanding is limited, but empirically,
the truncated-Z intervals seem to work reasonably well. 
Generalizations of this approach to (regularized- ) generalized linear models and the Cox PH model would also be interesting and
useful. The new procedures presented in this paper will be implemented in our \texttt{selectiveInference} package in the CRAN repository. 

\medskip

{\bf Acknowledgements}: The authors would like to thank Jonathan Taylor, who has been the intellectual leader in the post-selection inference effort at Stanford.

\appendix
\section*{Appendices}
\addcontentsline{toc}{section}{Appendices}
\renewcommand{\thesubsection}{\Alph{subsection}}

\subsection{Formulae for the  truncation points from the polyhedral lemma} \label{app:truncation:points}

Suppose that a selection region can be characterized by $\{\mathbf{A}y \leq b\}$. Let us decompose $y$ as:
\[
y=z\cdot c+\nu\qquad z=\eta^{\top}y\qquad c=\frac{\eta}{\lVert\eta\rVert_{2}^{2}}.
\]
For fixed $\mathbf{A},b$ and $\nu$, $z$ is truncated to
\[
\Nu^-(y) \leq z \leq \Nu^+(y)
\]
\begin{equation} \label{eq:amazing}
  \begin{aligned}
\Nu^-(y)&=\max_{j:(\mathbf{A}c)_j<0} \frac{b_j-(\bA \nu)_j}{(\mathbf{A}c)_j}\\
\Nu^+(y)&=\min_{j:(\mathbf{A}c)_j > 0} \frac{b_j-(\bA \nu)_j}{(\mathbf{A}c)_j}.
\end{aligned}
\end{equation}
In addition, for $\nu$ to be compatible with $\mathbf{A}$ and $b$, we require $\Nu^0(y) \geq 0$ where
\[
\Nu^0(y) =\min_{j:(\mathbf{A}c)_j = 0} (b_j-(\bA \nu)_j).
\]
The derivation and additional details can be found in Lemma 5.1 of \citet{lee2016}.

\subsection{A simple problem illustrating the conditioning events in detail} \label{app:simple:example:conditioning:events}

Here we  show the precise effects of different conditioning strategies in a simple example.  
Suppose $n=2$, $p=2$, and $\mathbf{X}=\left(\begin{array}{cc}
1 & \rho\\
0 & \sqrt{1-\rho^{2}}
\end{array}\right)$. Note that each column of $\mathbf{X}$ has unit norm and that $x_{1}^{\top}x_{2}=\rho$
where $x_{j}$ denotes the $j$th column. The solid pink
lines in Figure \ref{fig:TZ_int} denote the boundaries of the different
lasso selection regions in $y$-space (i.e., $\mathbb{R}^{2}$).
For example, if $y$ falls in the central parallelogram region
(whose boundaries are given by $x_{1}^{\top}y=\pm\lambda$
 and $x_{2}^{\top}y=\pm\lambda$), then neither variable will be active in the lasso solution. Consider:
\begin{itemize}
	\item Suppose $y$ equals the red dot in Figure \ref{fig:TZ_int}
	labeled ``Case 1.'' Then the lasso active set is $\hat{M}=\left\{ 1\right\} $
	and the sign of the active variable is $\hat{s}=-1$. As a result of
	this initial exploration, we settle on testing the hypothesis $H_{0}:\beta^{\left(\hat{M}\right)}=\mu_{1}=0$.
	Here $\eta=\left(1,0\right)^{\top}$ and $P_{\eta^{\perp}}y=y_{2}$.
	Assuming $\sigma=1$, our Z-statistic is simply $y_{1}$. If we condition
	on (i) $\hat{M}=\left\{ 1\right\} $, (ii) $\hat{s}=-1$ and (iii) $y_{2}$,
	then we know that $y_{1}$ must lie in $(-\infty,-\lambda)$. So the
	usual $\mathcal{N}\left(0,1\right)$ distribution of the Z-statistic is now
	$\mathcal{N}\left(0,1\right)$ truncated to $\left(-\infty,-\lambda\right)$.
	\item As discussed above, we conditioned on $\hat{s}$ to make the
	math simpler. If we condition only on $\hat{M}=\left\{ 1\right\} $
	and $y_{2}$, we still get that $y_{1}\in\left(-\infty,-\lambda\right)$.
	However, if $y$ were as in Case 2 of Figure \ref{fig:TZ_int},
	conditioning on both $\hat{M}$ and $\hat{s}$ leads to $y_{1}\in\left(\lambda,\infty\right)$
	while conditioning only on $\hat{M}$ yields $y_{1}\in\left(-\infty,-\lambda\right)\cup\left(\lambda,\infty\right)$,
	a union of intervals. That is, the less we condition on, the less
	severe the truncation. 
	\item In Section \ref{sec:full} we found it useful to condition on even less than just $\hat{M}$.
	In particular, in some situations, we may simply want to condition
	just on the fact that variable 1 was selected, $1\in\hat{M}$. In
	such a case $\hat{M}$ can equal $\left\{ 1\right\} $ or $\left\{1,2\right\} $.
	In Case 1, conditioning on $1\in\hat{M}$ and $y_{2}$ restricts $y_{1}$
	to $\left(-\infty,-\lambda\right)\cup\left(\lambda,\infty\right)$. 
\end{itemize}
\begin{figure}[hbtp]
	\begin{centering}
		\includegraphics[width=0.45\linewidth]{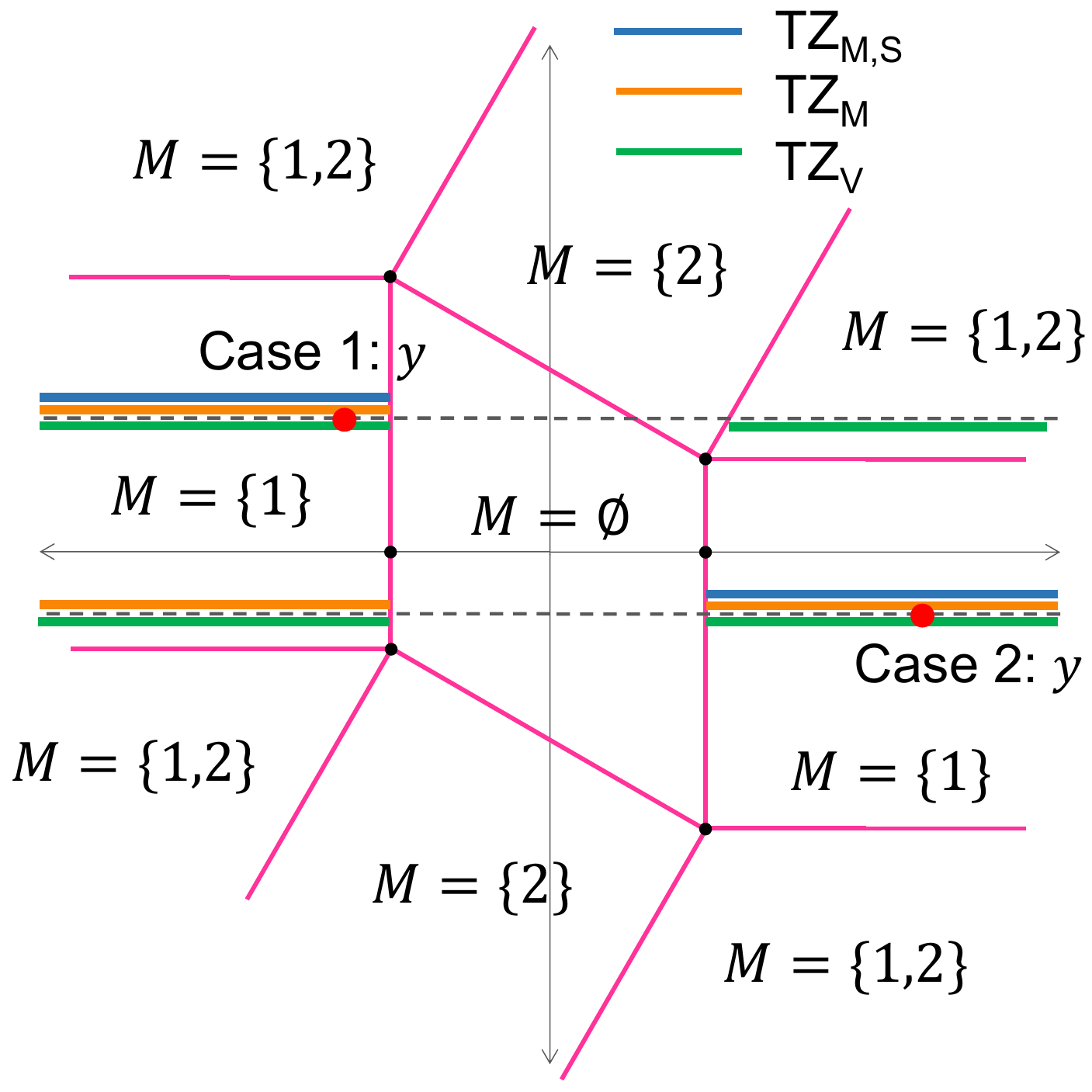}
		\par\end{centering}
	\caption{\em Plot of lasso selection regions in $y$ space when
		the first column of \textbf{$\mathbf{X}$} is $x_{1}=\left(1,0\right)^{\top}$ and the second column is $x_{2}=\left(\rho,\sqrt{1-\rho^{2}}\right)^{\top}$. For each hypothetical value of $y$, the blue line shows	the allowable range of $y_{1}$ (given fixed $y_{2}$) for $\hat{M}$ and $\hat{s}$ to equal their observed values (in Case 1, this would
		be $\hat{M}=\left\{1\right\} $ and $\hat{s}=-1$). The orange line shows the allowable range of $y_{1}$ if we condition only on $\hat{M}$ and the green line shows the allowable range of $y_{1}$ if we only
		require that a particular variable (in this case variable 1) is in the lasso active set.}
	\label{fig:TZ_int} 
\end{figure}

\subsection{Further simulations for the partial target} \label{app:partial:target}

Figures \ref{fig:partial_n100p50} and \ref{fig:partial_n100p125} show results for the cases $n=100,p=50$ and $n=100,p=125$ (the simulation settings are described in Section \ref{sec:sim_partial}).

\begin{figure}[bhtp]
	\begin{centering}
		\includegraphics[width=0.33\paperwidth]{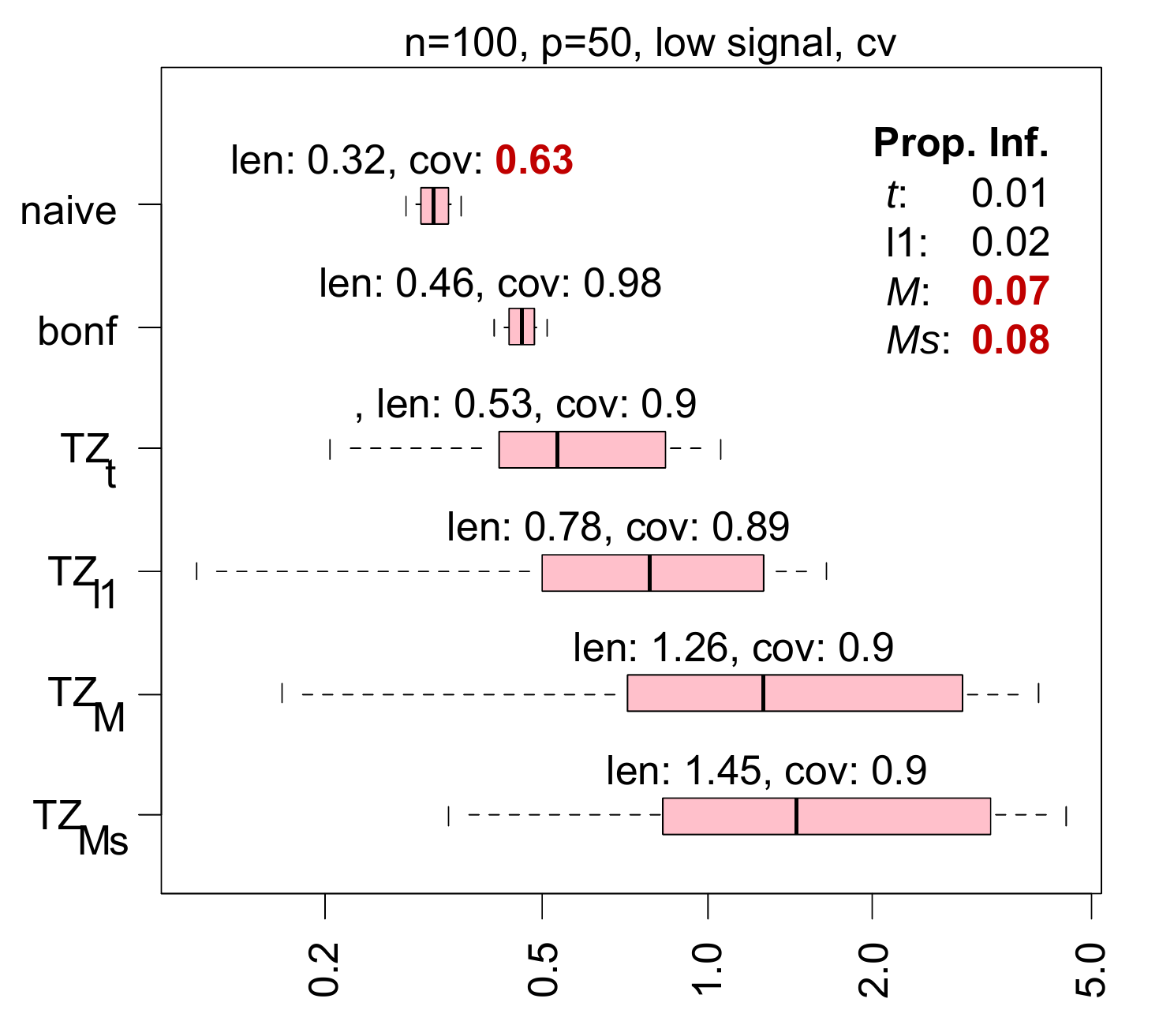}\includegraphics[width=0.33\paperwidth]{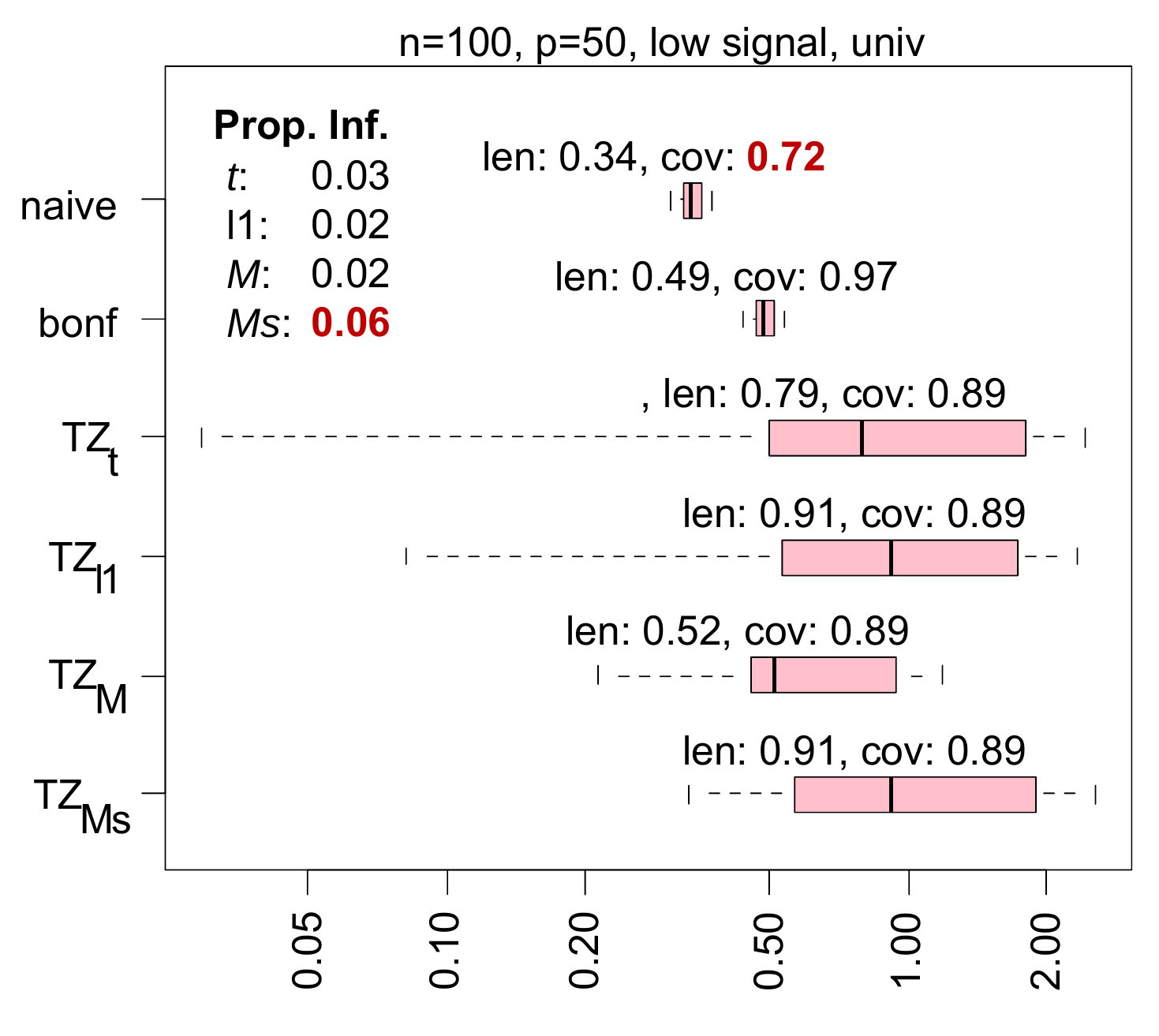}
		\par\end{centering}
	\begin{centering}
		\includegraphics[width=0.33\paperwidth]{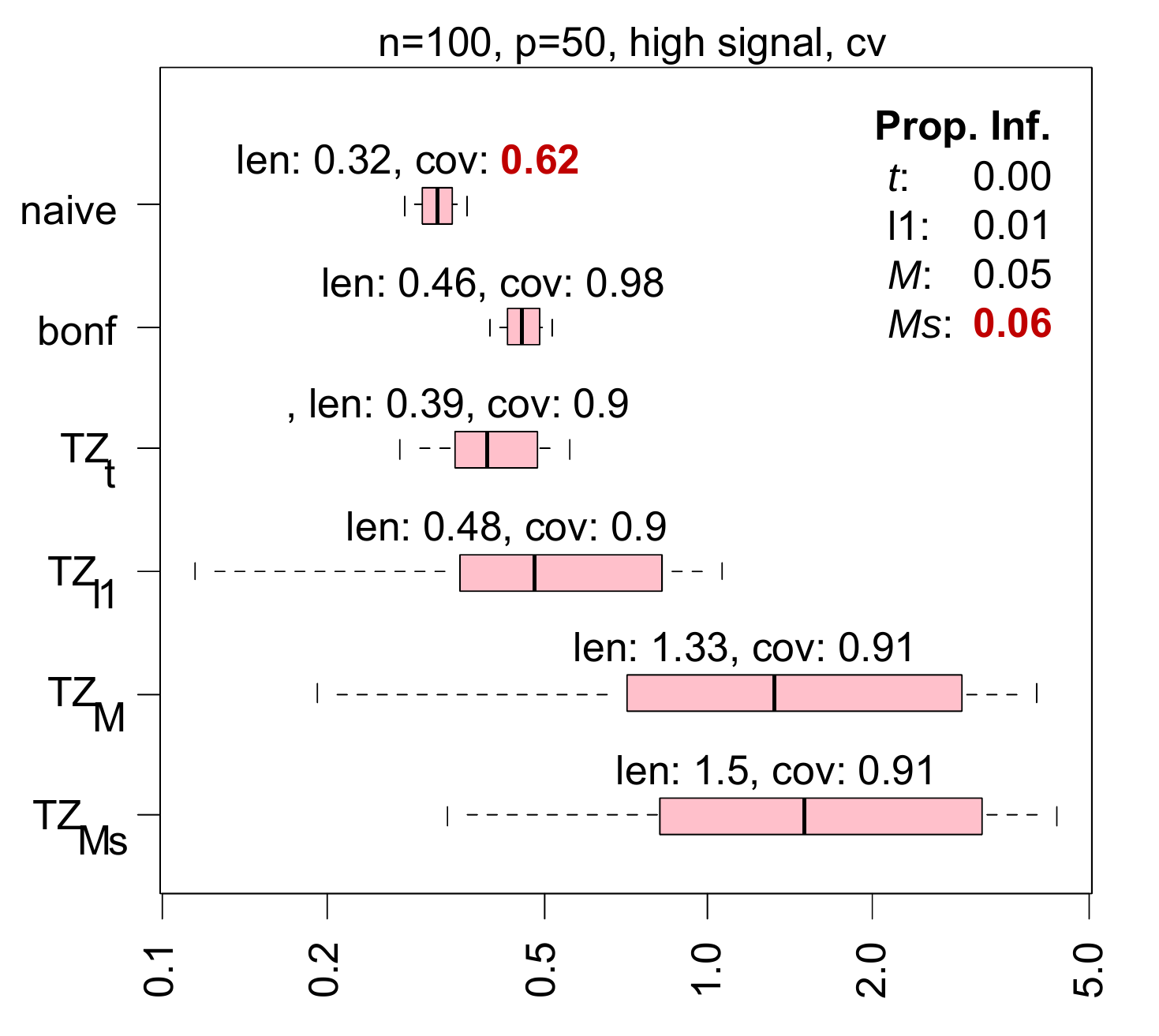}\includegraphics[width=0.33\paperwidth]{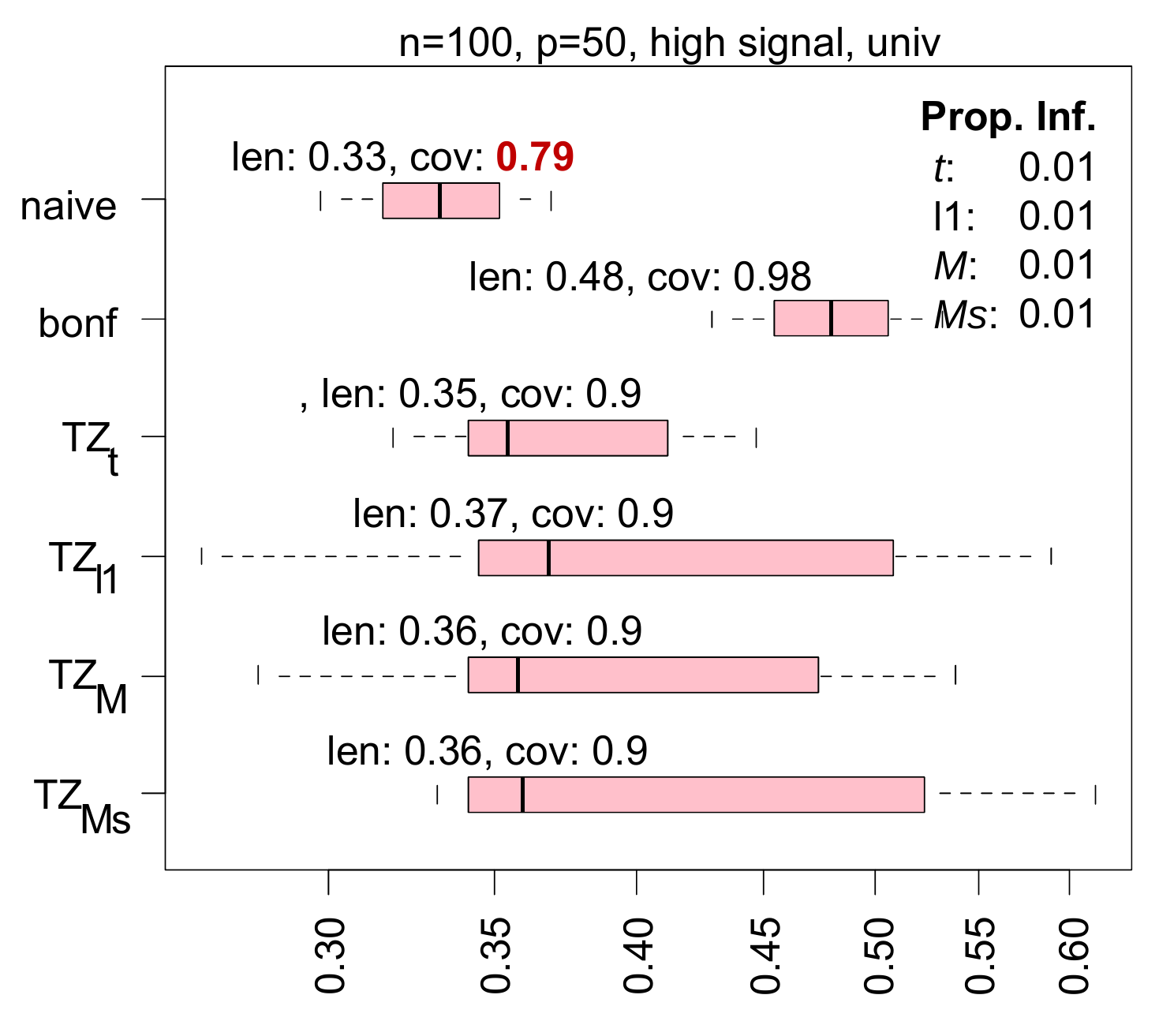}
		\par\end{centering}
	
	\caption{Boxplot of lengths of 90\% confidence intervals for ``partial''
		regression coefficients. Six interval methods are compared: naive
		(ignoring selection), Bonferroni adjusted, $\TZ_{\text{stab}-t}$, $\TZ_{\text{stab}-\ell_1}$, $\TZ_{M}$,
		and $\TZ_{Ms}$. Reported are the median interval length, the empirical
		coverage, and the proportion of ``infinite'' intervals (the infinite
		length results from numerical inaccuracies when inverting a truncated
		normal CDF); the boxplots set infinite lengths to the maximum finite
		observed length. Here $n=100,p=50$. The first five components of
		$\beta$ are set to $\delta_{\text{low}}=0.24$ (top panels) or $\delta_{\text{high}}=0.62$
		(bottom panels) and the remaining components are 0. The lasso penalty
		is either set at the universal threshold value $\sqrt{\frac{2\log p}{n}}\approx 0.28$
		(right panels) or at a value approximating the behavior of 10-fold
		cross validation (0.14 and 0.11 respectively for the low and high
		signal cases). }
	\label{fig:partial_n100p50}
\end{figure}

\begin{figure}[bhtp]
	\begin{centering}
		\includegraphics[width=0.33\paperwidth]{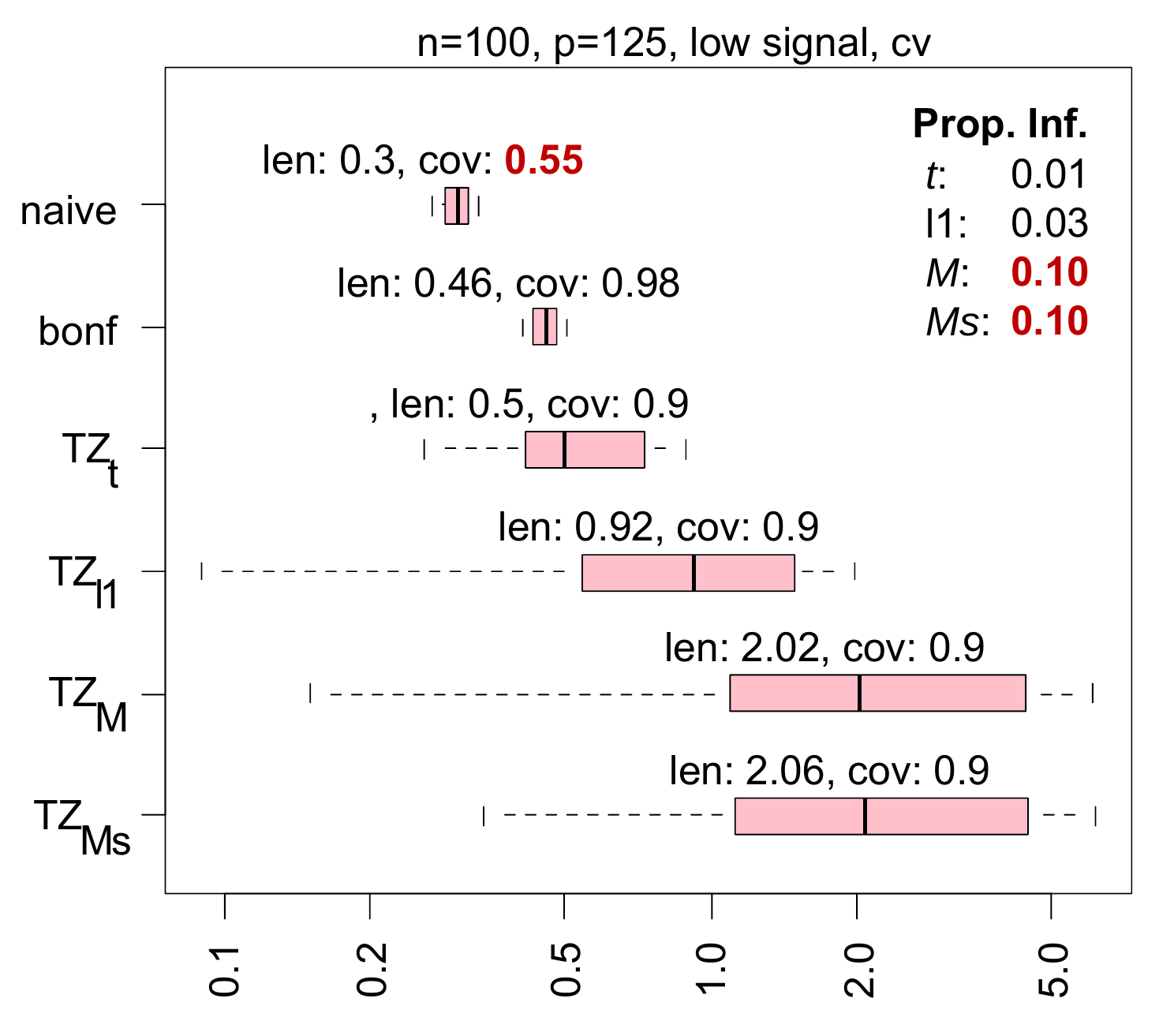}\includegraphics[width=0.33\paperwidth]{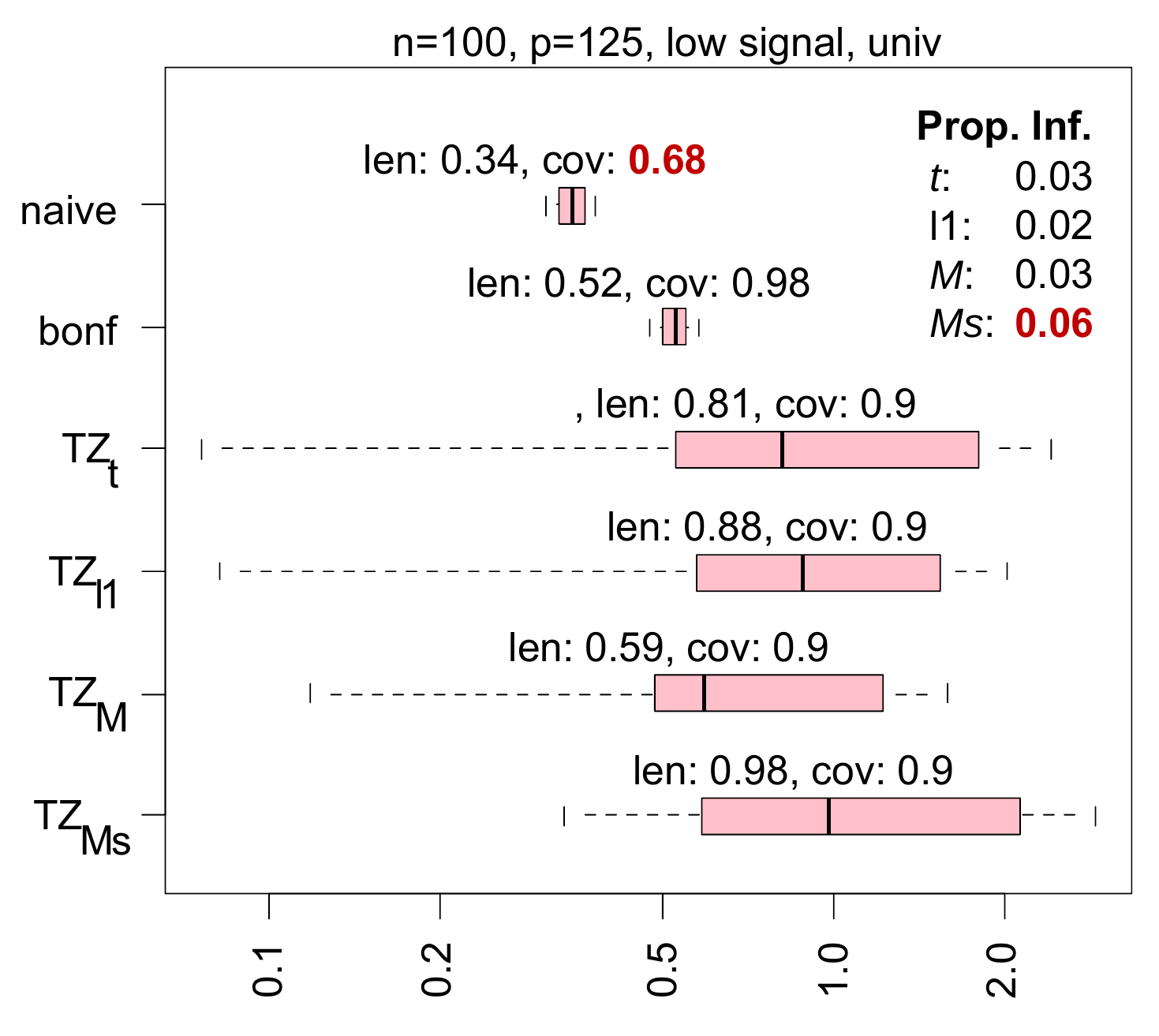}
		\par\end{centering}
	\begin{centering}
		\includegraphics[width=0.33\paperwidth]{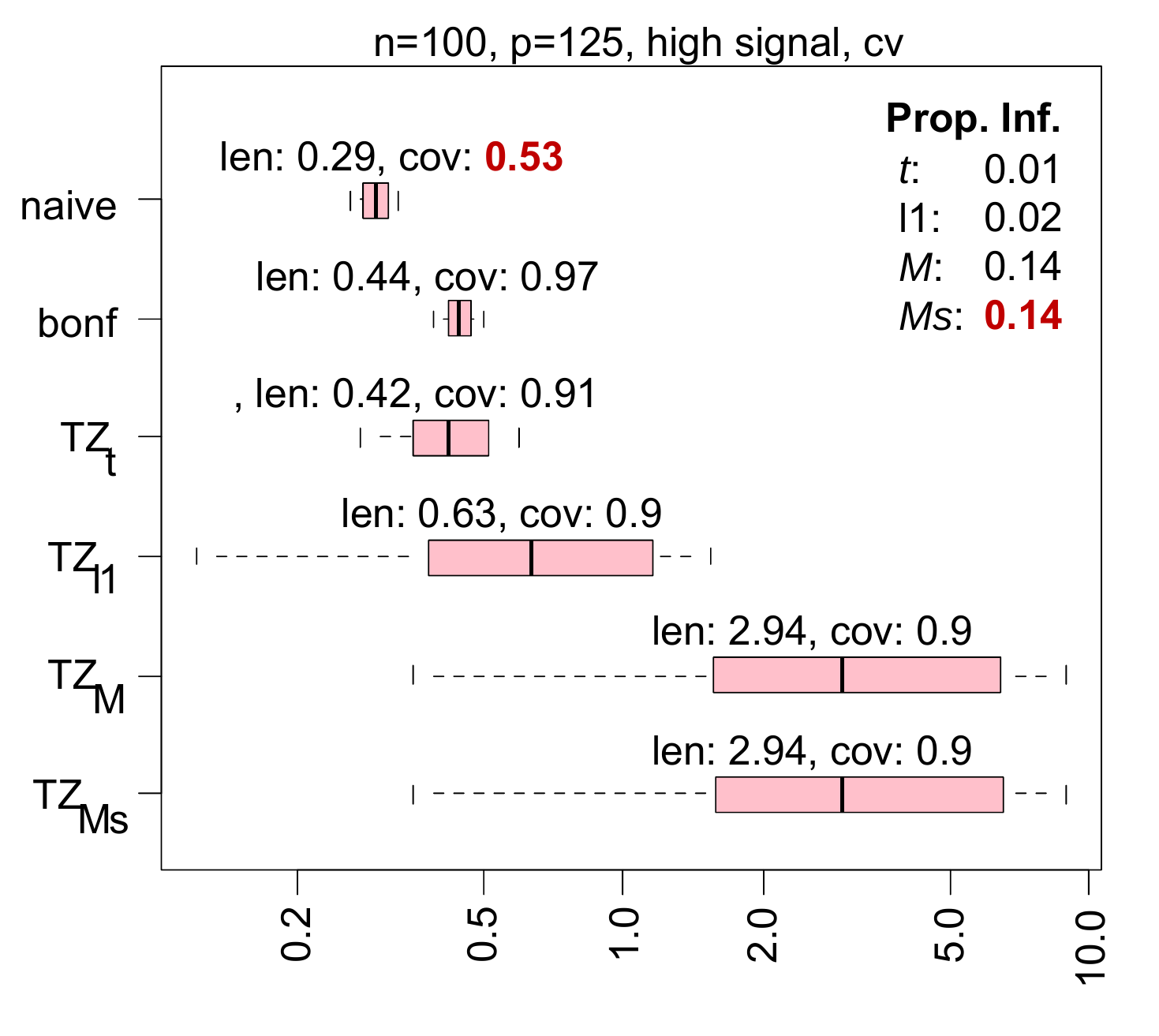}\includegraphics[width=0.33\paperwidth]{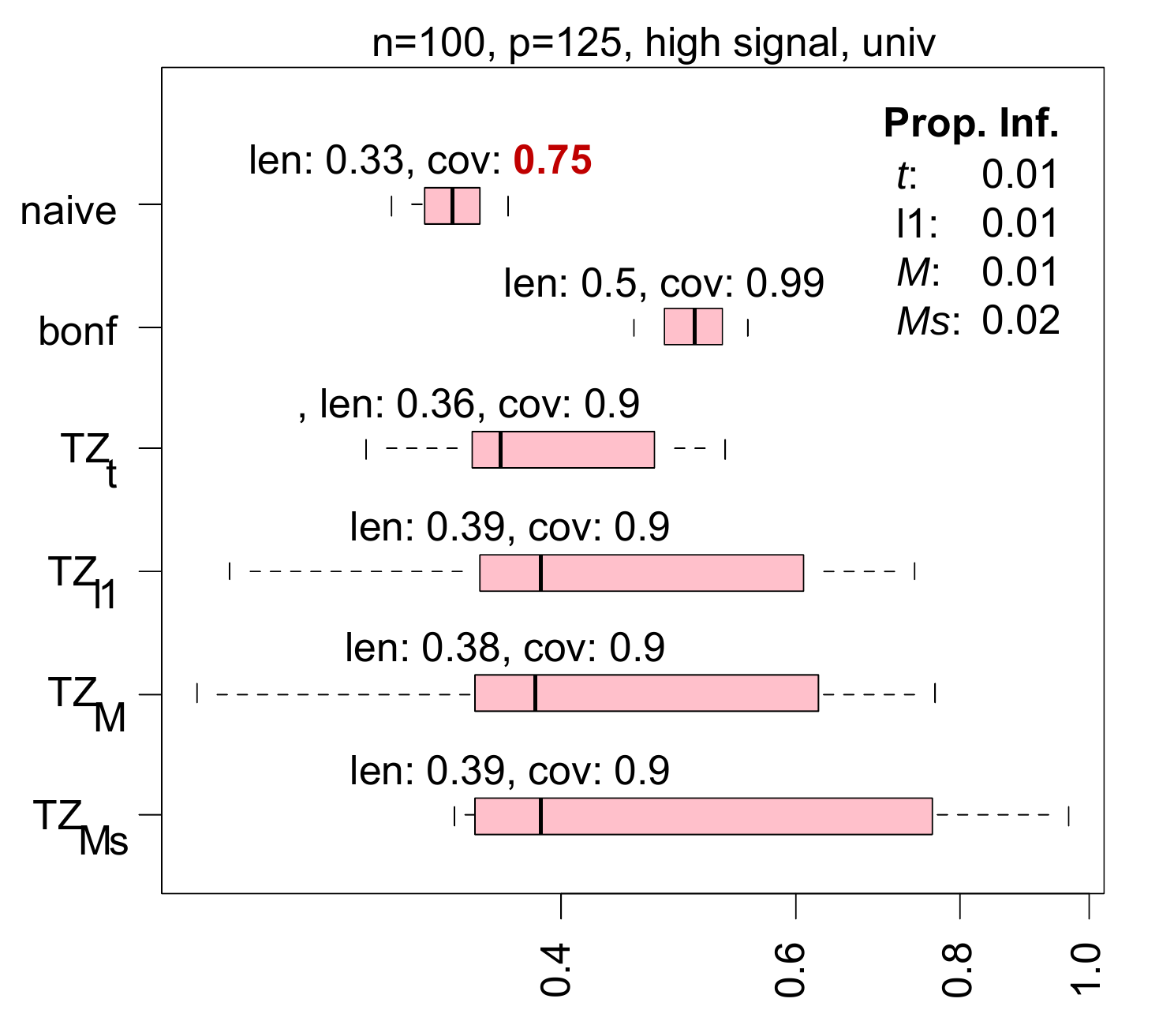}
		\par\end{centering}
	
	\caption{Boxplot of lengths of 90\% confidence intervals for ``partial''
		regression coefficients. Six interval methods are compared: naive
		(ignoring selection), Bonferroni adjusted, $\TZ_{\text{stab}-t}$, $\TZ_{\text{stab}-\ell_1}$, $\TZ_{M}$,
		and $\TZ_{Ms}$. Reported are the median interval length, the empirical
		coverage, and the proportion of ``infinite'' intervals (the infinite
		length results from numerical inaccuracies when inverting a truncated
		normal CDF); the boxplots set infinite lengths to the maximum finite
		observed length. Here $n=100,p=125$. The first five components of
		$\beta$ are set to $\delta_{\text{low}}=0.27$ (top panels) or $\delta_{\text{high}}=0.65$
		(bottom panels) and the remaining components are 0. The lasso penalty
		is either set at the universal threshold value $\sqrt{\frac{2\log p}{n}}\approx 0.31$
		(right panels) or at a value approximating the behavior of 10-fold
		cross validation (0.17 and 0.13 respectively for the low and high
		signal cases). }
	\label{fig:partial_n100p125}
\end{figure}

\subsection{Performance of the proposed methods for correlated predictors} \label{app:correlated}

We study whether our proposed methods continue to work well
when the predictors are correlated. As in the original simulations, there
are $k=5$ non-null variables but now they may be heavily correlated
with other variables. We consider two correlation schemes:

\begin{itemize}
\item Equi-correlation blocks: Divide all variables among five blocks of
size $p/5$. There is one non-null variable in each block and it has
a $\mathcal{N}\left(0,1\right)$ distribution. The other $p/5-1$ variables
in each block are generated from the non-null variable in that block
as follows:
\[
x_{j}=\rho x_{\text{non-null}}+\sqrt{1-\rho^{2}}\varepsilon\qquad\varepsilon\sim \mathcal{N}\left(0,1\right).
\]
We set $\rho=0.5$.
\item Toeplitz: The covariance matrix of the predictors has Toeplitz form
with the first row of the covariance matrix having the form
\[
\left(1,\rho,\rho^{2},\ldots\right).
\]
Five of the predictors are randomly chosen to be the non-null predictors.
We set $\rho=0.5$.
\end{itemize}

All simulations are performed with $n=100$ and $p=250$. Figures \ref{fig:partial_equi} and \ref{fig:partial_toeplitz} summarize
the results. The stable methods continue to outperform the $\TZ_{M}$ and
$\TZ_{Ms}$ intervals in the presence of strongly correlated predictors. In fact, we see that the stable-$t$ method shines even more than in the case of independent
predictors\textemdash in particularly, it significantly outperforms
the $\TZ_{M}$ and $\TZ_{Ms}$ intervals even when $\lambda$ is set to be the universal
threshold value. This performance advantage is due to the fact that the active set of the lasso becomes less stable when variables are correlated.

\begin{figure}[hbtp]
\begin{centering}
	\includegraphics[width=0.33\paperwidth]{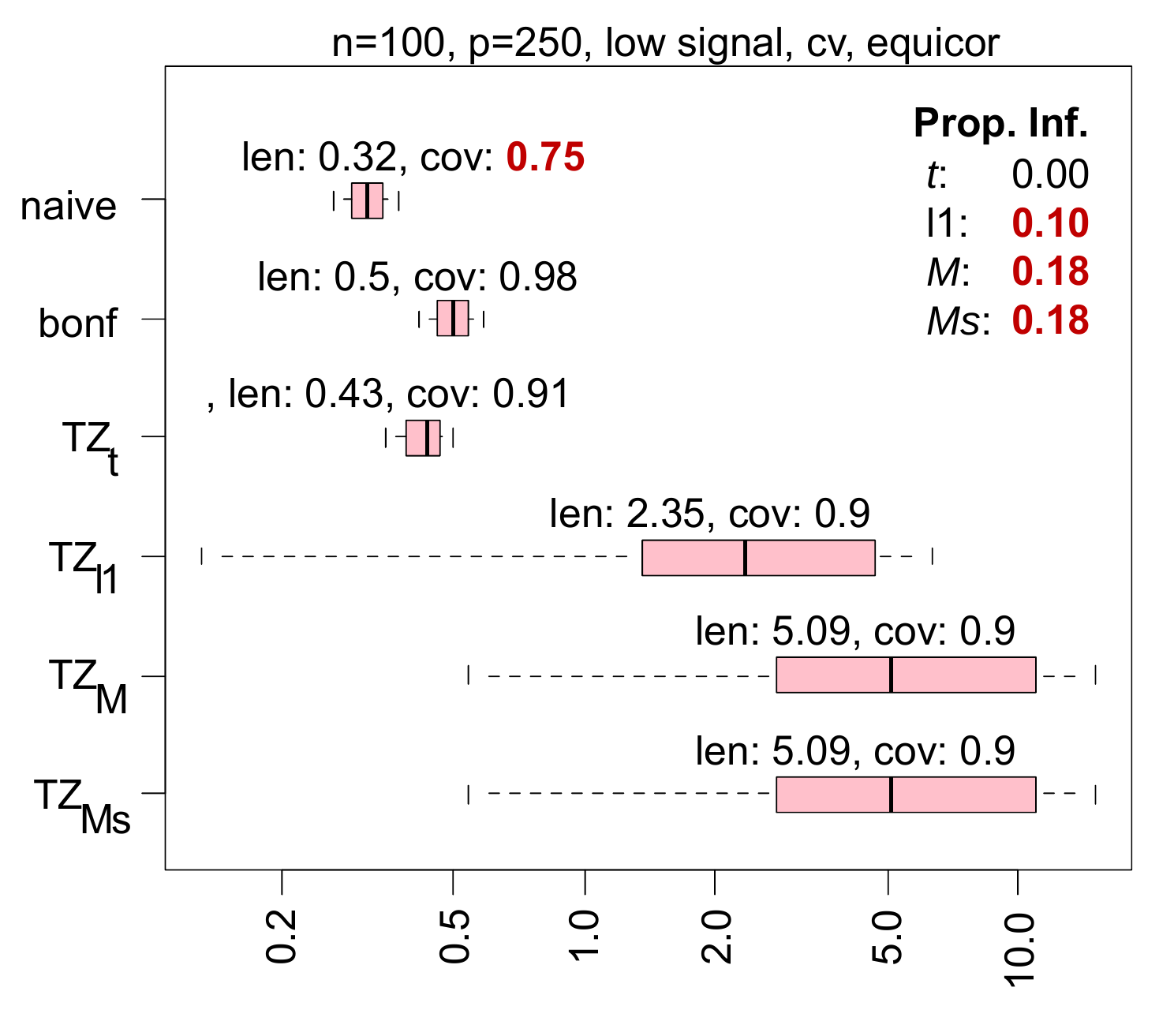}\includegraphics[width=0.33\paperwidth]{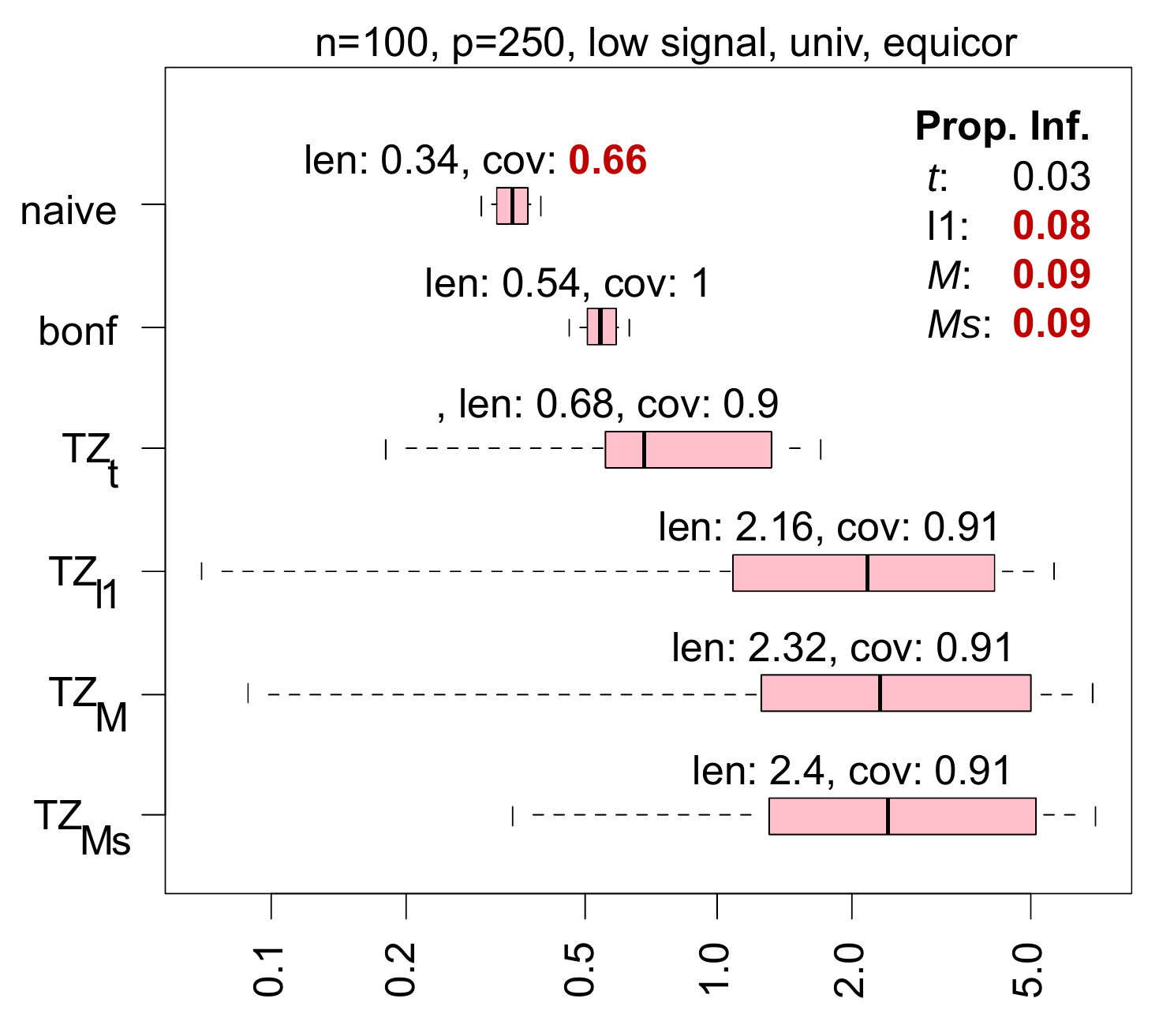}
	\par\end{centering}
\begin{centering}
	\includegraphics[width=0.33\paperwidth]{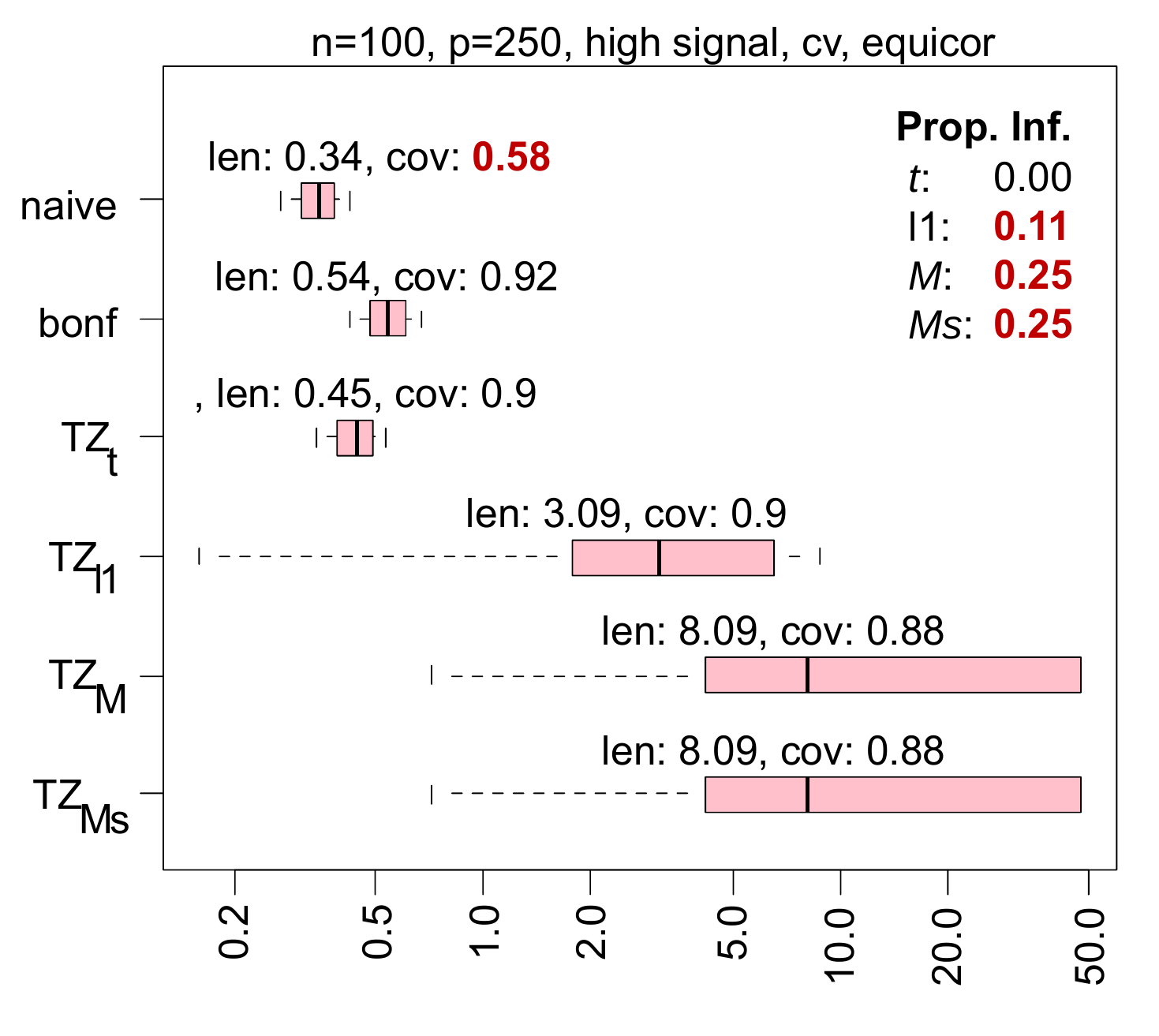}\includegraphics[width=0.33\paperwidth]{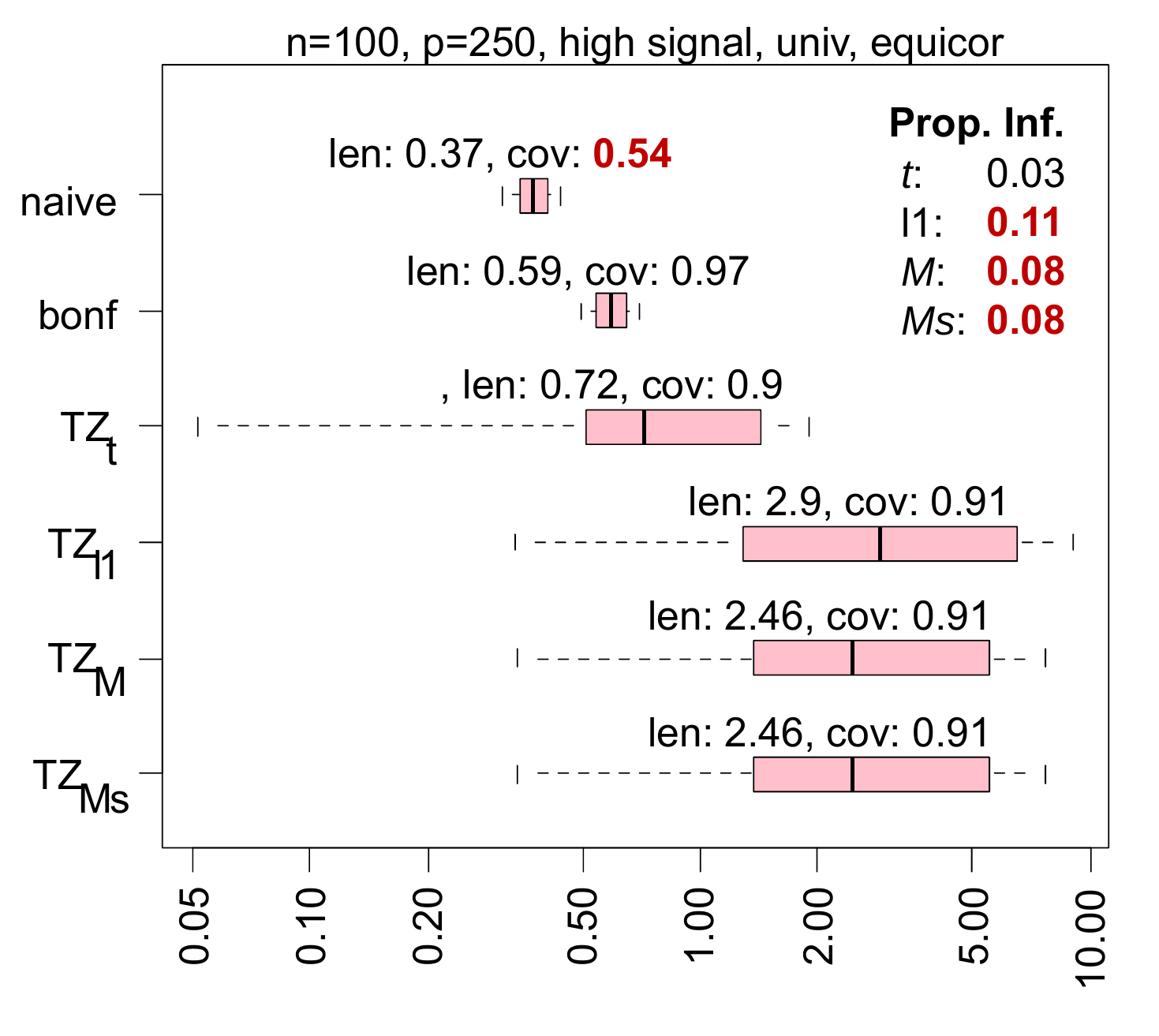}
	\par\end{centering}

\caption{$n=100,p=250$. Block equicorrelation covariance matrix with $\rho=0.5$. Boxplot of lengths of 90\% confidence intervals for ``partial''
	regression coefficients. Six interval methods are compared: naive
	(ignoring selection), Bonferroni adjusted, $\TZ_{\text{stab}-t}$, $\TZ_{\text{stab}-\ell_1}$, $\TZ_{M}$,
	and $\TZ_{Ms}$. Reported are the median interval length, the empirical
	coverage, and the proportion of ``infinite'' intervals (the infinite
	length results from numerical inaccuracies when inverting a truncated
	normal CDF); the boxplots set infinite lengths to the maximum finite
	observed length. The first five components of
	$\beta$ are set to $\delta_{\text{low}}=0.29$ (top panels) or $\delta_{\text{high}}=0.67$
	(bottom panels) and the remaining components are 0. The lasso penalty
	is either set at the universal threshold value $\sqrt{\frac{2\log p}{n}}\approx 0.33$
	(right panels) or at a value approximating the behavior of 10-fold
	cross validation (0.15 and 0.12 respectively for the low and high
	signal cases). }

\label{fig:partial_equi}
\end{figure}

\begin{figure}[hbtp]
\begin{centering}
	\includegraphics[width=0.33\paperwidth]{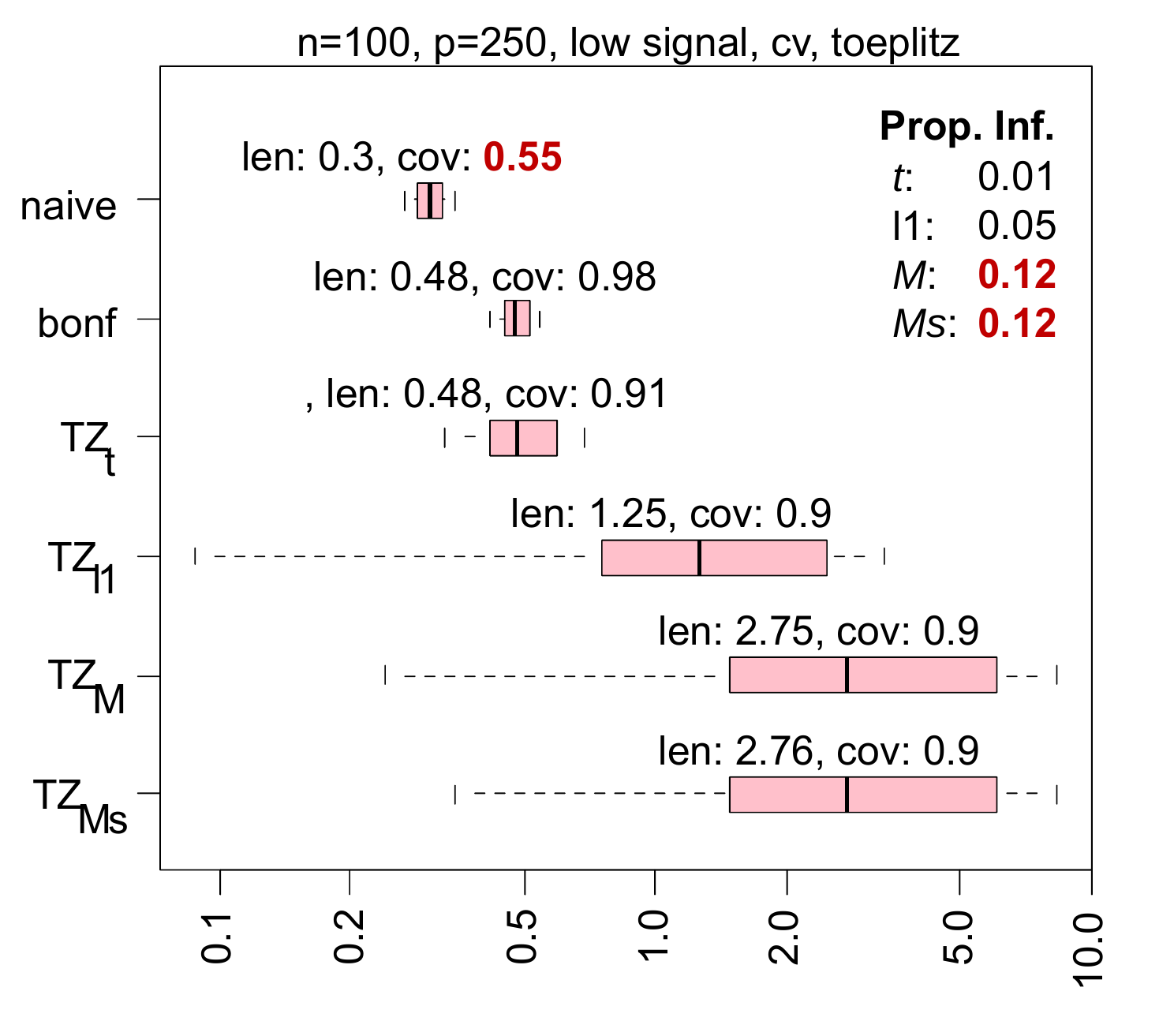}\includegraphics[width=0.33\paperwidth]{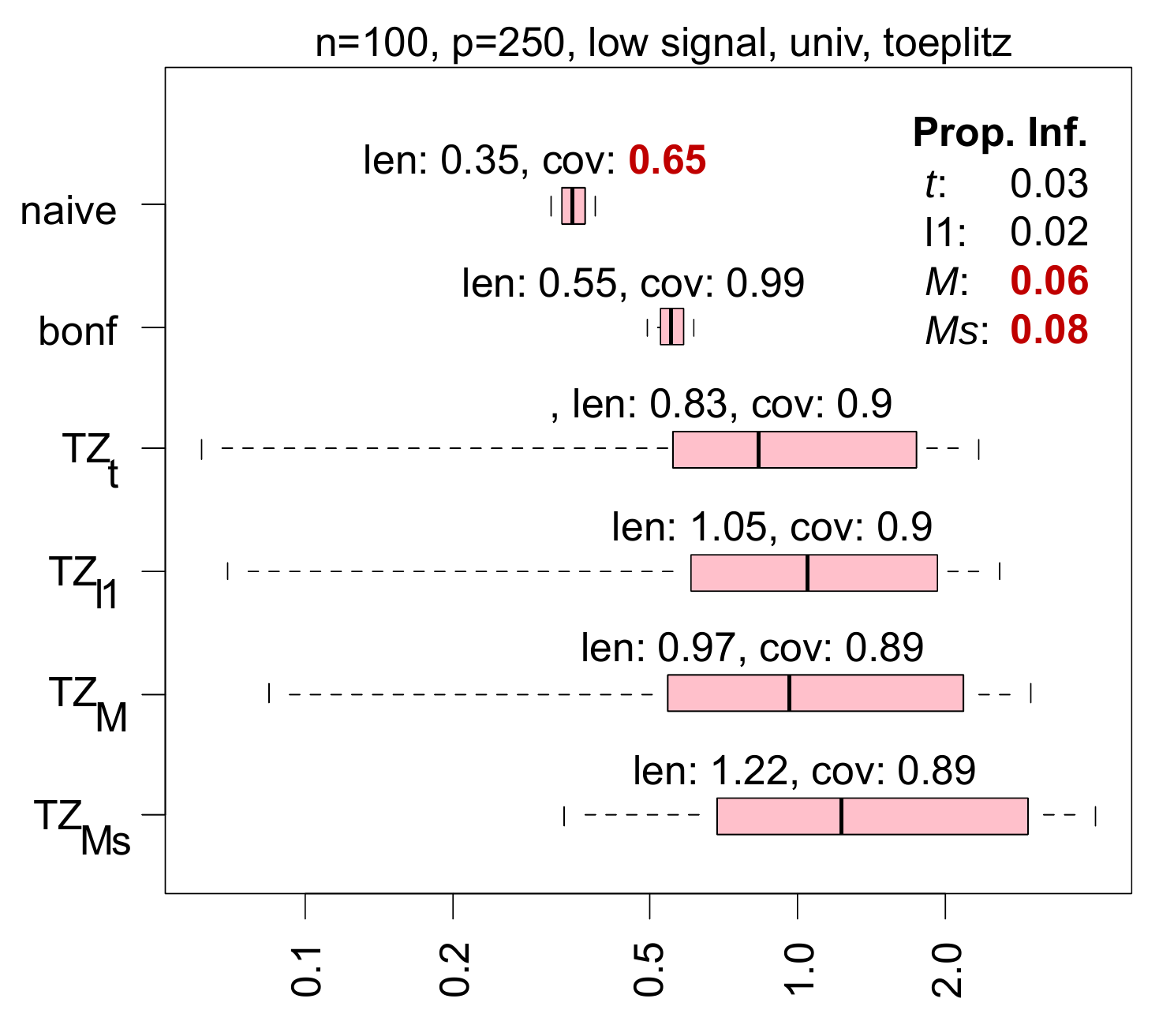}
	\par\end{centering}
\begin{centering}
	\includegraphics[width=0.33\paperwidth]{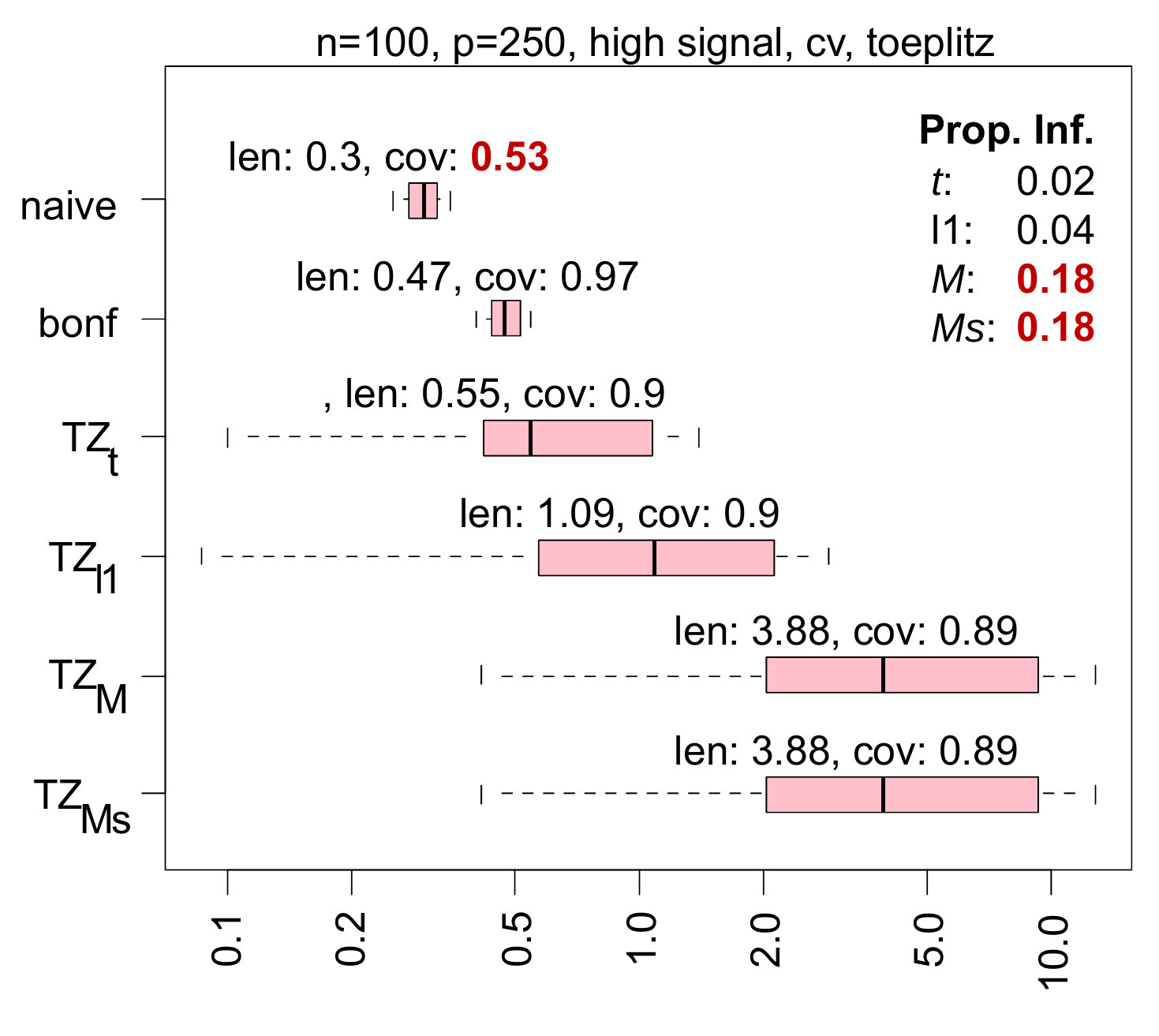}\includegraphics[width=0.33\paperwidth]{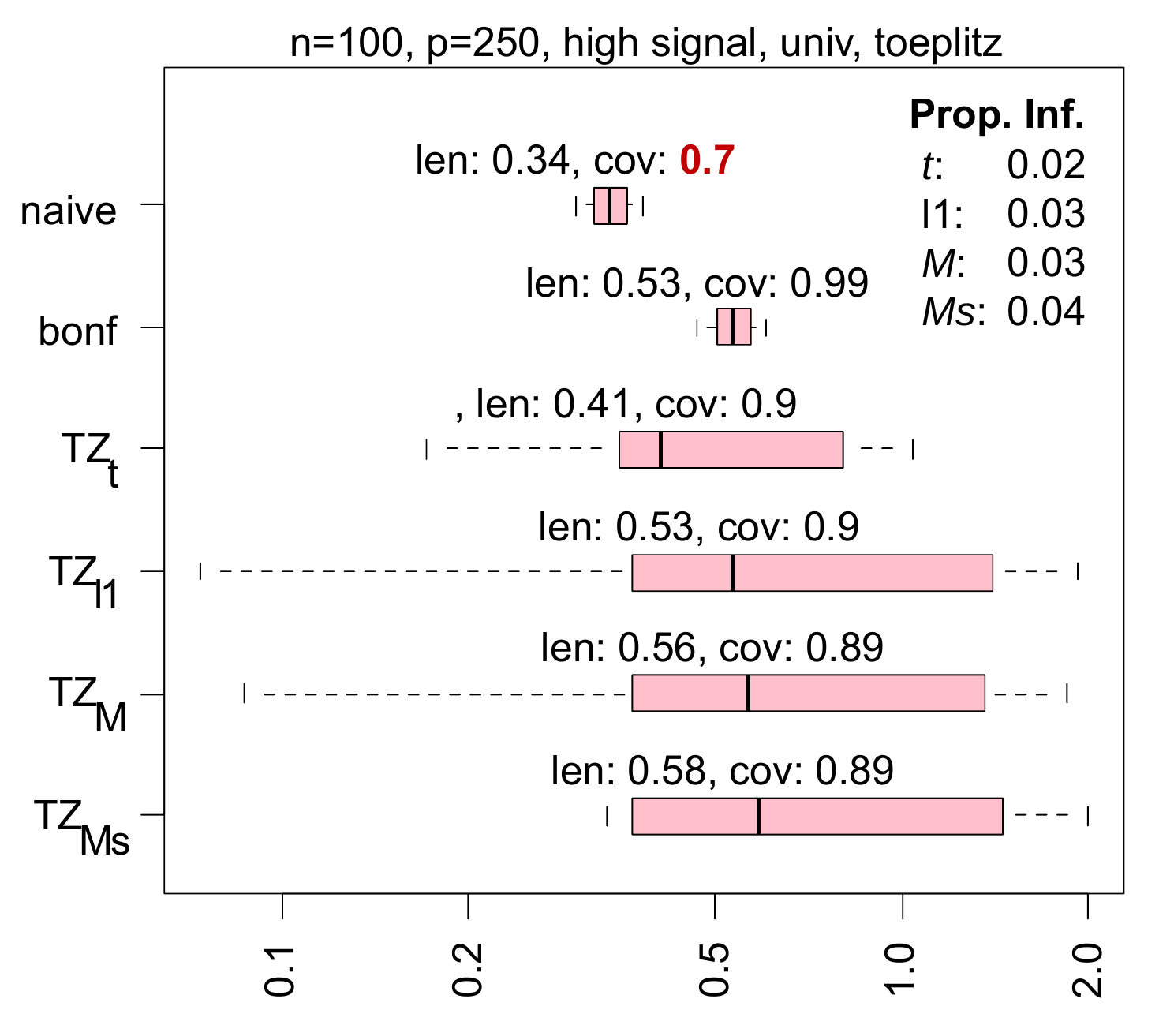}
	\par\end{centering}

\caption{$n=100,p=250$. Toeplitz matrix with geometric correlation decay and
	$\rho=0.5$. Boxplot of lengths of 90\% confidence intervals for ``partial''
	regression coefficients. Six interval methods are compared: naive
	(ignoring selection), Bonferroni adjusted, $\TZ_{\text{stab}-t}$, $\TZ_{\text{stab}-\ell_1}$, $\TZ_{M}$,
	and $\TZ_{Ms}$. Reported are the median interval length, the empirical
	coverage, and the proportion of ``infinite'' intervals (the infinite
	length results from numerical inaccuracies when inverting a truncated
	normal CDF); the boxplots set infinite lengths to the maximum finite
	observed length. The first five components of
	$\beta$ are set to $\delta_{\text{low}}=0.29$ (top panels) or $\delta_{\text{high}}=0.66$
	(bottom panels) and the remaining components are 0. The lasso penalty
	is either set at the universal threshold value $\sqrt{\frac{2\log p}{n}}\approx 0.33$
	(right panels) or at a value approximating the behavior of 10-fold
	cross validation (0.19 and 0.15 respectively for the low and high
	signal cases). }
\label{fig:partial_toeplitz}

\end{figure}

\subsection{Performance of the proposed methods when assumptions are violated} \label{app:violated:assumptions}

We now investigate whether the good performance of the stable methods
are robust to various violations of our assumptions. In particular,
we consider:
\begin{itemize}
	\item Non-normal errors: we consider the effect of heavy tails ($t_{3}$
	distribution) and skewness (skew normal with skewness coefficient
	10) in our error distribution. See Figures \ref{fig:partial_t3} and \ref{fig:partial_sn10}.
	\item Unknown $\sigma^{2}$: Estimating the noise level when $p>n$ is a
	challenging problem. However, \citet{reid2016variance_estimate} shows
	that 
	\[
	\hat{\sigma}^{2}=\frac{1}{n-p_{\hat{\lambda}}}\sum_{i=1}^{n}\left(y_{i}-x_{i}^{\top}\hat\beta_{\hat{\lambda}}\right)^{2}
	\]
	performs well where $\hat\beta_{\hat{\lambda}}$ is the lasso solution,
	$p_{\hat{\lambda}}$ is the number of active variables, and $\hat{\lambda}$
	is chosen by 10-fold cross validation. We then plug-in $\hat{\sigma}^{2}$
	in place of the unknown noise level. See Figure \ref{fig:unknown_variance}.

	\item Data driven choice of $\lambda$: the $\TZ$ methods adjust for
	the effect of selection assuming a pre-specified $\lambda$; they
	do not account for the fact that $\lambda$ is often chosen based
	on the data, e.g., via cross-validation. See Figure \ref{fig:data_driven_lam}.
\end{itemize}
For all simulations, we use $n=100$ and $p=250$. We summarize the
results as follows:
\begin{itemize}
	\item All the truncated-Z intervals maintained nominal coverage close to 0.9 inspite
	of heavy tails and skewness. Qualitatively, the improvements of the
	stable methods over the $\TZ_{M}$ and $\TZ_{Ms}$ intervals are the
	same as in the case of normal errors.
	\item When $\sigma^{2}$ was unknown, all truncated-Z methods maintained
	coverage close to the nominal level by plugging in an estimate $\hat{\sigma}^{2}$
	(see Figure \ref{fig:unknown_variance}). However, the lengths of
	the $\TZ_{M}$ and $\TZ_{Ms}$ intervals are much longer than in the
	case of known $\sigma^{2}$ (cf. Figure \ref{fig:partial_n100p250});
	the lengths of the stable methods were less impacted.
	\item When $\lambda$ was chosen by 10-fold cross validation rather than
	set at a fixed value, the truncated-Z methods still maintained coverage
	close to the nominal level of 0.9 despite ignoring the cost of estimating
	$\lambda.$ However, just as in the case of unknown $\sigma^{2}$,
	the lengths of the $\TZ_{M}$ and $\TZ_{Ms}$ intervals greatly increased
	while the lengths of the stable intervals were relatively immune.
\end{itemize}
What these simulations suggest is that (1) the truncated-Z interval
seem to possess coverage properties that are robust to various model
violations and (2) the lengths of the stable intervals are less likely
to be negatively impacted by model violations than the lengths of
the $\TZ_{M}$ and $\TZ_{Ms}$ intervals.

\begin{figure}[hbtp]
	\begin{centering}
		\includegraphics[width=0.33\paperwidth]{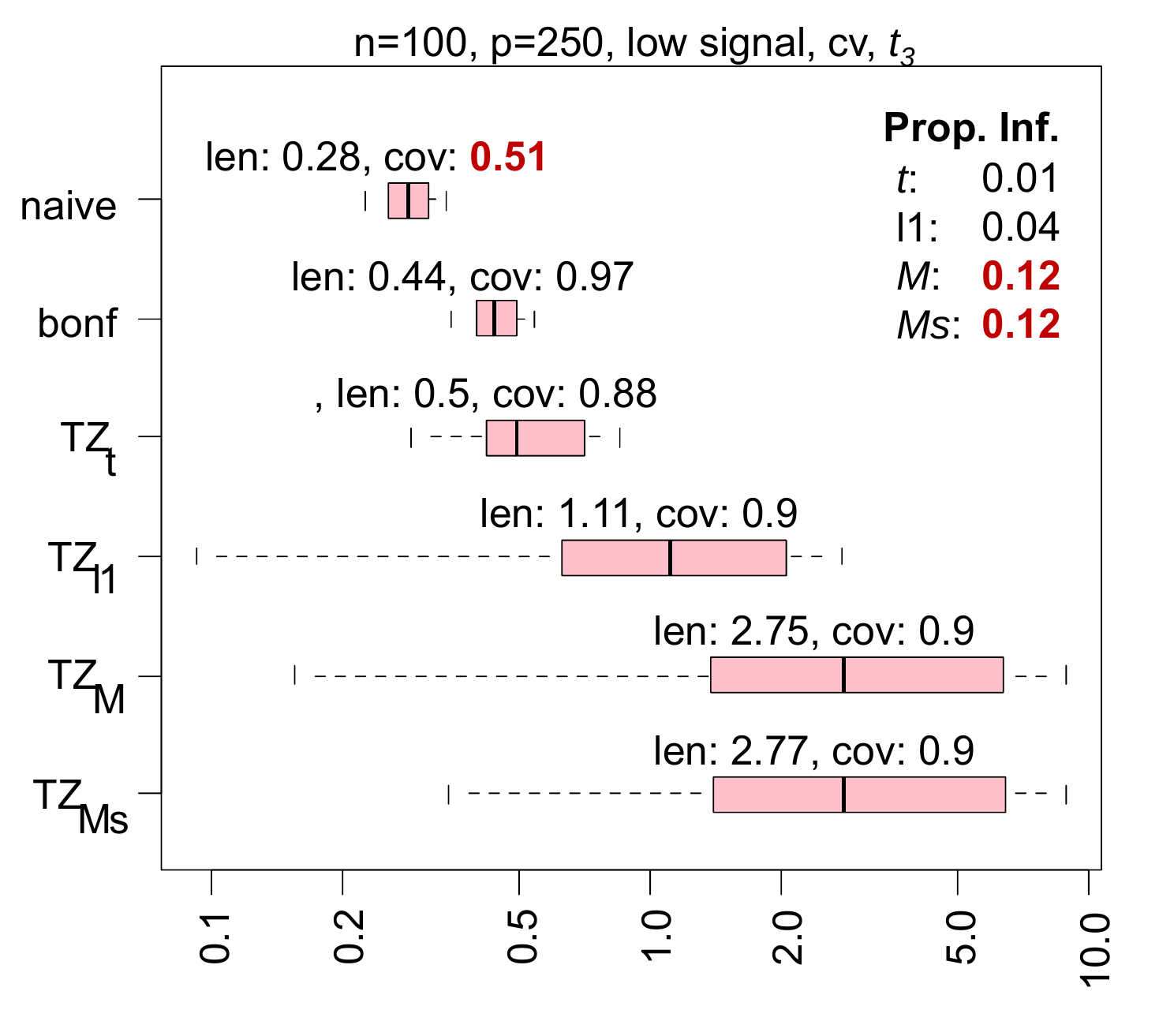}\includegraphics[width=0.33\paperwidth]{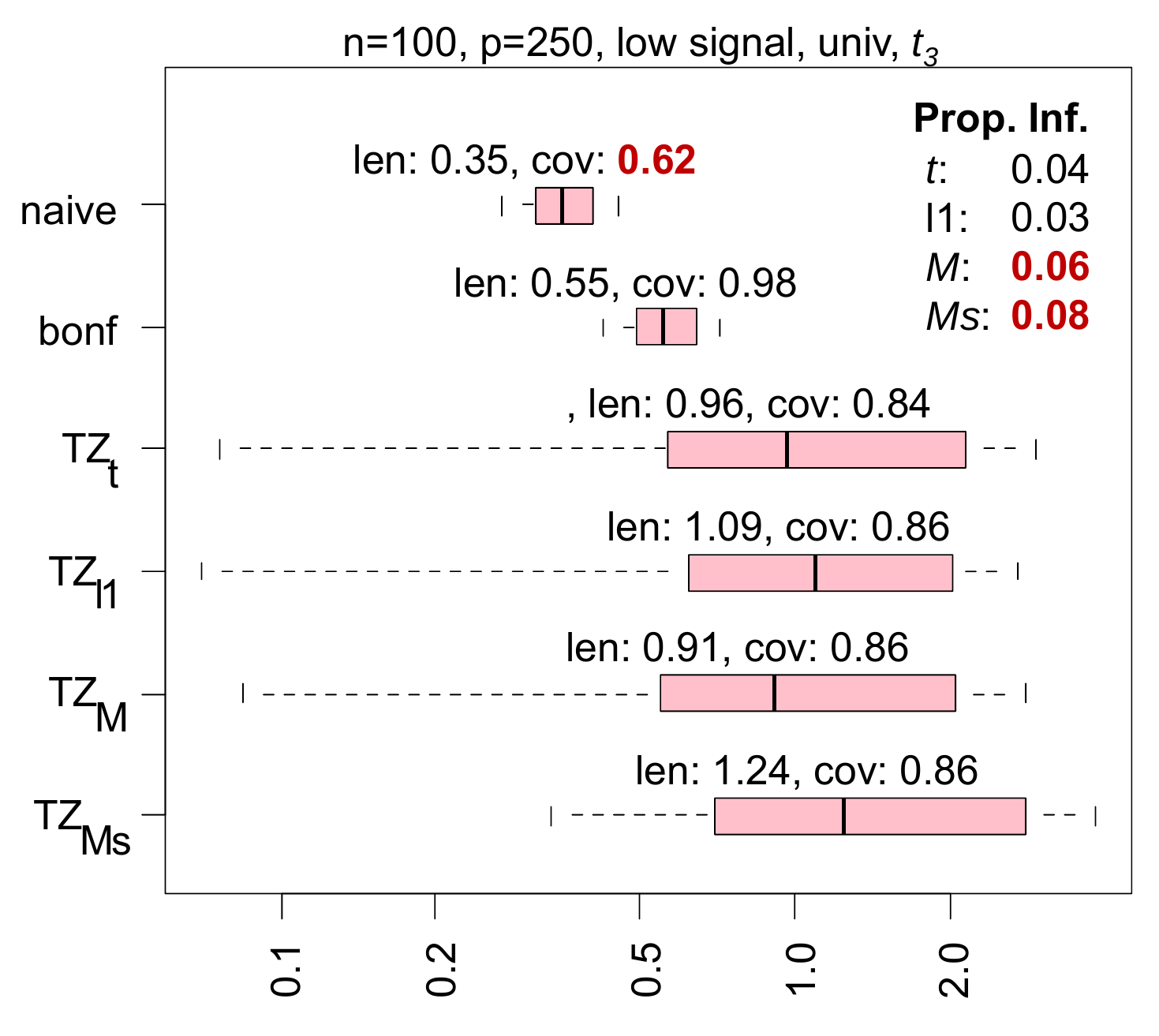}
		\par\end{centering}
	\begin{centering}
		\includegraphics[width=0.33\paperwidth]{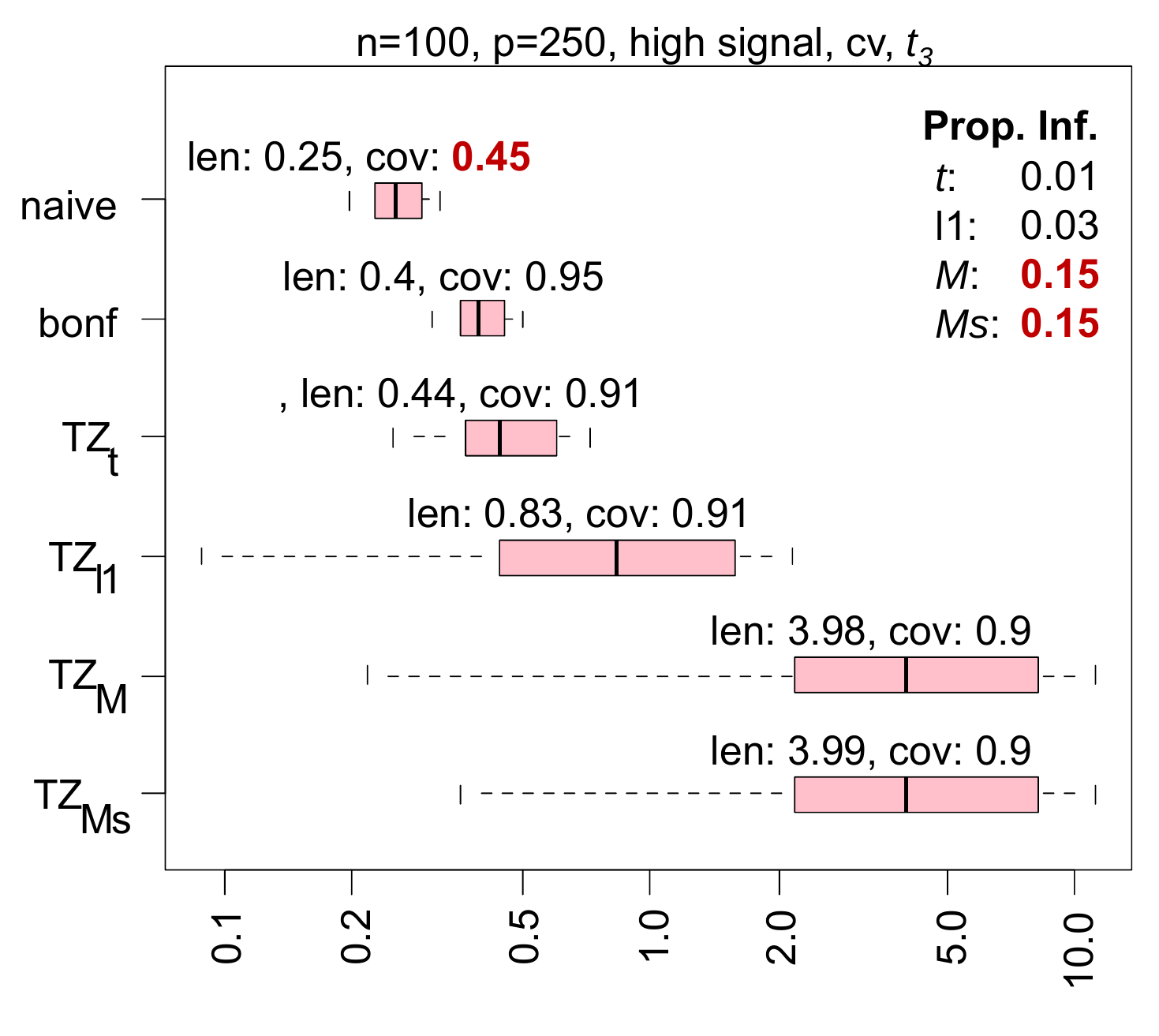}\includegraphics[width=0.33\paperwidth]{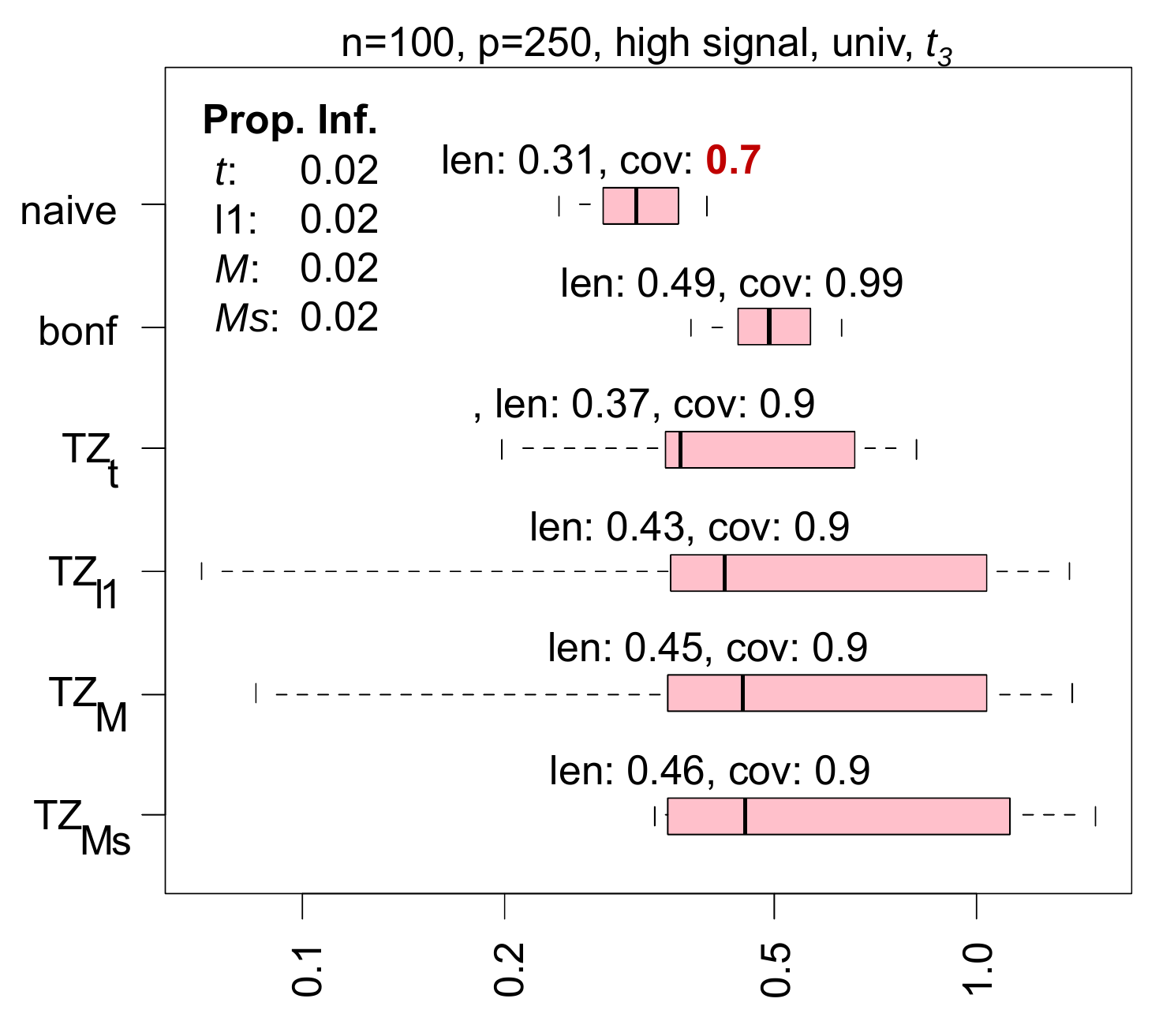}
		\par\end{centering}

	\caption{$n=100,p=250$. $t_{3}$ noise distribution. Boxplot of lengths of 90\% confidence intervals for ``partial''
		regression coefficients. Six interval methods are compared: naive
		(ignoring selection), Bonferroni adjusted, $\TZ_{\text{stab}-t}$, $\TZ_{\text{stab}-\ell_1}$, $\TZ_{M}$,
		and $\TZ_{Ms}$. Reported are the median interval length, the empirical
		coverage, and the proportion of ``infinite'' intervals (the infinite
		length results from numerical inaccuracies when inverting a truncated
		normal CDF); the boxplots set infinite lengths to the maximum finite
		observed length. The first five components of
		$\beta$ are set to $\delta_{\text{low}}=0.29$ (top panels) or $\delta_{\text{high}}=0.68$
		(bottom panels) and the remaining components are 0. The lasso penalty
		is either set at the universal threshold value $\sqrt{\frac{2\log p}{n}}\approx 0.33$
		(right panels) or at a value approximating the behavior of 10-fold
		cross validation (0.19 and 0.14 respectively for the low and high
		signal cases).}
		\label{fig:partial_t3}
	
\end{figure}

\begin{figure}[hbtp]
	\begin{centering}
		\includegraphics[width=0.33\paperwidth]{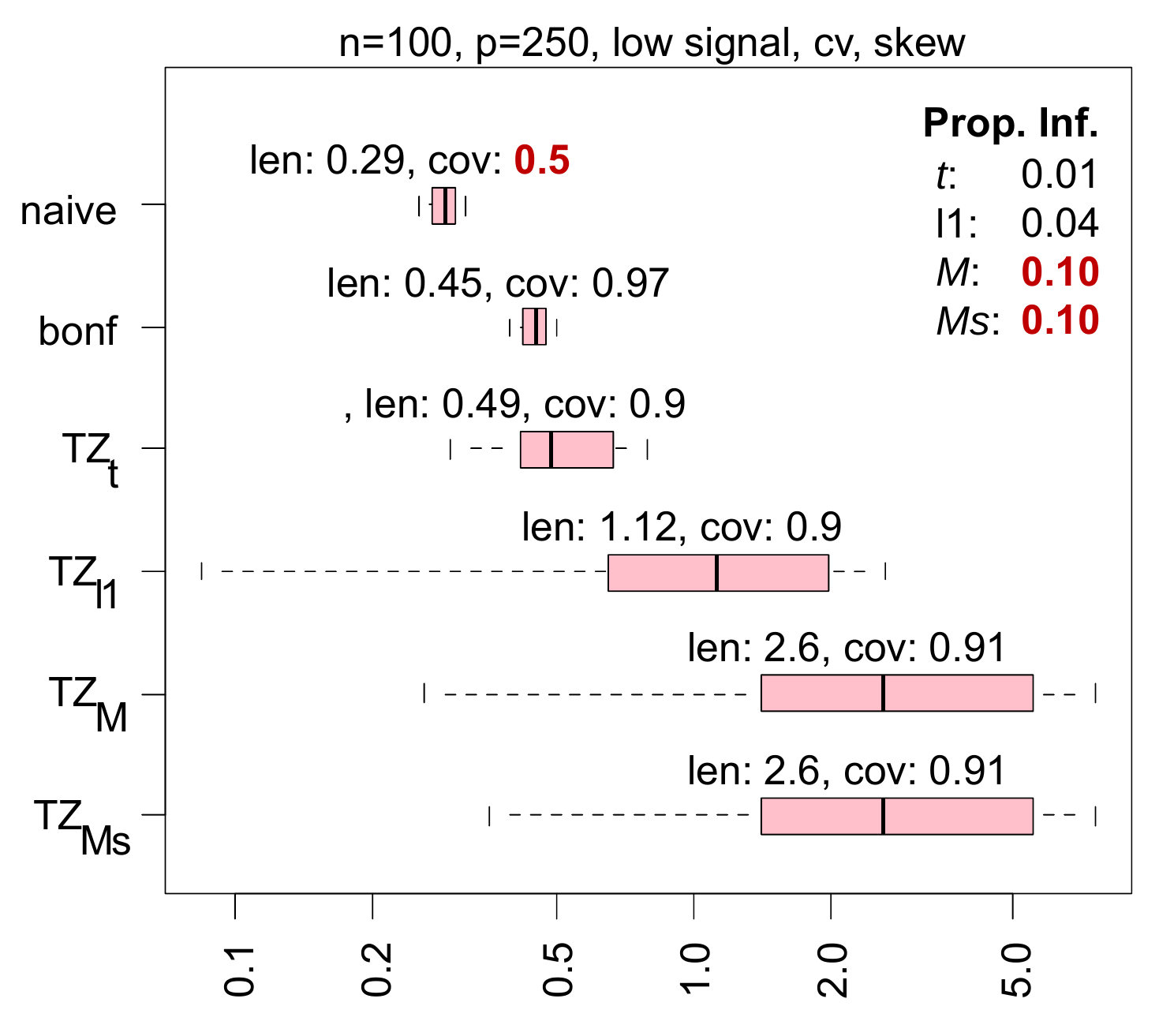}\includegraphics[width=0.33\paperwidth]{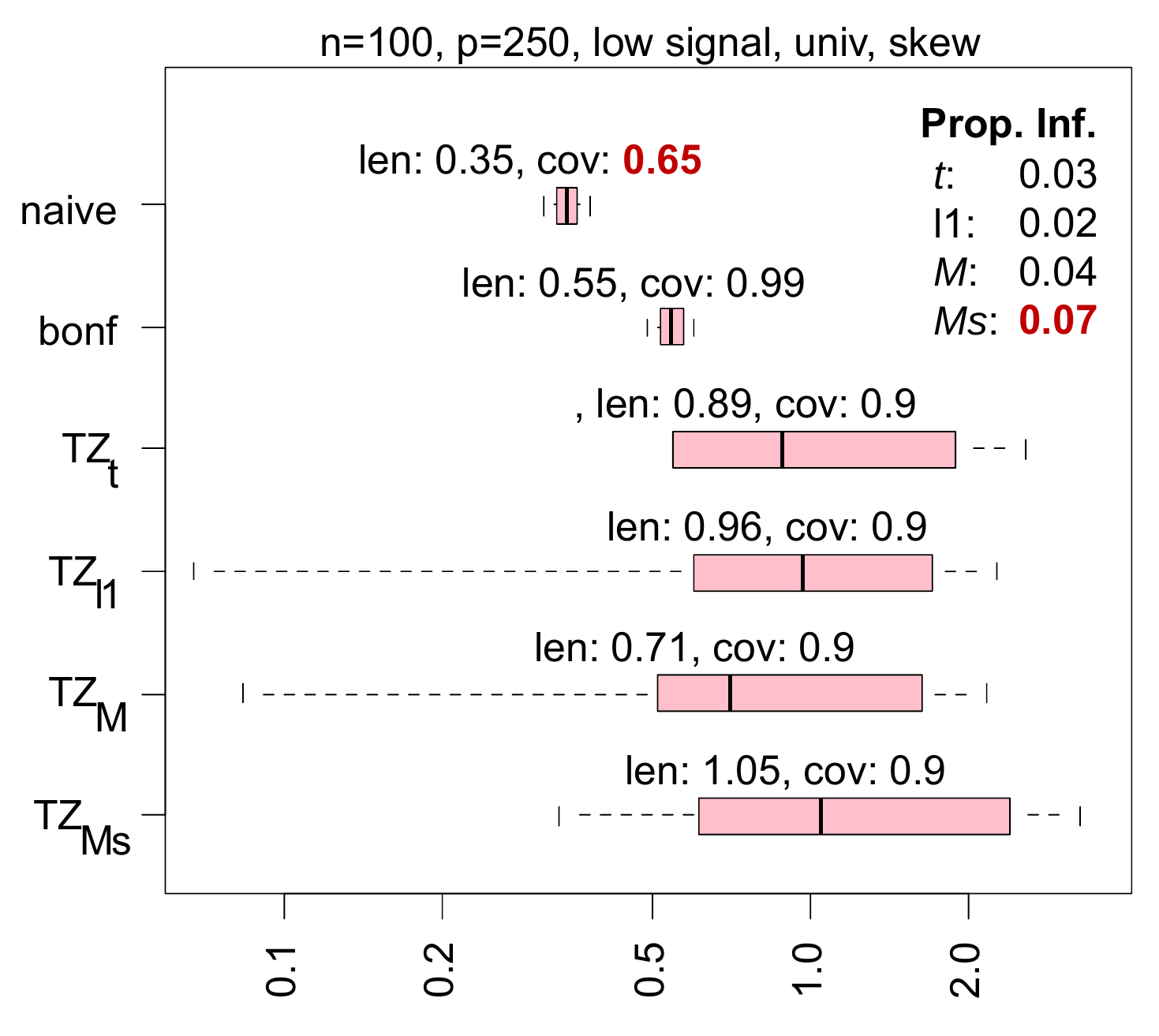}
		\par\end{centering}
	\begin{centering}
		\includegraphics[width=0.33\paperwidth]{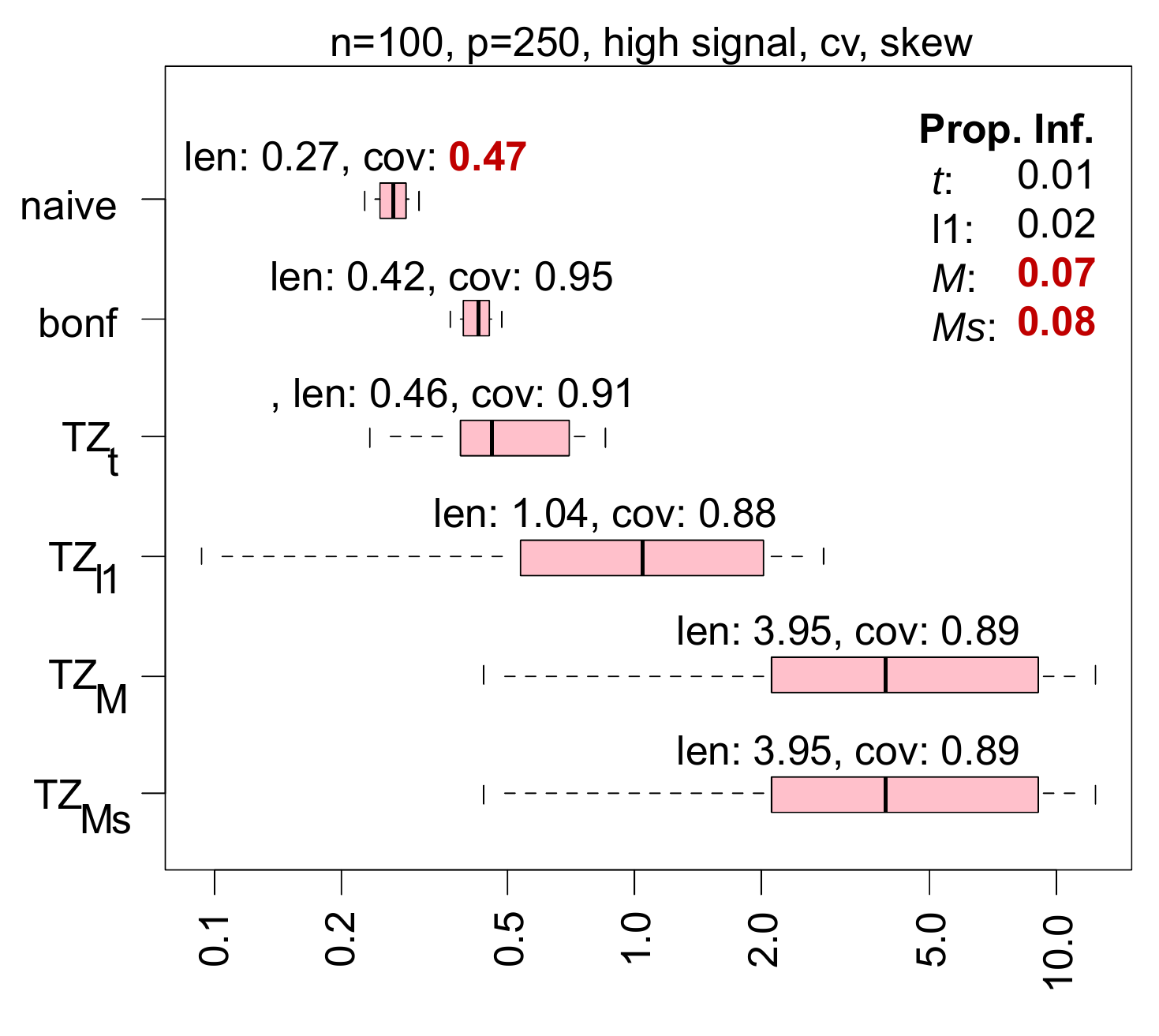}\includegraphics[width=0.33\paperwidth]{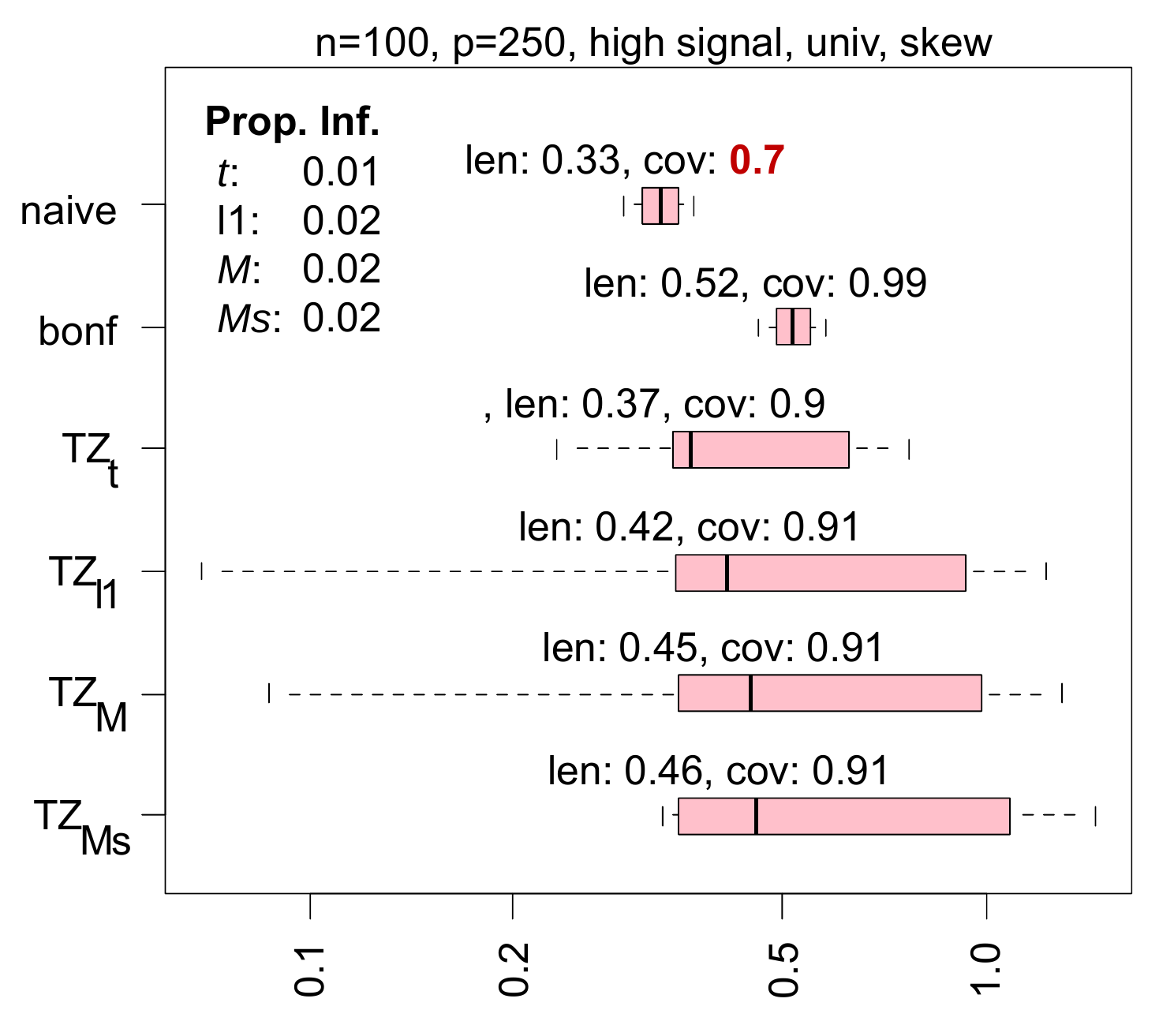}
		\par\end{centering}

	\caption{$n=100,p=250$. Skew normal noise distribution with skewness coefficient
		of 10. Boxplot of lengths of 90\% confidence intervals for ``partial''
		regression coefficients. Six interval methods are compared: naive
		(ignoring selection), Bonferroni adjusted, $\TZ_{\text{stab}-t}$, $\TZ_{\text{stab}-\ell_1}$, $\TZ_{M}$,
		and $\TZ_{Ms}$. Reported are the median interval length, the empirical
		coverage, and the proportion of ``infinite'' intervals (the infinite
		length results from numerical inaccuracies when inverting a truncated
		normal CDF); the boxplots set infinite lengths to the maximum finite
		observed length. The first five components of
		$\beta$ are set to $\delta_{\text{low}}=0.29$ (top panels) or $\delta_{\text{high}}=0.68$
		(bottom panels) and the remaining components are 0. The lasso penalty
		is either set at the universal threshold value $\sqrt{\frac{2\log p}{n}}\approx 0.33$
		(right panels) or at a value approximating the behavior of 10-fold
		cross validation (0.19 and 0.14 respectively for the low and high
		signal cases).}
	
		\label{fig:partial_sn10}
	
\end{figure}

\begin{figure}[hbtp]
	\begin{centering}
		\includegraphics[width=0.33\paperwidth]{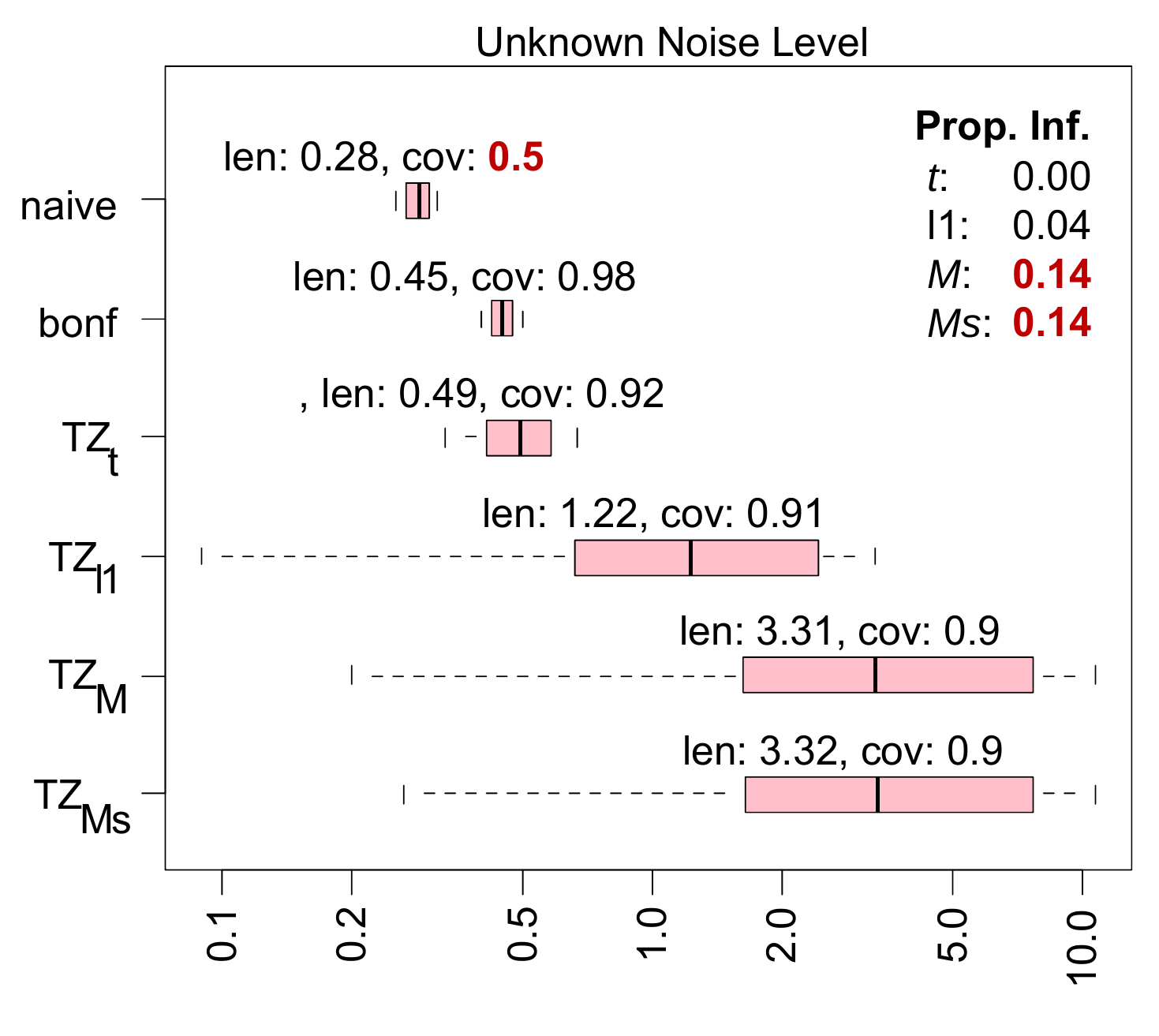}
		\par\end{centering}

	\caption{\em Boxplot of lengths of 90\% confidence intervals when $\sigma^{2}$
		is unknown and a plug-in estimate $\hat{\sigma}^{2}$ is used ($\hat{\sigma}^{2}$
		is the unbiased variance estimate for the linear model selected by
		the lasso). Reported are the median interval length, the empirical
		coverage, and the proportion of ``infinite'' intervals (the infinite
		length results from numerical inaccuracies when inverting a truncated
		normal CDF); the boxplots set infinite lengths to the maximum finite
		observed length.  Data have dimensions $n=100,p=250$. First five components
		of the true regression vector $\beta$ are non-zero and set at 0.29. $\lambda$
		was set to 0.19 to approximate what cross-validation would have chosen and $\lambda_{\text{high}}=\sqrt{\frac{2\log p}{n}}\approx0.33$.}
		\label{fig:unknown_variance}
	
\end{figure}

\begin{figure}[hbtp]
	\begin{centering}
		\includegraphics[width=0.33\paperwidth]{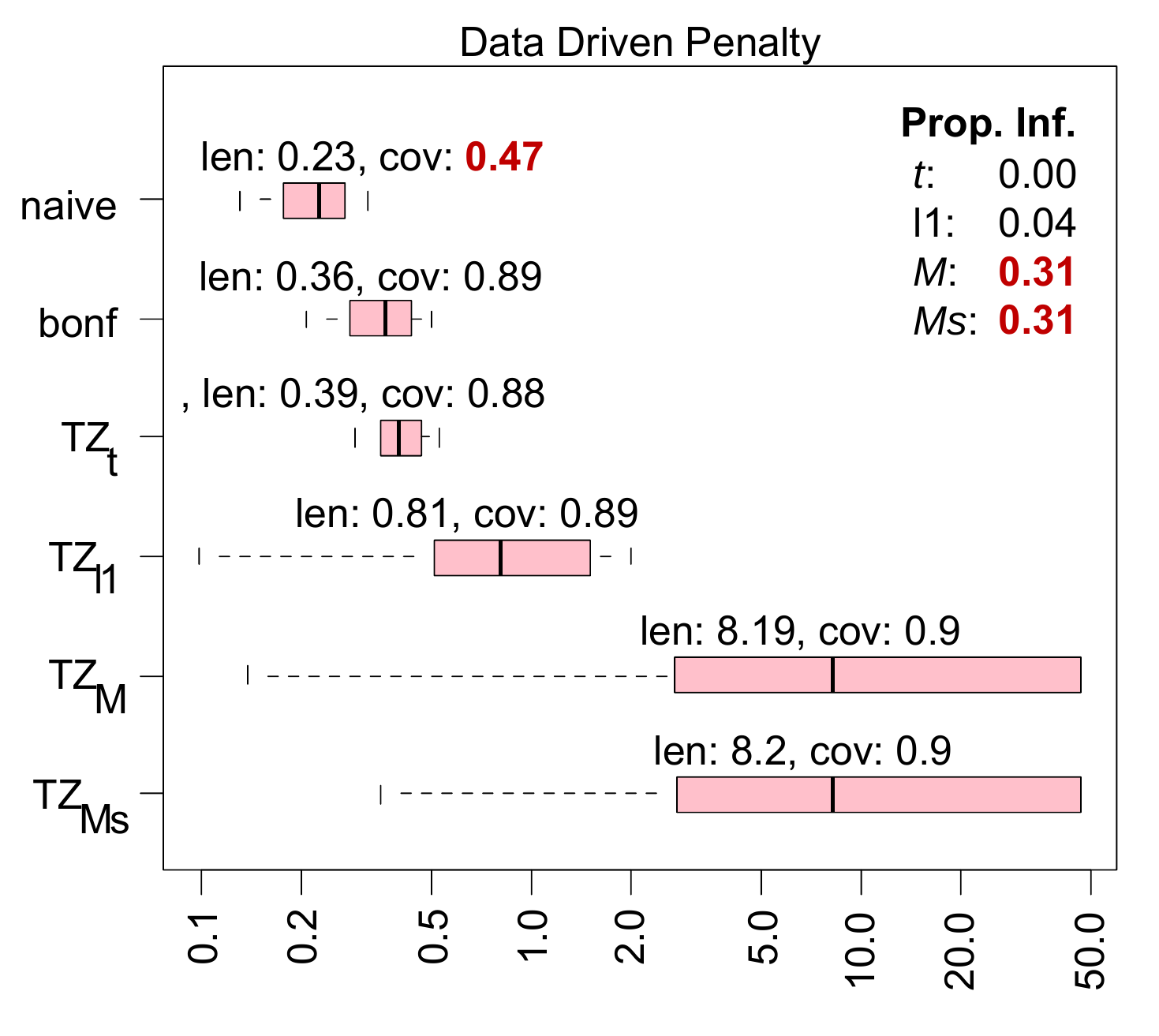}
		\par\end{centering}

	\caption{\em Boxplot of lengths of 90\% confidence intervals when $\lambda$
		is chosen by 10-fold cross-validation. Reported are the median interval length, the empirical
		coverage, and the proportion of ``infinite'' intervals (the infinite
		length results from numerical inaccuracies when inverting a truncated
		normal CDF); the boxplots set infinite lengths to the maximum finite
		observed length.  Data has dimensions $n=100,p=250$.
		First five components of the true regression vector $\beta$ are non-zero and
		set at 0.29. $\lambda_{\text{high}}=\sqrt{\frac{2\log p}{n}}\approx0.33$.}
		\label{fig:data_driven_lam}
	
\end{figure}

\newpage
\bibliographystyle{agsm}
\bibliography{tibs}
 \end{document}